%% file: main.tex
\documentclass[aps, prb, superscriptaddress, showpacs,floatfix,10pt,twocolumn]{revtex4-2}

\usepackage{graphicx}
\usepackage{xcolor}
\usepackage{transparent}
\usepackage{import}
\usepackage{overpic}
\usepackage{amssymb}
\usepackage{mathtools}
\usepackage{bbm}
\usepackage{bm}
\usepackage{mathrsfs}
\usepackage{verbatim}
\usepackage{dsfont}
\usepackage{calligra}
\usepackage{amsthm}
\usepackage[breaklinks=true,colorlinks,citecolor=blue,linkcolor=blue,urlcolor=blue]{hyperref}

\newtheorem{theorem}{Theorem}[section]
\newtheorem{corollary}{Corollary}[section]
\newtheorem{lemma}{Lemma}[section]
\newtheorem{defn}{Definition}[section]

\DeclareMathOperator*{\sumint}{%
\mathrlap{\displaystyle\int}
\mathrlap{\textstyle\sum}
\phantom{\mathrlap{\displaystyle
\int}\textstyle\sum}
}

\makeatletter
\DeclareRobustCommand{\cev}[1]{%
  {\mathpalette\do@cev{#1}}%
}
\newcommand{\do@cev}[2]{%
  \vbox{\offinterlineskip
    \sbox\z@{$\m@th#1 x$}%
    \ialign{##\cr
      \hidewidth\reflectbox{$\m@th#1\vec{}\mkern4mu$}\hidewidth\cr
      \noalign{\kern-\ht\z@}
      $\m@th#1#2$\cr
    }%
  }%
}
\makeatother

\newcommand{\tr}{\operatorname{Tr}}
\newcommand{\hc}{{\rm H.c.}}
\newcommand{\1}{\mathds{1}}
\renewcommand{\Re}{\operatorname{Re}}
\renewcommand{\Im}{\operatorname{Im}}

\renewcommand{\a}[1]{a^{\vphantom{\dagger}}_{#1}}
\newcommand{\ad}[1]{a^{\dagger}_{#1}}

\renewcommand{\c}[1]{c^{\vphantom{\dagger}}_{#1}}
\newcommand{\cd}[1]{c^{\dagger}_{#1}}
\renewcommand{\H}[1]{{\cal H}_{#1}}
\newcommand{\E}{{\cal E}}
\newcommand{\G}{{\cal G}}
\newcommand{\C}{{\cal C}}
\newcommand{\M}{{\cal M}}
\newcommand{\N}{{\cal N}}

\renewcommand{\AA}{{\mathcal A}}
\newcommand{\V}{{\cal V}}
\newcommand{\U}{{\cal U}}

\newcommand{\I}{{\cal I}}
\newcommand{\K}{{\cal K}}

\newcommand{\T}{{\cal T}}
\renewcommand{\O}{{\cal O}}

\newcommand{\g}{{\mathfrak g}}
\newcommand{\cc}{{\mathfrak c}}
\renewcommand{\d}{{\rm d}}
\newcommand{\ev}[1]{\bigl \langle #1 \bigr \rangle}
\newcommand{\vev}[1]{\bigl \langle #1 \bigr \rangle_0}
\newcommand{\tord}{{\mathbb{T}}}
\newcommand{\drot}{{\tilde{\mathbb{T}}}}
\newcommand{\effconj}{{}^{\bigstar}\!}

\begin{document}

\title{Exact quantum transport in non-Markovian open Gaussian systems}

\date{\today}

\author{Guglielmo Pellitteri}\email{guglielmo.pellitteri@sns.it}
\affiliation{Scuola Normale Superiore, 56126 Pisa, Italy}

\author{Vittorio Giovannetti}\email{vittorio.giovannetti@sns.it}
\affiliation{Scuola Normale Superiore, 56126 Pisa, Italy}
\affiliation{NEST and Istituto Nanoscienze-CNR, 56126 Pisa, Italy}

\author{Vasco Cavina}\email{vasco.cavina@sns.it}
\affiliation{Scuola Normale Superiore, 56126 Pisa, Italy} 

\begin{abstract}

We build an exact framework to evaluate heat, energy, and particle transport between Gaussian reservoirs mediated by a quadratic quantum system. By combining full counting statistics with newly developed non-Markovian master equation approaches, we introduce an effective master equation whose solution can be used to generate arbitrary moments of the heat statistics for any number of reservoirs. This theory applies equally to fermionic and bosonic systems, holds at arbitrarily strong coupling, and resolves out-of-equilibrium transient dynamics determined by the system’s initial state. In the steady-state, weak-coupling limit, we recover results analogous to those of the well-known Landauer--B\"{u}ttiker formalism. We conclude our discussion by demonstrating an application of the method to a prototypical fermionic system. Our results uncover a regime of transient negative heat conductance contingent upon the initial system preparation, providing a clear signature of non-trivial out-of-equilibrium dynamics.

\end{abstract}

\maketitle

\section{Introduction}
\label{sec:intro}
Understanding heat flows in quantum and nanoscale systems is a central problem in quantum thermodynamics and quantum technologies. It is essential for the design and optimization of quantum thermal machines and quantum engines~\cite{kosloff2014quantum, arrachea2007heat, erdman2017thermoelectric}, and it plays a key role in quantum transport, where energy currents, dissipation, and fluctuations govern the performance of nanoscale devices~\cite{benenti2017fundamental, dubi2011colloquium, esposito2015efficiency, Donarini2024transport}. Moreover, a reliable description of heat flow and heat noise is crucial for quantum information processing, as thermal fluctuations contribute to decoherence and limit the fidelity of quantum computation~\cite{paladino20141}.

In all these contexts, the study of the energetics has been extensively tackled using methods based on Markovian master equations~\cite{gorini1976classic, lindblad1976classic, Espositoreview, soret2022thermodynamic} and weak coupling, steady state approaches, like the celebrated Landauer--B\"uttiker formalism~\cite{landauer1957spatial, buttiker1986four}.
However, the study of non-Markovian dynamics and strong-coupling effects has recently become a critical research direction, driven by the need to meet emerging technological demands~\cite{ vacchini2011markovianity, tamascelli2018nonperturbative, deVega2017review, Breuer2016Review, Li2020Perspective, ankerhold2026colloquium}. In particular, the quest for scalability in quantum computation—requiring dense, low-temperature architectures—renders quantum devices increasingly susceptible to correlated noise~\cite{paladino2002decoherence, gulacsi2023signatures, zou2024spatially}, making such effects especially relevant. Harnessing them lies at the core of modern error-mitigation schemes~\cite{wang2025non, sannia2025non}. 
Similarly, in quantum thermodynamics, correlations and memory effects can be exploited to improve the performance of standard tasks, including work extraction and refrigeration~\cite{abiuso2019non, strasberg2016nonequilibrium, perarnau2015extractable, thomas2018thermodynamics, talkner2020colloquium}. Accordingly, a realistic characterization of memory and correlations within open quantum system theory is a necessary step toward realizing these ideas.
Consequently, several master equation-based approaches have been proposed for systems strongly coupled to their environment, ranging from stochastic methods for Gaussian baths~\cite{stockburger2002exact, stockburger2001diffusion, tilloy2017unraveling, diosi1998diffusion, cavina2025unifying} and group-theoretical methods~\cite{FeiQuan2019} to exact non-Markovian  master equations~\cite{Magazzu2024unified,colla2025unveiling, luczka1990spin, Hu1992Classic, Zhang2012ME, d2025exact} and influence functional approaches~\cite{Magazzu2022influence}.
In this article, we take on this challenge and establish a new method for computing energy statistics in open quantum systems strongly coupled to their environments. The framework developed here assumes Gaussian system and environment degrees of freedom with identical (bosonic or fermionic) statistics, yielding a fully exact theory valid at arbitrarily strong coupling. Crucially, this enables the resolution of out-of-equilibrium transient dynamics dependent on the initial preparation, thereby relaxing the weak-coupling and steady-state constraints inherent to standard scattering matrix approaches~\cite{NazarovBlanter2009, Datta1995, Imry2002, FisherLee1981, blasi2024exact}.
\begin{figure}[t]
  \centering
  \def\svgwidth{1.0\linewidth}
  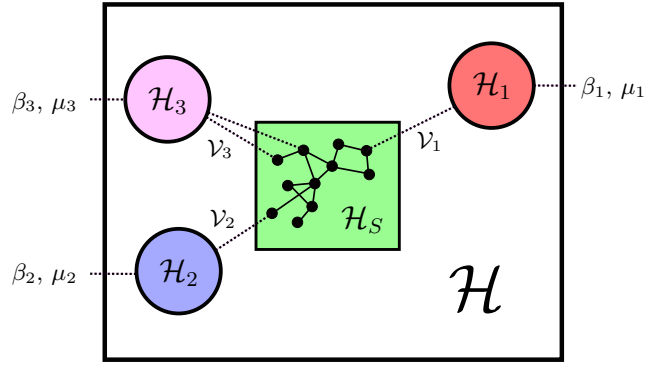
  \caption{Graphical representation of a quantum system described by a quadratic Hamiltonian $\H{S}$ interacting with arbitrary modes of multiple reservoirs $\H{\alpha}$ via potentials $\V_{\alpha}$. Each reservoir is initialized at thermal and chemical equilibrium at inverse temperature and chemical potential $\beta_\alpha,\, \mu_\alpha$. The unitary evolution of the system-environment compound is generated by the total Hamiltonian $\H{}$.} 
  \label{fig:1}
\end{figure}

The article is organized as follows. In Sec.~\ref{sec:II}, we introduce the theoretical framework, giving an overview on the Keldysh formalism and on exact Gaussian master equations (GMEs). Sec.~\ref{sec:III} is devoted to the description of the characteristic function (CF) formalism employed to study the quantum thermodynamics of the system-environment compound. In Sec.~\ref{sec:IV}, we present the tilted Gaussian master equation (tGME), a novel tool which governs the exact time evolution of the CF, thereby encoding the full information on the system-environment heat statistics. Moreover, we sketch the derivation of an exact equation for the heat transport within the system-environment compound, which is our main result. In Sec.~\ref{sec:V} we derive the weak-coupling limit of the latter equation, and we prove the reduction of our formalism to the Landauer--B\"uttiker scattering theory in the steady-state limit for the simple case of a single fermionic level. In Sec.~\ref{sec:VI}, we analyze a simple case study, showcasing an example of non-trivial transient dynamics in the form of negative heat conductance, and comparing exact results with weak-coupling calculations. Finally, in Sec.~\ref{sec:VII} we draw our conclusions.
\section{Exact non-Markovian quantum dynamics}
\label{sec:II}
In this section we describe the framework employed in this article. We first introduce open Gaussian systems in Sec.~\ref{subsec:IIa}. Then, we recap the formalism of the exact Gaussian master equation (GME) introduced in Ref.~\cite{d2025exact} (Sec.~\ref{subsec:IIb}), which we will later modify to describe heat, energy and particle transport within the system-environment compound.
\subsection{Open quantum systems in the Keldysh contour}
\label{subsec:IIa} 
We consider a quantum system coupled to $R$ reservoirs, that will be collectively referred to as the environment. The total Hamiltonian can be written as
\begin{equation}
\label{eq:total-hamiltonian}
    \H{} = \H{0} + \V = \H{S} + \H{E} + \V~,
\end{equation}
where $\H{E} = \sum_{\alpha=1}^{R} \H{\alpha}$, with $\H{\alpha}$ being the Hamiltonian of the reservoir $\alpha$ and $\V$ the system-environment coupling (as schematized in Fig.~\ref{fig:1}). The coupling is taken to be of the general form
\begin{equation}
\label{eq:interaction-hamiltonian}
    \V = \sum_\mu A_\mu \otimes B_\mu~,
\end{equation}
where $\{A_\mu\}_\mu$ is a set of system operators and $\{B_\mu\}_\mu$ are environment operators.  

The starting point of our discussion is the dynamics of $S$ and $E$. Given an initial compound state $\rho_{SE}(0)$, the interaction picture system state at subsequent time $t\geq 0$ reads
\begin{equation} \label{eq:evolv}
    \rho_{S}(t) =  \tr_E \bigl[\U(t, 0) \rho_{SE}(0)\, \U^\dagger(t, 0)\bigr]~,
\end{equation}
where $\U(t, 0)$ is the physical-time, interaction picture unitary evolution operator, i.e. the time-ordered exponential
\begin{equation} \label{eq:timeord}
    \U(t, 0) = \tord \exp \left[-i\int_0^t \d \tau\,\V(\tau)\right]~,
\end{equation}
where $\V(t) = \U_0^\dagger(t, 0) \V \,\U_0(t, 0) $ and $\U_0(t, 0)$ is the free evolution operator generated by $\H{0}$. Throughout this work, we adopt natural units in which $\hbar = 1$.

Let us now introduce an instrument that will be fundamental in the derivation of our main results: the Keldysh contour $\gamma(t)$.
For this purpose, it is convenient to think to $ \U(t, 0) $ and $\U^\dagger(t, 0)$ in Eq. \eqref{eq:evolv} as two ordered exponentials on two distinct domains $\gamma_-(t)$ and $\gamma_+(t)$, respectively. These two domains, called the forward $(-)$ and backward $(+)$ branches, are respectively ordered from $0$ to $t$ and from $t$ to $0$. The Keldysh contour $\gamma(t)$ is obtained by merging the branches in such a way that the initial time is halfway through the domain 
(see Fig.~\ref{fig:2}). 
The adoption of such a contour and an extension of the usual ordering operator $\tord$ to it ($\tord$ will order operators with earlier arguments in the contour to the right) allows us to express the dynamics in the following compact way
\begin{equation}
\notag \label{eq:state-keldysh-ord}
\begin{split}
    \rho_S(t) & =  \tr_E\tord \left\{\exp\left[ -i\int_{\gamma(t)} \d z\,\V(z) \right] \rho_{SE}^{\vphantom{\dagger}}(0)\right\}~,
\end{split}
\end{equation}
In the following, we will denote with a $+$ ($-$) subscript a complex variable belonging to the $\gamma_+$ ($\gamma_-$) branch (e.g. $z_{\pm}$ for $z \in \gamma_\pm(t)$).
For more details on the Keldysh formalism and on the Keldysh approach to open quantum systems and quantum thermodynamics, we refer to the extensive literature on the topic~\cite{stefanucci2013book, kamenev2011book, funo2018pathint, cavina2023convenient}. 
\subsection{Exact Gaussian master equation}
\label{subsec:IIb}
The Gaussian hypothesis amounts to assuming that both the system and the reservoirs consist of non-interacting particles obeying the same quantum statistics. Accordingly, their Hamiltonians are taken to be quadratic in suitably defined sets of ladder operators.
In a similar fashion, the coupling Hamiltonian $\V$ is chosen to be bilinear in the ladder operators of $S$ and $E$.
For simplicity, we take the initial state $\rho_{SE}(0)$ to be separable across all system and reservoir degrees of freedom, i.e.
\begin{equation} \label{eq:initial-state}
    \rho_{SE}(0) = \rho_S(0) \bigotimes_{\alpha = 1}^{R} \rho_\alpha(0)~. 
\end{equation}
In the interest of completeness, we remark that the calculations that follow can be generalized to the case of a correlated state under the Gaussian hypothesis, as shown in Ref.~\cite{d2025exact}. 
Within this Gaussian framework, physical intuition suggests that the bare environment Green's function (GF) governs the system dynamics. Such a GF is written on the Keldysh contour as
\begin{equation}
\label{eq:bare-GF}
    \C_{\mu\nu}(z, w) : = \tr\tord_\zeta\bigl[B_\mu(z) B_\nu(w) \,\omega_E\bigr]~,
\end{equation}
where $\omega_E := \bigotimes_{\alpha = 1}^R \rho_\alpha(0)$ is the initial state of the environment and
we introduced the modified contour ordering $\tord_\zeta$: for $\zeta = +1$ (bosons) it denotes the standard contour ordering $\tord$, while for $\zeta = -1$ (fermions) we account for each operator switch performed in order to attain the correct ordering with a minus sign. %
It has, indeed, recently been shown~\cite{d2025exact} that we can write a closed, exact equation for the system dynamics, the Gaussian master equation (GME) 
\begin{equation}
    \label{eq:gme}
    \frac{\d \rho_S(t)}{\d t} = \int_{0}^t \d \tau \, \G_{\mu\nu}(t_+, \tau) \bigl[{A_\nu(\tau) \rho_S(t)}, A_\mu(t)\bigr] + \hc~.
\end{equation}
Here and in the following, Einstein summation convention on $\mu, \nu$ is intended, and $\G$ is a dressed GF, with arguments in $\gamma(t)$, solving the following equation:
\begin{equation}
\label{eq:bare-dyson}
     \G(12)  =  \C(12)  + \iint_{\gamma(t)} \d (34)\, \G(13)\,\Sigma(34) \,  \C(42)~,
\end{equation}
where $\Sigma$ is a matrix built from the system operators [see Eq.~\eqref{sm-eq:self-energy} in App.~\ref{app:B}], and we employed the standard shorthand $z_i \to i$ for time variables (e.g. $\C(12) := \C(z_1, z_2)$). Eq.~\eqref{eq:bare-dyson} is formally identical to the Dyson equations commonly employed in many-body quantum theory~\cite{FetterWalecka, GiulianiVignale}. In analogy with the latter, $\Sigma$ plays the role of a self-energy, encoding the memory of repeated system–environment interactions and thereby accounting for interaction vertices to all orders in the coupling. At lowest non-trivial order in the coupling, i.e. $\G \to \C$, we recover the standard “memoryless" Redfield equation~\cite{d2025exact, breuer2007book, deVega2017review}.
\begin{figure}[t]
  \centering
  \def\svgwidth{1.0\linewidth}
  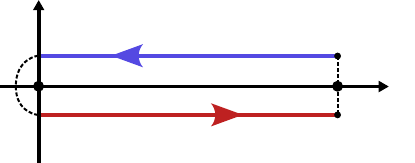
  \caption{Graphical representation of the Keldysh contour on the complex time plane.} 
  \label{fig:2}
\end{figure}
\section{From dynamics to thermodynamics: the characteristic function}
\label{sec:III}

In order to describe heat exchanges within the system-environment compound, we adopt the {\it two-point energy measurement (TPEM)} technique~\cite{Espositoreview, talkner2016aspects}, in which the relevant reservoirs observables are measured at consecutive times $t_0$ and $t_1$, and the difference of the outcomes of the two measurements is used to quantify their energy variations. The reservoirs Hamiltonians can be measured simultaneously since $[\H{\alpha}, \H{\alpha'}] = 0$ for all $\alpha,\alpha' = 1, \dots, R$. Moreover, we assume each reservoir Hamiltonian to be number conserving, i.e. $[\H{\alpha}, \N_{\alpha}] = 0$, $\N_\alpha$ being the number operator associated with the $\alpha$-th reservoir. Fixing $t_0=0$ and $t_1 = t \geq 0$, we define the energy and particle number average variations as
\begin{align}
    \ev{\Delta \E_\alpha} (t) &: = \tr \Bigl[\H{\alpha} \bigl[\rho_\alpha(t) - \rho_\alpha(0)\bigr]\Bigr]~,\\
    \ev{\Delta N_\alpha} (t) &: = \tr \Bigl[\N_{\alpha} \bigl[\rho_\alpha(t) - \rho_\alpha(0)\bigr]\Bigr]~.
\end{align}
The heat exchange follows directly from the first law of thermodynamics, which reads
\begin{equation}
\label{eq:heat-exchange}
    \begin{split}
        \ev{Q_\alpha} &= \ev{\Delta \E_\alpha} - \mu_\alpha \ev{\Delta N_\alpha}\\
        &  =\tr \Bigl[\Omega_{\alpha} \bigl[\rho_\alpha(t) - \rho_\alpha(0)\bigr]\Bigr]
    \end{split}
\end{equation}
where $\mu_\alpha$ is the chemical potential associated with the $\alpha$-th reservoir, and $\Omega_\alpha = \H{\alpha} -\mu_\alpha \N_\alpha$ is the corresponding grand-canonical Hamiltonian.
Given a collective initial state $\rho_{SE}(0)$ satisfying $[\rho_{SE}(0), \Omega_\alpha] = 0$ $\forall \alpha$, we can build a CF for the heat exchanged with each reservoir~\cite{Aurell_2020} as follows:
\begin{equation}
\label{eq:CF}
    \begin{split}
        \M(t; \bm \lambda) :=  \tr\left[ \tilde \U_{\bm \lambda}^{\vphantom{\dagger}} \,(t, 0) \rho_{SE}^{\vphantom{\dagger}}(0)\, \tilde \U_{-\bm\lambda}^\dagger(t, 0) \right]~,
    \end{split}
\end{equation}
where we introduced the tilted evolution operator
\begin{equation}
    \tilde \U_{\bm \lambda}^{\vphantom{\dagger}}(t, 0) := e^{-i \sum_\alpha\Omega_{\alpha}\lambda_\alpha/2} \,\U(t, 0) \,e^{i \sum_\alpha \Omega_{\alpha}\lambda_\alpha/2}~.
\end{equation}
Here, $\bm \lambda \in \mathbb R^{R}$ is a vector of tilting parameters. Such a generating function allows us to immediately obtain the heat exchange~\eqref{eq:heat-exchange} as
\begin{equation}
\label{eq:heat-from-CF}
    \ev{Q_\alpha}(t) = i\frac{\partial  \M(t; \bm \lambda)}{\partial \lambda_\alpha} \bigg|_{\bm \lambda = \bm 0}~.
\end{equation}
Exploiting the Keldysh contour in Fig.~\ref{fig:2}, we can write the generating function as trace of a single ordered exponential:
\begin{equation}
\label{eq:CF-exp}
\begin{split}
    &\M(t; \bm \lambda) = \\& \tr\tord \left\{\exp\left[ -i\int_{\gamma(t)} \d z\,\V_{\bm \lambda}(z) \right] \rho_{SE}^{\vphantom{\dagger}}(0)\right\}~,
\end{split}
\end{equation}
where $z \in \mathbb C$ and we defined
\begin{equation}
\label{eq:tilted-interaction}
    \V_{\bm \lambda}(z_\pm) : = u_{\mp \bm \lambda}^{\vphantom{\dagger}}\V(z)\, u_{\mp\bm\lambda}^\dagger 
\end{equation}
where $u_{\bm \lambda} := e^{-i\sum_\alpha\Omega_\alpha \lambda_\alpha/2}$.
In Eq.~\eqref{eq:CF-exp}, the time-ordering is inteded to act as a contour ordering, putting to the left operators that appear “later" in the sense of the Keldysh contour in Fig.~\ref{fig:2}. Both Eq.~\eqref{eq:heat-from-CF} and Eq.~\eqref{eq:CF-exp} are proven in App.~\ref{app:A}.

In order to invoke the machinery of exact master equations for our purposes, we factorize the total trace in Eq.~\eqref{eq:CF-exp} as $\tr = \tr_S \tr_E$, obtaining 
\begin{equation}
\label{eq:CF-trace}
    \M(t; \bm \lambda) : = \tr_S \tilde \rho_{S}(t; \bm\lambda)~.
\end{equation}
Here, $\tilde \rho_{S}(t; \bm\lambda)$ is a tilted density operator given by
\begin{equation}
\label{eq:CF-contour}
\begin{split}
    &\tilde \rho_{S}(t; \bm\lambda) :=\\ & \tr_E\tord \biggl\{\exp\left[ -i\int_{\gamma(t)} \d z\,\V_{\bm \lambda}(z) \right] \rho_{SE}^{\vphantom{\dagger}} (0)\biggr\} ~,
\end{split}
\end{equation}
and satisfying $\tilde \rho_{S}(t; \bm 0) = \rho_S(t)$. We remark that $\tilde \rho_{S}(t; \bm\lambda)$ is not a state, for it is not necessarily positive semidefinite nor normalized. The goal is now to write an exact master equation for $\tilde \rho_{S}(t; \bm\lambda)$, in analogy to the GME~\eqref{eq:gme} for $\rho_S(t)$, and therefore to obtain the CF as its trace on the system's degrees of freedom.

Since the environment grand-canonical Hamiltonians $\{\Omega_\alpha\}_\alpha$ commute with the system operators, we can write the tilted interaction Hamiltonian~\eqref{eq:tilted-interaction} as
\begin{equation}
\label{eq:tilted-interaction-sum}
    \V_{\bm \lambda}(z) = \sum_\mu A_\mu(z) \otimes \tilde B_\mu(z; \bm \lambda)~,
\end{equation}
where only the environment operators $\{\tilde B_\mu(z; \bm \lambda)\}_\mu$ with 
\begin{equation}
\label{eq:tilted-operators}
    \tilde B_\mu(z_\pm;\bm \lambda) : = u_{\mp \bm \lambda}^{\phantom{\dagger}} B_{\mu}(z) u_{\mp \bm \lambda}^\dagger 
\end{equation}
depend on the contour branch. The other main ingredient will be a {\it tilted} GF, that, in analogy with Eq.~\eqref{eq:bare-GF}, is defined as
\begin{equation}
\label{eq:tilted-GF}
    \tilde\C_{\mu\nu}(z, w; \bm \lambda) : = \tr\tord_\zeta\bigl[\tilde B_\mu(z; \bm \lambda) \tilde B_\nu(w; \bm \lambda) \omega_E\bigr]~.
\end{equation}
Evaluated in $\bm \lambda = \bm 0$, the latter is just the standard contour-ordered environment GF~\eqref{eq:bare-GF}. 

By hermiticity of the interaction Hamiltonian~\eqref{eq:interaction-hamiltonian}, we can write the set of environment operators $\{B_\mu\}_\mu$ as a collection of pairs $\{(B_\alpha^\dagger, B_\alpha^{\vphantom{\dagger}})\}_{\alpha=1}^R$, each pair corresponding to a different reservoir. Without loss of generality, we assume $B_\alpha$ to be linear in the $\alpha$-th reservoir annihilation operators. It is thus immediate to see that, for the case of $R$ uncorrelated reservoirs, it is always possible to order the reservoir operators in such a way that the matrix $\tilde\C$ will be of the following block form in the $\mu,\nu$ indices:
\begin{equation}
\label{eq:tilted-GF-blocks}
   \tilde \C = \operatorname{diag} \bigl(\tilde \C_1, \tilde\C_2, \dots, \tilde  \C_R\bigr)~,
\end{equation}
where each $\tilde \C_\alpha$ is a $2\times 2$ block associated to the $\alpha$-th reservoir, containing the four possible correlators between the operators $\tilde{B}_{\alpha}^{\dag}$, $\tilde{B}_{\alpha}$. 
As proven in App.~\ref{app:A}, the tilted GF can be straightforwardly connected to the physical-time GF
\begin{equation}
\label{eq:chi}
    \begin{split}
        \chi_{\mu\nu}(s)
        & = \tr \bigl[B_\mu(s) B_\nu(0) \omega_E\bigr]\\
        & = \operatorname{diag} \bigl(\chi_1, \chi_2, \dots, \chi_R \bigr)_{\mu\nu}~,
    \end{split}
\end{equation}
just by fixing the arguments $z$ and $w$ on different branches of the Keldysh contour. We choose $z = z_{\pm}$ and $w = w_{\pm}$ (we recall that $z_{\pm}$ and $w_{\pm}$ take values in the interval $[0,t]$ and are assigned either to the forward branch $\gamma_+$ or to the backward branch $\gamma_-$) and obtain
\begin{align}
\label{eq:C-tord}
    \tilde \C^{\,\tord}_\alpha(z, w) & := \tilde \C_\alpha(z_-, w_-) \\&= \theta(z-w) \chi_\alpha(z-w) 
    + \zeta \theta(w-z) \chi_\alpha^T(w-z)\notag~,\\
\label{eq:C-drot}
    \tilde \C^{\,\drot}_\alpha(z, w) & := \tilde \C_\alpha(z_+, w_+)\\&= \theta(w-z) \chi_\alpha(z-w) 
    + \zeta \theta(z-w) \chi_\alpha^T(w-z)\notag~,
\end{align}
where $^T$ denotes matrix transposition, and 
\begin{align}
\label{eq:C->}
    \tilde \C^>_\alpha (z, w; \bm \lambda)  & := \tilde \C_\alpha(z_+, w_-)\\& = e^{-i\mu_\alpha\lambda_\alpha \sigma_z} \,\chi_\alpha(z-w+\lambda_\alpha)\notag~,\\
\label{eq:C-<}
    \tilde \C^<_\alpha (z, w; \bm \lambda) & := \tilde \C_\alpha(z_-, w_+)\\& = \zeta e^{i\mu_\alpha\lambda_\alpha \sigma_z} \,\chi_\alpha^T(w-z+\lambda_\alpha)\notag~.
\end{align}
In the literature on the Keldysh formalism, the latter functions are commonly referred to, respectively, as the time-ordered, anti-time-ordered, greater and lesser components of the GF~\cite{stefanucci2013book, kamenev2011book}. We would like to emphasize that the tilting does {\it not} affect the $\tord$ and $\drot$ components of the GF, as it is evident in Eqs.~\eqref{eq:C-tord}-\eqref{eq:C-drot}. Indeed, we have $\tilde\C^{\,\tord} \equiv \C^{\tord}$, $\tilde\C^{\,\drot} \equiv \C^{\drot}$. The tilting only acts on the greater and lesser components, as in Eqs.~\eqref{eq:C->} and~\eqref{eq:C-<}, which also give us a recipe to immediately evaluate them from the physical-time GF $\chi$. We can now proceed to derive an exact equation for the tilted density operator $\tilde \rho_S(t; \bm \lambda)$, which will in turn give us access to a simple formula for the heat exchange between reservoirs.
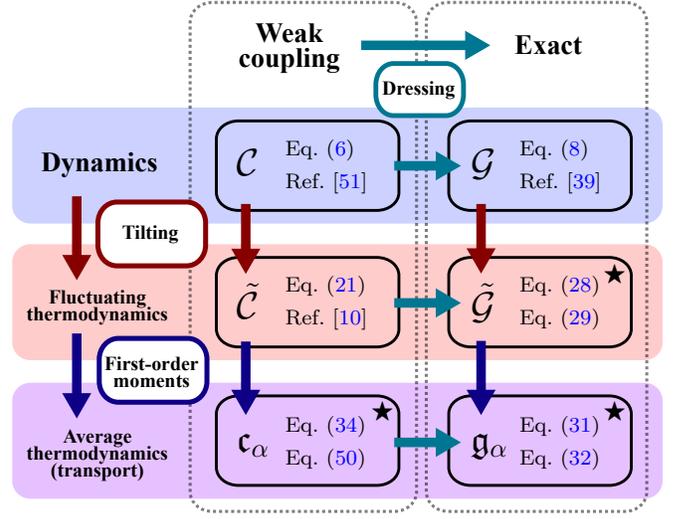
\begin{figure}
    \centering
    \begin{overpic}[width=1.\linewidth]{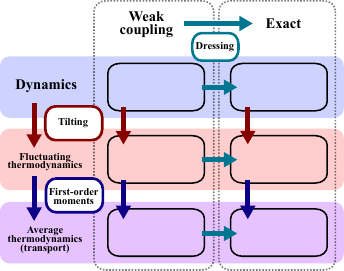}
        %
        \put(34.5, 52){\Large $\C$}
        \put(42, 55){\footnotesize Eq.~\eqref{eq:bare-GF}}
        \put(42, 50){\footnotesize Ref.~\cite{breuer2007book}}
        \put(70.5, 52){\Large $\G$}
        \put(78, 55){\footnotesize Eq.~\eqref{eq:bare-dyson}}
        \put(78, 50){\footnotesize Ref.~\cite{d2025exact}}
        \put(34.5, 30){\Large $\tilde\C$}
        \put(42, 34){\footnotesize Eq.~\eqref{eq:tilted-GF}}
        \put(42, 29){\footnotesize Ref.~\cite{Espositoreview}}
        \put(70.5, 30){\Large $\tilde\G$}
        \put(78, 34.){\footnotesize Eq.~\eqref{eq:tGME}}
        \put(78, 29.){\footnotesize Eq.~\eqref{eq:tilted-dyson}}
        \put(90.9, 35.7){$\bigstar$}
        \put(34.5, 9.5){\Large $\cc_\alpha$}
        \put(42, 12.5){\footnotesize Eq.~\eqref{eq:goth-c}}
        \put(42, 7.5){\footnotesize Eq.~\eqref{eq:wc-transport-equation}}
        \put(55.2, 14.7){$\bigstar$}
         \put(70.5, 9.5){\Large $\g_\alpha$}
        \put(78, 12.5){\footnotesize Eq.~\eqref{eq:transport-equation}}
        \put(78, 7.5){\footnotesize Eq.~\eqref{eq:goth-g}}
        \put(90.9, 14.7){$\bigstar$}
    \end{overpic}
    \caption{A summary of the results presented in Secs.~\ref{sec:II} to~\ref{sec:IV}. Star-marked boxes contain equations introduced in the present work. The equations valid in the weak-coupling regime (left) are set apart from the exact results (right), which are valid for arbitrary system-environment coupling strength. The exact dynamics of the system is captured by the GME introduced in Ref.~\cite{d2025exact}, characterized by the dressed GF $\G$, that can be obtained by “dressing" the contour GF [Eq.~\eqref{eq:bare-GF}] via the Dyson equation~\eqref{eq:bare-dyson}. In the weak-coupling limit, we recover the standard “memoryless" Redfield equation~\cite{breuer2007book, deVega2017review}. The tilting procedure, which maps $\C \to \tilde \C$ and $\G \to \tilde \G$, allows us to study the fluctuating thermodynamics of the system-environment compound via a similar machinery~\cite{Espositoreview}. Indeed, the CF~\eqref{eq:CF}, which time evolution is described by $\tilde \G$ via the tGME~\eqref{eq:tGME}, contains all the information on the heat statistics of the compound. Finally, by evaluating the first-order moments of such a distribution, we obtain exact transport equations governed by heat kernels $\g_\alpha$, which are in turn obtained from the “memoryless" kernels $\cc_\alpha$ by dressing them according to the Dyson form~\eqref{eq:goth-g-dyson-contour}.} 
    \label{fig:3}
\end{figure}
\section{Exact transport equations}
\label{sec:IV}
In this section we present the main results of this article. In Sec.~\ref{subsec:IVa}, we discuss the tilted version of the GME~\eqref{eq:gme} and sketch the derivation of an exact equation for the heat exchange within the system-environment compound. In Sec.~\ref{subsec:IVb}, we analyze the kernel which defines such an equation, and we show how we can immediately derive analogous forms for the energy and particle exchange.
\subsection{An exact equation for the heat current}
\label{subsec:IVa}
The derivation of the tilted Gaussian master equation (tGME), i.e. the equation governing the evolution of $\tilde\rho_S(t; \bm \lambda)$, follows similar steps to the non-tilted case~\cite{d2025exact} and is carried out in App.~\ref{app:B}. The tGME reads as follows:
\begin{equation}
\label{eq:tGME}
    \begin{split}
        \frac{\d \tilde\rho_S(t; \bm \lambda)}{\d t} 
         = \int_{\gamma(t)} \d z\, &\Bigl\{
            \tilde\G_{\mu \nu}(t_+, z; \bm \lambda) \overline{ A_\nu(z) \tilde\rho_S(t; \bm \lambda) } A_\mu(t) \\ & \!\!\!
            - \tilde\G_{\mu \nu}(t_-, z; \bm \lambda) A_\mu(t) \overline{A_\nu(z) \tilde\rho_S(t; \bm \lambda)}\Bigr\} ~,
    \end{split}
\end{equation}
where $\overline{ A(z)\,\bullet}$ is  $A(z) \,\bullet$ if $z \in\gamma_-(t)$ and $ \zeta \bullet A(z) $ if $z\in\gamma_+(t)$. The tilted, dressed GF $\tilde \G$ is defined by the Dyson form
\begin{equation}
\label{eq:tilted-dyson}
    \tilde \G(12)  = \tilde \C(12)  + \iint_{\gamma(t)} \d (34)\,\tilde \G(13)\,\Sigma(34) \, \tilde \C(42)~,
\end{equation}
which has a physical interpretation analogous to that of Eq.~\eqref{eq:bare-dyson}. As a consistency check, we note that in the $\bm \lambda \to \bm 0$ limit the tGME reduces to the standard GME~\eqref{eq:gme}, and we have $\tilde \G(z, w; \bm \lambda \to \bm 0) \equiv \G(z, w)$, $\G$ being the dressed GF that solves Eq.~\eqref{eq:bare-dyson}.

It is now pretty straightforward to derive an equation for the heat transport directly from the tGME~\eqref{eq:tGME}. Indeed, by writing Eq.~\eqref{eq:tGME} in physical time, we obtain a real-time ODE for the $\tilde\rho_S(t;\bm \lambda)$. Then, recalling that $\M(t; \bm \lambda) : = \tr_S \tilde \rho_{S}(t; \bm\lambda)$ by definition, we can obtain a formula for the heat current flowing into the $\alpha$-th reservoir just by differentiating Eq.~\eqref{eq:tGME} with respect to $\lambda_\alpha$ in $\bm \lambda = \bm 0$ and then taking the trace with respect to the system degrees of freedom:
\begin{equation}
\label{eq:heat-current}
    \begin{split}
        I_Q^{(\alpha)}(t) & := \frac{\d \ev{Q_\alpha}(t)}{\d t}\\ & = i\tr_S \left[\frac{\partial}{ \partial \lambda_\alpha}\frac{\d \tilde\rho_S(t; \bm \lambda)}{\d t} \bigg|_{\bm \lambda = \bm 0} \,\right]~.
    \end{split}
\end{equation}
By carrying out this procedure, as done in App.~\ref{app:C}, we obtain the following heat transport equation, which is the main result of this work:
\begin{equation}
\label{eq:transport-equation}
    \begin{split}
        I_Q^{(\alpha)}(t)=  \int_{0}^{t} \d \tau\, 
            \bigl[\g^>_{\alpha;\mu\nu}(t, \tau) - \g^\tord_{\alpha;\mu\nu}(t, \tau) \bigr] &\ev{A_\mu(t)A_\nu(\tau)}_t \\&+{\rm c.c.} ~.
    \end{split}
\end{equation}
Here, $\g^>_\alpha$ and $\g^\tord_\alpha$ are Keldysh components [see also Eqs. \eqref{eq:C->}, \eqref{eq:C-tord}] of the {\it dressed heat kernel} corresponding to the $\alpha$-th reservoir, defined as
\begin{equation}
\label{eq:goth-g}
    \g_\alpha (z,w) : = i\frac{\partial\tilde \G(z, w; \bm \lambda )}{\partial \lambda_\alpha} \bigg|_{\bm \lambda = \bm 0}~,
\end{equation}
which is convolved with the system correlators
\begin{equation}
\label{eq:system-correlators}
    \ev{A_\mu(t)A_\nu(\tau)}_t : = \tr_S\left[ A_\mu(t)A_\nu(\tau) \rho_S(t) \right]~.
\end{equation}
We emphasize that evaluating these correlators requires full knowledge of the system dynamics, encoded in the density matrix $\rho_S(t)$. In order to have an exact equation, the latter must be calculated via the GME~\eqref{eq:gme}. We further observe that, through the system correlators~\eqref{eq:system-correlators}, the heat current depends on the system’s initial preparation. Further differentiation of the tGME~\eqref{eq:tGME} yields ordinary differential equations for higher-order moments of the heat distribution. A closed form for the heat covariance matrix $\ev{Q_\alpha Q_\beta}(t)$ is presented and derived in App.~\ref{app:F}. 

Let us now turn to the dressed heat kernel~\eqref{eq:goth-g}. The calculation of the latter, carried out in App.~\ref{app:D}, is performed via first-order-in-$\bm \lambda$ Taylor expansion of the tilted Dyson equation~\eqref{eq:tilted-dyson}. It naturally finds a close recursive form in terms of the {\it bare} heat kernel
\begin{equation}
\label{eq:goth-c}
    \cc_\alpha(z,w) : = i\frac{\partial\tilde \C(z, w; \bm \lambda )}{\partial \lambda_\alpha} \bigg|_{\bm \lambda = \bm 0}~.
\end{equation}
Such a recursive form can be presented as follows:
\begin{equation}
\label{eq:goth-g-dyson-contour}
    \begin{split}
        &{\g}_\alpha(12)  = {\cc}_\alpha(12) \\
        &+\iint_{\gamma(t)} \d(34)\,\Bigl\{\G(13) \Sigma(34)   {\cc}_{\alpha}(42) \\
        &\phantom{\iint_{\gamma(t)} \d(34)\,\Bigl\{} + {\g}_{\alpha}(13) \Sigma(34)  \C(42)\Bigr\}~.
    \end{split}
\end{equation}
We observe that the solution of Eq.~\eqref{eq:goth-g-dyson-contour} requires knowledge of the GME dressed GF $\G$ given by Eq.~\eqref{eq:bare-dyson}. We now proceed to rewrite Eq.~\eqref{eq:goth-g-dyson-contour} in physical time, obtaining coupled integral equations for the Keldysh components of $\g_\alpha$. The resulting Dyson form for the pair $\vec \g_\alpha := (\g^>_\alpha, \g^\drot_\alpha)$ is
\begin{equation}
\label{eq:physical-time-goth-g-dyson}
        \vec \g_\alpha
        =
        \begin{pmatrix}
            \cc_\alpha^> - \K^< \cc_\alpha^>\\
            \K^> \cc_\alpha^<
        \end{pmatrix}
        +
        \begin{pmatrix}
             \I^\tord\ & -\I^>\\
             \I^< & -\I^\drot
        \end{pmatrix}
        \vec \g_\alpha~,
\end{equation}
where we defined the following functionals:
\begin{align}
\label{eq:K-functional}
    \bigl[\K^X f\bigr](12) & : = \iint_{[0, t]} \d^2 \bm w\,\G^X(13) \,\Sigma^\tord(34)  \, f(42)~,\\
\label{eq:I-functional}
    \bigl[\I^Yf\bigr](12) & : = \iint_{[0, t]} \d^2 \bm w\,f(13)\,\Sigma^\tord(34)\, \C^Y(42) ~,
\end{align}
for $X\in[>,<]$ and $Y \in [>, <,\tord, \drot]$. A simpler form of such functionals in terms of elementary correlators is given in App.~\ref{app:D}. The remaining Keldysh components, i.e. $\g^<_\alpha$ and $\g^\tord_\alpha$, can be obtained from $\g^>_\alpha$ and $\g^\drot_\alpha$ respectively by performing the simple effective conjugation
\begin{equation}
   \effconj{f} := \zeta (\1_R\otimes \sigma_x) f^* (\1_R\otimes \sigma_x) ~,
\end{equation}
as proven in Corollary~\ref{cor:3} to Lemma~\ref{lemma:4} in App.~\ref{app:C}. We note that Eq.~\eqref{eq:physical-time-goth-g-dyson} may be solved employing well-studied tools of many-body and condensed matter physics (see e.g. Refs.~\cite{Talarico2019scalable, lamic2025solving}). In Fig.~\ref{fig:3} we recap the results obtained so far, remarking the connection between our work and the existing literature.
\subsection{Details on the bare heat kernel}
\label{subsec:IVb}
In order to give a complete picture, we now discuss the bare heat kernel and identify inside the latter the contributions connected to energy transport and those related to particle exchange. Such a procedure will allow us to write closed forms for the energy current and for the particle number current, namely
\begin{equation}
    I_{\E}^{(\alpha)}(t) := \frac{\d \ev{\E_\alpha}(t)}{\d t}~, \quad I_{N}^{(\alpha)}(t) := \frac{\d \ev{N_\alpha}(t)}{\d t}~.
\end{equation}
Out of the four Keldysh components of $\cc_\alpha$, only $\cc^>_\alpha$ and $\cc^<_\alpha$ are non-zero, as can be immediately seen by looking at definiton~\eqref{eq:goth-c} together with Eqs.~\eqref{eq:C-tord}-\eqref{eq:C-<}. Moreover, as argued above, we have $\cc^<_\alpha = \effconj\cc^>_\alpha$, so we may restrict out attention to $\cc^>_\alpha$. Since only the $\alpha$-th block in the matrix form~\eqref{eq:tilted-GF-blocks} of the bare, tilted GF depends on $\lambda_\alpha$, the bare heat kernel consists of a single non-zero block, i.e.
\begin{equation}
    \cc_\alpha = \operatorname{diag}\left(0, \dots,i\frac{\partial\tilde \C_\alpha}{\partial \lambda_\alpha} \bigg|_{\bm \lambda = \bm 0}, \dots, 0\right)~,
\end{equation}
which in the following---by a slight abuse of notation---we will denote simply as $\cc_\alpha$. We observe that this property does not hold for the dressed heat kernel $\g_\alpha$, as out-of-diagonal terms appear after “dressing" the bare GF. 
Per the derivation in App.~\ref{app:E}, we can express $\cc_\alpha$ in terms of the physical-time GF~\eqref{eq:chi}:
\begin{equation}
\label{eq:bare-heat-kernel}
     \cc_\alpha^> (z -w)= i \frac{\partial \chi_\alpha(s)}{\partial s}\Bigg|_{s=z-w} +\mu_\alpha  \sigma_z \chi_{\alpha}(z-w)~,
\end{equation}
where we emphasized the time-homogeneity of the bare heat kernel. The RHS of Eq.~\eqref{eq:bare-heat-kernel} is particularly eloquent, as it is already separated into contributions of different nature: $i\partial_s\chi$ is evidently responsible for energy transport, as it is the only remaining contribution when substituting the standard reservoir Hamiltonians $\H{\alpha}$ to the grand-canonical Hamiltonians $\Omega_\alpha$ in Eq.~\eqref{eq:CF}, which is equivalent to setting $\mu_\alpha = 0$ $\forall \alpha$. The remaining term, proportional to $\chi_\alpha$, can be therefore identified as responsible for particle number transport (in the opposite direction). We may thus define the bare {\it energy} and {\it particle number} kernels
\begin{align}
    \bigl[\cc_{\alpha}^{(\E)}(z - w)\bigr]^> & := i\frac{\partial \chi_\alpha(s)}{\partial s}\Bigg|_{s=z-w} ~,\\
     \bigl[\cc_{\alpha}^{(N)}(z -w)\bigr]^> & :=  -\sigma_z   \chi_{\alpha}(z-w)~,
\end{align}
such that
\begin{equation}
    \cc_\alpha =\cc^{(\E)}_\alpha-\mu_\alpha \cc^{(N)}_\alpha~.
\end{equation}
It is then immediate to see that we can solve for the fixed point of the recursive map~\eqref{eq:physical-time-goth-g-dyson} with $\cc^{(\E)}_\alpha$ and $\cc^{(N)}_\alpha$ as starting points, thereby obtaining dressed kernels $\g^{(\E)}_\alpha$ and $\g^{(N)}_\alpha$, respectively. Substituting these to $\g_\alpha$ in Eq.~\eqref{eq:transport-equation} yields exact expressions for $I_{\E}^{(\alpha)}(t)$ and $I_{N}^{(\alpha)}(t)$ of the same formal structure as that for $I_{Q}^{(\alpha)}(t)$.

We close the discussion on the bare heat kernel by giving a simple analytical recipe to evaluate the latter for the case of uncorrelated thermal reservoirs characterized by spectral densities $\{J_\alpha(\omega)\}_\alpha$. Inspecting the definition of $\chi$ given in Eq.~\eqref{eq:chi}, is it straightforward to see that each bare heat kernel must assume the matrix form
\begin{equation}
\label{eq:bare-heat-kernel-matrix}
    \cc_\alpha = 
    \begin{pmatrix}
        0 & \cc_{\alpha,-}\\
        \cc_{\alpha,+} & 0
    \end{pmatrix}~,
\end{equation}
where the notation is reminiscent of the conventional association of the elements of $\chi$ with absorption and emission processes~\footnote{This $\pm$ subscript introduced here, which refers to the standard association of certain GF matrix elements to absorption and emission processes, is not to be confused with the similar notation introduced in Sec.~\ref{subsec:IIa} for variables on the $\gamma_{\pm}$ branches of the Keldysh contour.}. As shown in App.~\ref{app:E}, the two relevant matrix elements can be computed explicitly as
\begin{align}
\label{eq:goth-c-}
    \cc_{\alpha, -}^>(s)  
        & =  - \int_0^{\infty} \d \omega \,\bigl(\omega-\mu_\alpha\bigr) J_\alpha(\omega) f_\alpha(\omega) e^{i\omega s}~, \\[5pt]
\label{eq:goth-c+}
    \cc_{\alpha, +}^>(s)  
        & = + \int_0^{\infty} \d \omega \,\bigl(\omega-\mu_\alpha\bigr) J_\alpha(\omega) \bigl[1+\zeta f_\alpha(\omega) \bigr]e^{-i\omega s}~,
\end{align}
where 
\begin{equation}
    f_\alpha(\omega) = \frac{1}{e^{\beta_{\alpha}(\omega - \mu_\alpha)} - \zeta}
\end{equation}
is the Bose-Einstein (Fermi-Dirac) distribution for $\zeta = 1$ ($\zeta = -1$).

\section{Weak-coupling limit and steady-state limit}
\label{sec:V}

Inspecting the map given by Eq.~\eqref{eq:physical-time-goth-g-dyson}, we see that acting on the $(\g^>, \g^\drot)^{(n)}$ pair with the $\I$ functionals shifts the pair to the next order in the system-environment coupling. The term $\K^{\drot} \cc^>_\alpha$ ($\K^{>} \cc^<_\alpha$), when calculated by plugging in an iteration $\G^{(n')}$ of the dressed GF exact at order $n'+1$ in the coupling, is of order $n'+2$. Therefore, the lowest-order non-vanishing contribution to $\g_\alpha$ is just the bare heat kernel $\cc_\alpha$, exact up to first order in the coupling, in complete analogy with the dressed GF solving Eq.~\eqref{eq:bare-dyson}. Since the bare heat kernel lacks a time-ordered Keldysh component, we can immediately write the weak-coupling (w.c.) limit of the exact transport equation~\eqref{eq:transport-equation} as follows:
\begin{equation}
\label{eq:wc-transport-equation}
    \begin{split}
        I_Q^{(\alpha)}(t) \overset{\rm w.c.}{=}  \int_{0}^{t} \d \tau\, 
            \cc^>_{\alpha;\mu\nu}(t- \tau)&\ev{A_\mu(t)A_\nu(\tau)}_t +{\rm c.c.} ~.
    \end{split}
\end{equation}
We may get an even more eloquent form if we exploit the form~\eqref{eq:bare-heat-kernel-matrix} of the bare heat kernel and identify---in the absence on anomalous system correlations---the system correlators associated with emission and absorption into the $\alpha$-th reservoir, i.e.
\begin{align}
    \AA^-_\alpha(t, \tau) &:= \ev{A_\alpha^{\dagger}(t)A_\alpha(\tau)}_t~,\\
    \AA^+_\alpha(t, \tau) &:= \ev{A_\alpha(t)A_\alpha^{\dagger}(\tau)}_t~.
\end{align}
Thus, Eq.~\eqref{eq:wc-transport-equation} may be presented as
\begin{equation}
\label{eq:wc-transport-equation-II}
    \begin{split}
         I_{Q}^{(\alpha)}(t) \overset{\rm w.c.}{=} &\int_0^t \d \tau\, \Bigl[\cc_{\alpha,-}^>(t-\tau)\AA_\alpha^+(t, \tau) \\
        & \phantom{\int_0^t \d \tau} \!\!+ \cc_{\alpha,+}^>(t-\tau)\AA_\alpha^-(t, \tau) \Bigr] + {\rm c.c.}~,
    \end{split}
\end{equation}
where the absorption component of the heat kernel couples to the emission correlator of the system, and vice versa. To be consistent with the w.c. limit, the system correlators evolution must be described within the weak-coupling limit of the GME, i.e. Eq.~\eqref{eq:gme} with $\G \to \C$.

We observe that Eq.~\eqref{eq:wc-transport-equation}, although being valid in the weak-coupling limit only, still yields the heat currents at each instant $t \geq 0$, thus incorporating the full out-of-equilibrium transient dynamics of the system-environment compound. Moreover, such dynamics will clearly depend on the initial preparation of the system.

To gain insight on the steady-state reduction of the present framework, let us consider a simple example in which fermionic thermal reservoirs are coupled through a single fermionic mode $\a{0}$ of energy $\omega_0$, such that the relevant system correlators are given by
\begin{align}
    \AA_0^-(t,\tau) & = e^{+i\omega_0(t-\tau)}\,f_0(t)~,\\
    \AA_0^+(t,\tau) & = e^{-i\omega_0(t-\tau)}\,\bigl[1-f_0(t)\bigr]~,
\end{align}
where $f_0(t) = \ev{\ad{0} \a{0}}_t$ is the population of the mode at time $t \geq 0$. Equation~\eqref{eq:wc-transport-equation-II} immediately yields the following:
\begin{equation}
    \begin{split}
        I_{Q}^{(\alpha)}(t) & \overset{\rm w.c.}{=} \int_0^{\infty} \d \omega\,(\omega-\mu_\alpha)\,J_\alpha(\omega)\\&\times\Bigl\{f_0(t) - f_\alpha(\omega) \Bigr\} \left[ 2\frac{\sin[(\omega-\omega_0)t]}{\omega-\omega_0}\right]~.
    \end{split}
\end{equation}
We may now consider the distributional limit
\begin{equation}
    \lim_{t\to \infty}\frac{\sin[(\omega-\omega_0)t]}{\omega-\omega_0} = \pi \delta(\omega-\omega_0)~,
\end{equation}
and admit a broadening of the system spectrum, i.e.
\begin{equation} 
\label{eq:broadening}
    \pi \delta(\omega-\omega_0) \to \frac{\gamma}{(\omega-\omega_0)^2 + \gamma^2}~,
\end{equation}
in order to mimic a transport channel of finite spectral width. The latter is assumed to approach a steady-state population $f_0(\omega) := \lim_{t\to\infty} f_0(t, \omega)$, which must now be treated as energy-dependent as a consequence of the broadening [Eq.~\eqref{eq:broadening}]. Eventually, we obtain
\begin{equation}
\label{eq:landauer-buttiker}
    I_{Q}^{(\alpha)}(\infty)  \overset{\rm w.c.}{=} \int_{0}^{\infty} \d \omega\,(\omega-\mu_\alpha)\,\T_{\alpha}(\omega)\bigl[f_0(\omega)- f_\alpha(\omega) \bigr]~,
\end{equation}
where we defined the transmission function
\begin{equation}
    \T_\alpha(\omega) = \frac{2\gamma J_\alpha(\omega)}{(\omega-\omega_0)^2 + \gamma^2}~.
\end{equation}
Eq.~\eqref{eq:landauer-buttiker} matches the well-known formula for the heat current flowing into the $\alpha$-th lead that we would have obtained by applying the Landauer--B\"uttiker scattering formalism~\cite{landauer1957spatial, buttiker1986four}.

\section{A case study: transient negative conductance}
\label{sec:VI}
\begin{figure}[t]
  \centering
  \def\svgwidth{1.0\linewidth}
  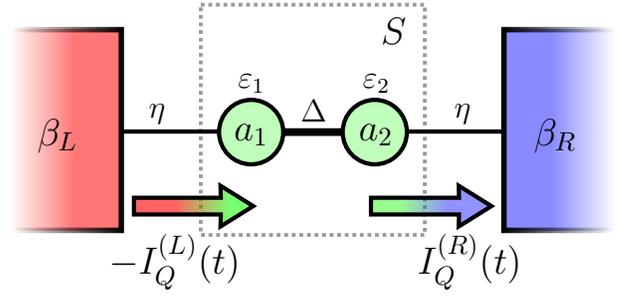
  \caption{Schematic representation of the model analyzed in Sec.~\ref{sec:VI}, showing the current flow within the system-environment compound.} 
  \label{fig:4}
\end{figure}
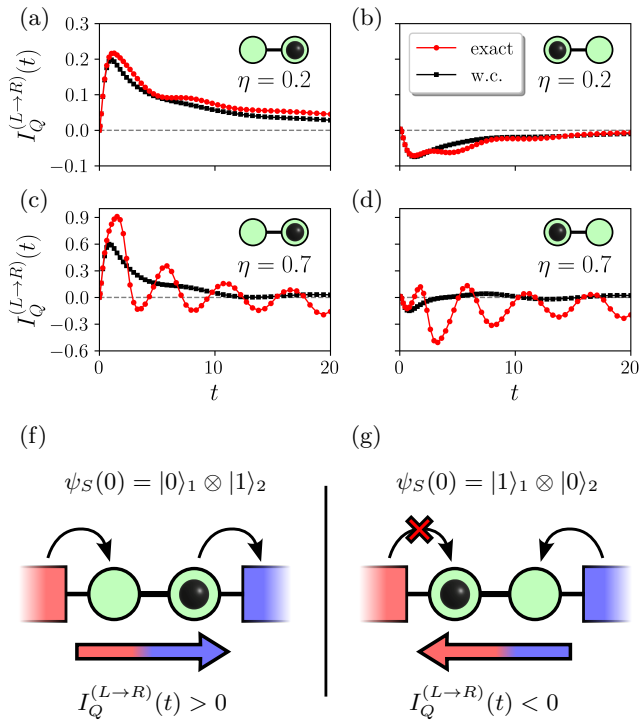
\begin{figure}[t]
    \centering
    \begin{tabular}{c}
    \begin{overpic}[width = 1.\linewidth]{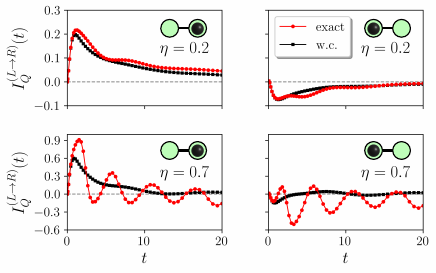}
        \put(3,60){(a)}
        \put(54,60){(b)}
        \put(3,32){(c)}
        \put(54,32){(d)}
    \end{overpic}\\[10pt]
    \begin{overpic}[width = 1.\linewidth]{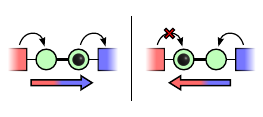}
        \put(3,45){(e)}
        \put(10,38){$ \psi_S(0) = |0\rangle_1 \otimes |1\rangle_2 $}
        \put(12,4){$  I_Q^{(L\to R)}(t) > 0 $}
        \put(54,45){(f)}
        \put(61,38){$ \psi_S(0) = |1\rangle_1 \otimes |0\rangle_2 $}
        \put(63,4){$  I_Q^{(L\to R)}(t) < 0 $}
    \end{overpic}
    \end{tabular}
    \caption{(a)-(d) Exact (red dots) and weak-coupling (black squares) current flow $I_Q^{(L\to R)}$ [Eq.~\eqref{eq:LtoR-current}] for different values of the coupling constant $\eta$ and different system initialization (depicted in the inset). The system parameters are $\varepsilon_1 = 0.5$, $\varepsilon_2=1.0$, $\Delta = 0.7$. The reservoirs are initialized at inverse temperatures $\beta_L = 1$ and $\beta_R = 10$. Other numerical parameters are specified in App.~\ref{app:G}. When the system is initialized in $|1\rangle_1 \otimes |0\rangle_2$ [(b) and (d)] we observe transient anomalous heat flow from the colder reservoir to the hotter one. In the strong-coupling regime [(c) and (d)] the w.c. approximation fails to capture the full current dynamics. (e)-(f) Pictorial representation of the phenomenology behind the emergence of negative heat conductance at small $t$: when site $1$ is full and site $2$ is empty [(f)], Pauli exclusion prevents heat exchange between the hotter reservoir and the system, and the latter is forced to absorb heat from the colder reservoir, thus resulting in net heat flow against the thermal bias.}
    \label{fig:5}
\end{figure}
In this section, we employ the exact transport equation [Eq.~\eqref{eq:transport-equation}] to investigate a minimal model, revealing highly non-trivial transient dynamics.

The model we choose, schematized in Fig.~\ref{fig:4}, consists of a two-site fermion chain subject to anomalous pairing. The system Hamiltonian is
\begin{equation}
    \H{S} = \varepsilon_1 \ad{1}\a{1} + \varepsilon_2 \ad{2}\a{2} + \Delta \Bigl(\ad{1}\ad{2}-\a{1}\a{2}\Bigr)~,
\end{equation}
where $\varepsilon_1, \varepsilon_2$ are on-site energies and $\Delta$ is the anomalous pairing energy. We couple such chain to multimode reservoirs characterized by Hamiltonians
\begin{equation}
    \H{\alpha} = \sum_l \omega_{\alpha, l}\cd{\alpha, l} \c{\alpha, l}~,
\end{equation}
where $\alpha = L, R$ denote the left and right reservoirs, respectively, and each mode $l$ is represented by the annihilation operator $\c{\alpha, l}$ . Each reservoir is linearly coupled to the a single site of the chain:
\begin{equation}
    \V = \sum_l \left(g^{\vphantom{\dagger}}_{L, l} \,\ad{1} \c{L, l} + g^{\vphantom{\dagger}}_{R, l} \,\ad{2} \c{R, l}\right) + \hc~,
\end{equation}
where $g^{\vphantom{\dagger}}_{\alpha, l}$ is the coupling strength relative to the $l$-th mode. In the following, we assume both reservoirs possess the same Lorentzian spectral density, such that the tunneling strength is set by a global coupling constant $\eta$, as argued in App.~\ref{app:G}. Each reservoir is initially prepared in a Gibbs state with $\mu_L = \mu_R =0$:
\begin{equation}
    \rho_\alpha(0) = \frac{e^{-\beta_\alpha \H{\alpha}}}{\tr e^{-\beta_\alpha \H{\alpha}}}~.
\end{equation}
We study the heat current flow from the left reservoir the right reservoir, defined as
\begin{equation}
\label{eq:LtoR-current}
    I_Q^{(L\to R)}(t) := I_Q^{(R)}(t) - I_Q^{(L)}(t)~.
\end{equation}
If we choose $\beta_R \gg \beta_L$, so that---in the steady-state regime---a thermal gradient will be eventually established, driving heat from the hotter left reservoir to the colder right one. This implies positivity of the heat conductance $C_Q$ at steady state, where $C_Q$ is given by
\begin{equation}
    C_Q(t) := \frac{I_Q^{(L\to R)}(t)}{T_L - T_R}~.
\end{equation}
The heat current [Eq.~\eqref{eq:LtoR-current}] is shown in Fig.~\ref{fig:5}(a-d) for different coupling constants and system initializations, comparing the exact results of Eq.~\eqref{eq:transport-equation} with its w.c. limit~\eqref{eq:wc-transport-equation}.

When the system is initialized in the pure single-fermion state in which site $2$ is fully occupied and site $1$ is empty, i.e.
\begin{equation}
    \psi_S(0) = |0\rangle_1 \otimes |1\rangle_2 ~,
\end{equation}
as in Figs.~\ref{fig:5}a and~\ref{fig:5}c, we see that the heat conductance is mostly positive in the transient regime too. Indeed, as schematized in Fig.~\ref{fig:5}e, the hopping processes that carry fermions from the hot reservoir to the cold reservoir are favored by (effective) temperature differences. If we instead prepare the system in the state
\begin{equation}
    \psi_S(0) = |1\rangle_1 \otimes |0\rangle_2 ~,
\end{equation}
swapping the filling of the two sites with respect to the first case, we observe mostly-negative heat conductance, as the heat flows from the colder reservoir to the hotter one. This transient negative heat conductance is a strongly non-equilibrium transient phenomenon which finds a simple explanation illustrated in Fig.~\ref{fig:5}f: the hopping process which would bring a fermion from the hot reservoir $R$ to fermion site $1$ is suppressed due to Pauli exclusion, as site $1$ is completely full at $t = 0$. Moreover, since site $2$ is initialized in a state with zero effective temperature, hopping from the cold $L$ reservoir to empty site $2$ is favored. The confluence of these two factors gives rise to a transient regime in which heat flows against the imposed thermal bias. We note that, in the strong-coupling regime [Fig.~\ref{fig:5}c and~\ref{fig:5}d], the w.c. approximation fails to capture the strongly oscillatory behavior we observe in the exact solution to the transport equation.
\section{Conclusions}
\label{sec:VII}

{
We formulated an exact theory for the characteristic function for quantum thermodynamics using the Keldysh contour technique [see Eqs.~\eqref{eq:CF-trace}-\eqref{eq:CF-contour}]. We then adapted the exact Gaussian master equation formalism to derive an ordinary differential equation for this generating function, the tilted Gaussian master equation (tGME) [Eq.~\eqref{eq:tGME}]. The latter is a novel formal instrument which contains the full information on the heat statistics of the system-environment compound at any given time. We employed such an instrument to obtain a concise transport equation for the heat currents within the system-environment compound [Eq.~\eqref{eq:transport-equation}], which is the primary result of this article. This toolkit enables a rigorous treatment of non-Markovian effects and out-of-equilibrium transient regimes in Gaussian systems. By taking the weak-coupling and steady-state limits of Eq.~\eqref{eq:transport-equation}, we demonstrated that our approach reduces to the Landauer--Büttiker scattering formalism. Finally, we applied the exact transport equation~\eqref{eq:transport-equation} to study heat transport in a minimal fermion model. Our results uncover a regime of transient negative heat conductance, serving as a striking demonstration of non-trivial out-of-equilibrium dynamics.
Our results provide a foundation for several viable developments. An exact treatment of heat transport in Gaussian systems could uncover the emergence on non-trivial phenomena, such as rectification and negative differential conductance~\cite{tesser2022heat, ingi2017reversal, sierra2014strongly}. Another promising direction is the study on non-Markovian effects in devices exhibiting bipolar thermoelectricity~\cite{antola2025quantum} or in Cooper pair splitters \cite{mayo2025tur, Cao2015thermoelectric}.
Furthermore, the toolkit presented here, in synergy with the GME formalism (Ref.~\cite{d2025exact}), could find fruitful applications in describing memory effects in noisy quantum architectures---a prerequisite for tailoring efficient error correction codes for quantum computation~\cite{bourgeois2026transport}.

\section*{Acknowledgments}

We thank Fabio Taddei and Antonio D'Abbruzzo for useful discussions. We acknowledge financial support by MUR (Ministero dell’Universit{\`a} e della Ricerca) through the PNRR MUR project PE0000023-NQSTI. 
    
\appendix

\onecolumngrid

\section{The characteristic function framework on the Keldysh contour}
\label{app:A}
In this appendix, we derive Eq.~\eqref{eq:CF-exp} for the CF as a contour-ordered integral and we prove Eq.~\eqref{eq:heat-from-CF} for calculating the heat exchange in the TPEM scheme. Moreover, we establish two useful lemmas relating the four Keldysh components of the tilted GF, and we obtain the explicit forms~\eqref{eq:C-tord}-\eqref{eq:C-<} for said components.
\begin{proof}[Proof of Eq.~\eqref{eq:CF-exp}]
Let us start from the well-known form [Eq.~\eqref{eq:CF}] for the CF: 
\begin{equation}
\label{sm-eq:CF}
    \begin{split}
        \M(t; \bm \lambda) :=  \tr\left[ \tilde \U_{\bm \lambda}^{\vphantom{\dagger}} (t, 0) \rho_{SE}^{\vphantom{\dagger}}(0) \,\tilde \U_{-\bm\lambda}^\dagger(t, 0) \right]~,
    \end{split}
\end{equation}
The explicit form of the tilted evolution operator is
\begin{equation}
    \tilde \U_{\bm \lambda} (t, 0) := u_{\bm \lambda}^{\vphantom{\dagger}}\, \U(t, 0)\, u_{\bm\lambda}^\dagger:= e^{-i \sum_\alpha \Omega_{\alpha}\lambda_\alpha/2} \,\U(t, 0) \,e^{i \sum_\alpha\Omega_{\alpha}\lambda_\alpha/2}~.
\end{equation}
Using the time-ordered exponential form for the interaction-picture unitary [Eq.~\eqref{eq:timeord}], we act on the left and on the right with $u_{\bm \lambda}^{\vphantom{\dagger}},\,u_{\bm \lambda}^\dagger$:
\begin{equation}
    \begin{split}
        \tilde \U_{\bm \lambda}(t, 0) 
        & = u_{\bm \lambda}^{\vphantom{\dagger}}\, \tord \exp \left[-i\int_0^t \d \tau\,\V(\tau)\right]\, u_{\bm\lambda}^\dagger \\[3pt]
        & \overset{\rm (i)}{=}  \tord\left\{ u_{\bm \lambda}^{\vphantom{\dagger}}\exp \left[-i\int_0^t \d \tau\,\V(\tau)\right] u_{\bm\lambda}^\dagger\right\}\\[3pt]
        & \overset{\rm (ii)}{=}  \tord \exp \left[-i\int_0^t \d \tau\,u_{\bm \lambda}^{\vphantom{\dagger}}\V(\tau)\,u_{\bm\lambda}^\dagger\right]~.
    \end{split}
\end{equation}
Here, (i) is justified by $\Omega_{\alpha}$ being time-independent, and for (ii) we used $u_{\bm \lambda}^{\vphantom{\dagger}} u_{\bm \lambda}^\dagger = \1$. The “conjugate" operator appearing in Eq.~\eqref{sm-eq:CF} is
\begin{equation}
    \begin{split}
        \tilde \U^\dagger_{-\bm \lambda}(t, 0) 
        & = u_{\bm \lambda}^{\dagger}\, \drot \exp \left[+i\int_0^t \d \tau\,\V(\tau)\right]\, u_{\bm\lambda}^{\vphantom{\dagger}} \\[3pt]
        & = \drot\left\{ u_{\bm \lambda}^{\dagger}\exp \left[+i\int_0^t \d \tau\,\V(\tau)\right]  u_{\bm\lambda}^{\vphantom{\dagger}} \right\}\\[3pt]
        & = \drot \exp \left[+i\int_0^t \d \tau\,u_{-\bm \lambda}^{\vphantom{\dagger}}\V(\tau) \,u_{-\bm\lambda}^\dagger \right]~.
    \end{split}
\end{equation}
Here, $\drot$ is the anti time-ordering superoperator, which arranges operators in reverse chronological order compared to $\tord$. Introducing the Keldysh contour $\gamma(t) = \gamma_+(t) \oplus \gamma_-(t)$ as in Fig.~\ref{fig:2} we can thus write the generating function as the trace of a single time-ordered exponential, eventually recovering Eq.~\eqref{eq:CF-exp} with definition~\eqref{eq:tilted-interaction}:
\begin{equation}
\label{sm-eq:CF-exp}
\begin{split}
    \M(t; \bm \lambda) & =  \tr\tord \left\{\exp\left[ -i\int_{\gamma(t)} \d z\,\V_{\bm \lambda}(z) \right] \rho_{SE}^{\vphantom{\dagger}}(0)\right\}~.
\end{split}
\end{equation}
\end{proof}
\begin{proof}[Proof of Eq.~\eqref{eq:heat-from-CF}]
This result follows from straightforward algebraic manipulation:
\begin{equation}
    \begin{split}
        i\frac{\partial}{\partial \lambda_\alpha} \M(t; \bm \lambda) \biggr|_{\bm \lambda = \bm 0}
        & =   i\frac{\partial}{\partial \lambda_\alpha}\tr\left[ \tilde\U_{\bm \lambda}^{\vphantom{\dagger}} (t, 0) \rho_{SE}^{\vphantom{\dagger}}(0) \tilde\U_{-\bm\lambda}^\dagger(t, 0) \right] \biggr|_{\bm \lambda = \bm 0}\\[3pt]
        & =   i\tr\left[ \frac{\partial \tilde \U_{\bm \lambda}^{\vphantom{\dagger}} (t, 0)}{\partial \lambda_\alpha} \rho_{SE}^{\vphantom{\dagger}}(0) \,\tilde\U_{-\bm\lambda}^\dagger(t, 0) + \tilde\U_{\bm \lambda}^{\vphantom{\dagger}} (t, 0) \rho_{SE}^{\vphantom{\dagger}}(0) \frac{\partial \tilde\U_{-\bm\lambda}^\dagger(t, 0)}{\partial \lambda_\alpha} \right]\biggr|_{\bm \lambda = \bm 0} \\[3pt]
        & = \frac{1}{2}\tr\left[  \bigl[\Omega_\alpha,\tilde\U_{\bm \lambda}^{\vphantom{\dagger}} (t, 0) \bigr] \rho_{SE}^{\vphantom{\dagger}}(0) \,\tilde\U_{-\bm\lambda}^\dagger(t, 0) -\tilde\U_{\bm \lambda}^{\vphantom{\dagger}} (t, 0) \rho_{SE}^{\vphantom{\dagger}}(0)  \bigl[ \Omega_\alpha,\tilde\U_{-\bm\lambda}^\dagger(t, 0) \bigr]\right] \biggr|_{\bm \lambda = \bm 0}\\[3pt]
        & = \tr\Bigl[  \Omega_\alpha\,\tilde\U_{\bm \lambda}^{\vphantom{\dagger}} (t, 0)  \rho_{SE}^{\vphantom{\dagger}}(0)\, \tilde\U_{-\bm\lambda}^\dagger(t, 0) \Bigr] -  \frac{1}{2}\tr\Bigl[ \tilde\U_{\bm \lambda}^{\vphantom{\dagger}} (t, 0) \Bigl\{\Omega_\alpha, \rho_{SE}^{\vphantom{\dagger}}(0) \Bigr\}\,\tilde\U_{-\bm\lambda}^\dagger(t, 0) \Bigr] \biggr|_{\bm \lambda = \bm 0}\\[3pt]
        & \overset{\rm (i)}{=} \tr\Bigl[  \Omega_\alpha\,\tilde\U_{\bm \lambda}^{\vphantom{\dagger}} (t, 0)  \rho_{SE}^{\vphantom{\dagger}}(0) \,\tilde\U_{-\bm\lambda}^\dagger(t, 0) \Bigr] -  \tr\Bigl[  \Omega_\alpha\, \rho_{SE}^{\vphantom{\dagger}}(0) \,\tilde\U_{-\bm\lambda}^\dagger(t, 0)\, \tilde\U_{\bm \lambda}^{\vphantom{\dagger}} (t, 0)\Bigr] \biggr|_{\bm \lambda = \bm 0}\\[3pt]
        & \overset{\rm (ii)}{=} \tr\Bigl[  \Omega_\alpha\,\U (t, 0)  \rho_{SE}^{\vphantom{\dagger}}(0) \,\U^\dagger(t, 0) \Bigr] -  \tr\Bigl[ 
        \Omega_\alpha \rho_{SE}^{\vphantom{\dagger}}(0) \,\U^\dagger(t, 0)\, \U(t, 0) \Bigr] \\[3pt]
        & = \tr\Bigl[  \Omega_\alpha\bigl[\rho_{SE}^{\vphantom{\dagger}}(t) - 
        \rho_{SE}^{\vphantom{\dagger}}(0)  \bigr]\Bigr] \\[3pt]
        & \overset{\rm (iii)}{=} \tr_\alpha \Bigl[  \Omega_\alpha\bigl[\rho_{\alpha}^{\vphantom{\dagger}}(t) - 
        \rho_{\alpha}^{\vphantom{\dagger}}(0)  \bigr]\Bigr] \\[3pt]
        & : = \ev{Q_\alpha}~.
    \end{split}
\end{equation}
In (i) we used the hypothesis that the initial state commutes with $\Omega_{\alpha}$ for all $\alpha$, together with the cyclicity of the trace. In (ii), we used $\tilde \U_{\bm \lambda = \bm 0}(t, 0) = \U(t, 0)$ and unitarity of $\U(t, 0)$. For (iii), we performed the partial trace on $S$ and on $E \setminus \alpha$, recalling then definition~\eqref{eq:heat-exchange} of the heat exchange.
\end{proof}
We now turn our attention to the tilted GF $\tilde \C$ [Eq.~\eqref{eq:tilted-GF}]. We start by establishing a useful lemma. By hermiticity of $\V_{\bm \lambda}(t)$, we can always chose the sets $\{A_\mu\}_\mu$ and $\{B_\mu\}_\mu$ in such a way that for every index $\mu$ in $\V_{\bm \lambda}(z) = \sum_{\mu}A_\mu(z) \otimes \tilde B_\mu(z;\bm \lambda)$ a conjugate index $\bar \mu$ is defined such that
\begin{equation}
\label{sm-eq:conjugate-operator}
    A^\dagger_\mu(z) = A^{\vphantom{\dagger}}_{\bar \mu}(z^*)~, \qquad \tilde B^\dagger_\mu(z;\bm \lambda) = \tilde B_{\bar \mu}^{\vphantom{\dagger}}(z^*; -\bm \lambda)~.
\end{equation}
where $z^*_{\pm} = z_\mp$ ($\,^*\,$ flips the contour branches into each other). Therefore, we can state the following:
\begin{lemma} \label{lemma:1}
For every pair of indices $\mu, \nu$ and for every $z, w \in \gamma(t)$,
\begin{equation}
\label{sm-eq:lemma1}
    \bigl[\tilde\C_{\mu \nu}(z, w; \bm \lambda)\bigr]^* = \zeta \tilde \C_{\bar \mu \bar \nu} (z^*, w^*;-\bm \lambda)~.
\end{equation}
\end{lemma}
\begin{proof} [Proof of Lemma~\ref{lemma:1}]%
Using observation~\eqref{sm-eq:conjugate-operator}, we have:
\begin{equation}
    \begin{split}
        \bigl[\tilde\C_{\mu \nu}(z, w; \bm \lambda)\bigr]^* 
        & : = \vev{\tord_{\zeta} \bigl\{\tilde B_\mu(z;\bm \lambda) \tilde B_\nu(w;\bm \lambda)\bigr\}}^*\\[2pt]
        & = \tr \bigl[ \tord_{\zeta}\bigl\{ \tilde B_{\mu}(z; \bm \lambda) \tilde B_{\nu}(w; \bm \lambda) \Omega_0\bigr\}\bigr]^*\\[2pt]
        & = \tr \bigl[ \drot_{\zeta}\bigl\{\tilde  B_{\nu}^\dagger(w; \bm \lambda)  \tilde B_{\mu}^\dagger(z; \bm \lambda)\Omega_0\bigr\}\bigr]\\[2pt]
        & = \tr \bigl[ \tord_{\zeta}\bigl\{\tilde B_{\bar\nu}(w^*; -\bm \lambda) \tilde  B_{\bar\mu}(z^*; -\bm \lambda)\Omega_0\bigr\}\bigr]\\[2pt]
        & = \zeta\tr \bigl[ \tord_{\zeta}\bigl\{  \tilde B_{\bar\mu}(z^*; -\bm \lambda)\tilde B_{\bar\nu}(w^*; -\bm \lambda) \Omega_0\bigr\}\bigr]\\[2pt]
        & = \zeta\tilde \C_{\bar \mu \bar \nu}(z^*, w^*; -\bm \lambda)~.
    \end{split}
\end{equation}
\end{proof}
Let us also give a definition:
\begin{defn}[\bf effective conjugate] \label{defn:effective-conjugate}
Given a set of complex-valued functions $F_{\mu \nu}(\bm \lambda)$, we define the \textbf {effective conjugate} of each element as
\begin{equation}
    \bigl[\effconj F\bigr]_{\mu \nu}(\bm \lambda) := \zeta F^*_{\bar \mu \bar \nu}(-\bm \lambda)~,
\end{equation}
where the mapping $\mu \to \bar \mu$ is defined by Eq.~\eqref{sm-eq:conjugate-operator}.
\end{defn}
By choosing a convenient ordering of the reservoir operators, such a mapping can be always be expressed as
\begin{align}
\label{sm-eq:barred-index}
    \begin{cases}
        \bar \mu = \mu -1 \qquad {\rm if} ~\mu ~{\rm is~even}~,\\
        \bar \mu = \mu +1 \qquad {\rm if} ~\mu ~{\rm is~odd}~.
    \end{cases}
\end{align}
In this case, the effective conjugation can be presented in the simple matrix form
\begin{equation}
    \effconj F(\bm \lambda) = \zeta(\1_R\otimes \sigma_x) F^*(-\bm \lambda)(\1_R\otimes \sigma_x)~.
\end{equation}
Given these facts, we can prove the following:
\begin{corollary} \label{cor:1}
Given Definition~\ref{defn:effective-conjugate}, the following identities for the bare tilted GF hold:
\begin{equation} 
\label{sm-eq:cor1}
    \tilde \C^{<} = \effconj\tilde \C^> ~, \qquad  \tilde \C^{\,\drot} = \effconj\tilde \C^{\,\tord} ~.
\end{equation}
\end{corollary}
\begin{proof}[Proof of Corollary~\ref{cor:1}] By Lemma~\ref{lemma:1}, we have
\begin{align}
    \tilde \C^<(z, w;\bm \lambda)  := \tilde \C(z_-, w_+; \bm \lambda) 
    & = \zeta (\1_R\otimes \sigma_x) [\tilde \C(z_+, w_-; -\bm \lambda)]^*(\1_R\otimes \sigma_x)\\
    & := \zeta (\1_R\otimes \sigma_x) [\tilde \C^>(z, w; -\bm \lambda)]^*(\1_R\otimes \sigma_x)\notag \\
    & := \effconj \tilde \C^>(z, w;\bm \lambda) \notag~,\\[3pt]
    \tilde \C^{\,\drot}(z, w;\bm \lambda)  := \tilde \C(z_+, w_+; \bm \lambda) 
    & = \zeta (\1_R\otimes \sigma_x) [\tilde \C(z_-, w_-; -\bm \lambda)]^*(\1_R\otimes \sigma_x)\\
    & := \zeta (\1_R\otimes \sigma_x) [\tilde \C^{\,\tord}( z, w; -\bm \lambda)]^*(\1_R\otimes \sigma_x)\notag \\
    & := \effconj \tilde \C^{\,\tord}(z, w;\bm \lambda) \notag~.
\end{align}
\end{proof}
We'll now derive explicit expression for the four Keldysh components. Let us start by considering the time evolution of the tilted environment operators appearing in Eq.~\eqref{eq:tilted-interaction-sum}, which depends on the specific contour branch under consideration. As stated in the main text, exploiting hermiticity of the interaction Hamiltonian~\eqref{eq:interaction-hamiltonian}, we can write the set of environment operators $\{B_\mu\}_\mu$ as a collection of pairs $\{(B_\alpha^\dagger, B_\alpha^{\vphantom{\dagger}})\}_{\alpha=1}^R$, each pair corresponding to a different reservoir. By leveraging the Gaussian hypothesis, we constrain $B_\alpha^{\vphantom{\dagger}}$ ($B_\alpha^\dagger$) to be linear in the $\alpha$-th reservoir annihilation (creation) operators. Thus, we have
\begin{equation}
\label{sm-eq:tilted-reservoir-operators}
    \begin{split}
        \tilde B_\alpha^{\vphantom{\dagger}}(z_\pm;\bm \lambda) 
        & := e^{\pm i\sum_{\beta}\Omega_{\beta} \lambda_{\beta}/2 }\,e^{i\H{E}z } B_\alpha^{\vphantom{\dagger}} e^{-i\H{E}z } \,e^{\mp i\sum_{\beta}\Omega_{\beta} \lambda_{\beta}/2 }\\
        & = e^{\mp i\sum_{\beta}\mu_{\beta} \N_{{\beta}} \lambda_{\beta}/2} e^{i\sum_{{\beta}} \H{\beta}(z \pm \lambda_{\beta}/2)} B_\alpha^{\vphantom{\dagger}}  \,e^{-i\sum_{\beta} \H{\beta}(z \pm \lambda_{\beta}/2)} e^{\pm i\sum_{\beta}\mu_{\beta}\N_{\beta}\lambda_{\beta}/2} \\
        & \overset{\rm(i)}{=} e^{\mp i\mu_{\alpha} \N_{{\alpha}} \lambda_{\alpha}/2} e^{i \H{\alpha}(z \pm \lambda_{\alpha}/2)} B_\alpha^{\vphantom{\dagger}} \, e^{-i\H{\alpha}(z \pm \lambda_{\alpha}/2)} e^{\pm i\mu_{\alpha}\N_{\alpha}\lambda_{\alpha}/2} \\
        & \overset{\rm(ii)}{=} e^{\mp i\mu_{\alpha} \N_{{\alpha}} \lambda_{\alpha}/2} B_\alpha^{\vphantom{\dagger}}(z\pm \lambda_\alpha/2) e^{\pm i\mu_{\alpha}\N_{\alpha}\lambda_{\alpha}/2} \\
        & \overset{\rm(iii)}{=} e^{\pm i\mu_{\alpha} \lambda_{\alpha}/2} B_\alpha^{\vphantom{\dagger}}(z\pm \lambda_\alpha/2) ~.
    \end{split}
\end{equation}
In (i) we used the fact that $[B_\alpha^{\vphantom{\dagger}},\Omega_\beta] = 0$ $\forall \beta \neq \alpha$. In (ii), we exploited the definition of interaction-picture operators. In (iii), we used linearity of $B_\alpha^{\vphantom{\dagger}}$ in the reservoir annihilation operators. Moreover, we observe that, once again by hermiticity of $\V$ and since the Hamiltonian of the environment is particle-number conserving ($[\H{\alpha}, \N_\alpha] = 0$ $\forall \alpha$), each $2 \times 2$ block $\tilde \C_\alpha$ ($\chi_\alpha$) in Eq.~\eqref{eq:tilted-GF-blocks} (Eq.~\eqref{eq:chi}) is constrained to be of the form
\begin{equation}
\label{sm-eq:C-alpha-block}
    \tilde \C_\alpha = 
    \begin{pmatrix}
        0 & \tilde\C_{\alpha, -} \\
        \tilde\C_{\alpha, +} & 0
    \end{pmatrix}~,
    \qquad 
    \chi_\alpha = 
    \begin{pmatrix}
        0 & \chi_{\alpha}^- \\
        \chi_{\alpha}^+ & 0
    \end{pmatrix}~,
\end{equation}
where we chose an ordering for the environment operators and defined
\begin{align}
    \tilde\C_{\alpha, -}(z, w; \bm \lambda ) & : = \tr \tord_\zeta \bigl\{\tilde B_\alpha^\dagger(z;\bm \lambda) \tilde B^{\vphantom{\dagger}}_\alpha(w;\bm \lambda) \,\omega_E\bigr\} ~, \qquad \quad\chi_{\alpha}^-(s): =\tr \bigl[B^\dagger_\alpha(s) B^{\vphantom{\dagger}}_\alpha(0) \,\omega_E\bigr]~,\\
    \tilde\C_{\alpha, +}(z, w; \bm \lambda ) &: = \tr \tord_\zeta \bigl\{\tilde B^{\vphantom{\dagger}}_\alpha(z;\bm \lambda) \tilde B_\alpha^\dagger(w;\bm \lambda) \,\omega_E\bigr\} ~, \qquad \quad\chi_{\alpha}^+(s): =\tr \bigl[B^{\vphantom{\dagger}}_\alpha(s) B_\alpha^\dagger(0)\, \omega_E\bigr]~.
\end{align}
From now on, we will employ this notation, already introduced in Eq.~\eqref{eq:bare-heat-kernel-matrix} of the main text. 

We can now proceed to derive explicit forms for the Keldysh components of $\tilde \C$.
\begin{proof} [Proof of Eqs.~\eqref{eq:C-tord} to~\eqref{eq:C-<}]%
By Corollary~\ref{cor:1}, we observe that we can obtain the Keldysh components of every $\tilde\C_{\alpha, +}$ from the conjugate components of $\tilde\C_{\alpha, -}$, so we only need to calculate the latter. We start by showing that $\tilde\C^{\,\tord}$ is not affected by the tilting [Eq.~\eqref{eq:C-tord}]:
\begin{equation}
\begin{split}
    \tilde \C_{\alpha, -}^{\,\tord}(z,w) 
    & :=  \tr \tord_\zeta \bigl\{\tilde B_\alpha^\dagger(z_-;\bm \lambda) \tilde B_\alpha^{\vphantom{\dagger}}(w_-;\bm \lambda) \,\omega_E\bigr\} \\[3pt]
    & \overset{\rm (i)}{=} \theta (z-w) \tr\bigl[\tilde B_\alpha^\dagger(z_-;\bm \lambda) \tilde B_\alpha^{\vphantom{\dagger}}(w_-;\bm \lambda) \,\omega_E\bigr] +\zeta \theta(w-z) \tr \bigl[\tilde B_\alpha^{\vphantom{\dagger}}(w_-;\bm \lambda)\tilde B_\alpha^\dagger(z_-;\bm \lambda)  \,\omega_E\bigr] \\[3pt]
    & \overset{\rm (ii)}{=} \theta(z-w)\tr \bigl[e^{+ i\mu_{\alpha} \lambda_{\alpha}/2} B^\dagger_\alpha(z- \lambda_\alpha/2) e^{- i\mu_{\alpha} \lambda_{\alpha}/2} B_\alpha^{\vphantom{\dagger}}(w- \lambda_\alpha/2)\,\omega_E\bigr] \\
    & + \zeta \theta(w-z) \tr \bigl[e^{- i\mu_{\alpha} \lambda_{\alpha}/2} B_\alpha^{\vphantom{\dagger}}(w- \lambda_\alpha/2)e^{+ i\mu_{\alpha} \lambda_{\alpha}/2} B^\dagger_\alpha(z- \lambda_\alpha/2) \,\omega_E\bigr] \\[3pt]
    & \overset{\rm (iii)}{=}  \theta(z-w)\tr \bigl[B_\alpha^\dagger(z- \lambda_\alpha/2-w+ \lambda_\alpha/2) B_\alpha^{\vphantom{\dagger}}(0)\,\omega_E\bigr]\\
    & + \zeta \theta(w-z) \tr \bigl[ B_\alpha^{\vphantom{\dagger}}(w- \lambda_\alpha/2-z+ \lambda_\alpha/2 )B_\alpha^\dagger(0) \,\omega_E\bigr] \\[3pt]
    & = \theta(z-w)\chi_\alpha^-(z-w) +\zeta \theta(w-z)\chi_\alpha^+(w-z)\\[3pt]
    & :=  \C_{\alpha, -}^{\,\tord}(z,w) ~. 
\end{split}
\end{equation}
Here, in (i) we used the definition of contour ordering, in (ii) we plugged in Eqs.~\eqref{sm-eq:tilted-reservoir-operators} for the tilted environment operators and in (iii) we exploited homogeneity of $\chi$ [Eq.~\eqref{eq:chi}] in the time variables. We proceed identically for $\tilde \C_{\alpha, -}^{\,\drot}(z,w)$, concluding that $\tilde \C^{\,\tord} = \C^{\,\tord}$ and $\tilde \C^{\,\drot} = \C^{\,\drot} $. For the greater and lesser components [Eqs.~\eqref{eq:C->}, \eqref{eq:C-<}] we have
\begin{equation}
    \begin{split}
    \tilde \C_{\alpha, -}^{\,>}(z,w; \bm \lambda) 
        & :=  \tr \tord_\zeta \bigl\{\tilde B^\dagger_\alpha(z_+;\bm \lambda) \tilde B_\alpha^{\vphantom{\dagger}}(w_-;\bm \lambda) \,\omega_E\bigr\} \\[3pt]
        & = \tr\bigl[\tilde B_\alpha^{\vphantom{\dagger}}(w_-;\bm \lambda)\,\omega_E\, \tilde B_\alpha^\dagger(z_+;\bm \lambda) \bigr]   \\
        & \overset{\rm (i)}{=} \tr\bigl[\tilde B_\alpha^\dagger(z_+;\bm \lambda)\tilde B_\alpha^{\vphantom{\dagger}}(w_-;\bm \lambda)\,\omega_E \bigr]   \\[3pt]
        & = \tr\bigl[e^{- i\mu_{\alpha} \lambda_{\alpha}/2} B_\alpha^\dagger(z+ \lambda_\alpha/2)e^{-i\mu_{\alpha} \lambda_{\alpha}/2} B_\alpha^{\vphantom{\dagger}}(w- \lambda_\alpha/2)\,\omega_E \bigr]   \\
        & \overset{\rm (ii)}{=} e^{-i\mu_{\alpha} \lambda_{\alpha}}\tr\bigl[B^\dagger_\alpha(z+ \lambda_\alpha/2) B_\alpha^{\vphantom{\dagger}}(w- \lambda_\alpha/2)\,\omega_E \bigr]   \\[3pt]
        & = e^{-i\mu_{\alpha} \lambda_{\alpha}}\tr\bigl[B^\dagger_\alpha(z-w+ \lambda_\alpha) B_\alpha^{\vphantom{\dagger}}(0)\,\omega_E \bigr]   \\[3pt]
        & = e^{-i\mu_{\alpha} \lambda_{\alpha}} \chi_\alpha^-(z-w+\lambda_\alpha)~,
\end{split}
\end{equation}
and
\begin{equation}
    \begin{split}
    \tilde \C_{\alpha, -}^{\,<}(z,w; \bm \lambda) 
        & :=  \tr \tord_\zeta \bigl\{\tilde B^\dagger_\alpha(z_-;\bm \lambda) \tilde B_\alpha^{\vphantom{\dagger}}(w_+;\bm \lambda) \,\omega_E\bigr\} \\[3pt]
        & = \zeta \tr\bigl[\tilde B^\dagger_\alpha(z_-;\bm \lambda) \,\omega_E \,\tilde B_\alpha^{\vphantom{\dagger}}(w_+;\bm \lambda) \bigr]   \\
        & \overset{\rm (i)}{=} \zeta \tr\bigl[\tilde B_\alpha^{\vphantom{\dagger}}(w_+;\bm \lambda) \tilde B_\alpha^\dagger(z_-;\bm \lambda) \,\omega_E  \bigr]     \\[3pt]
        & = \zeta \tr\bigl[e^{+i\mu_{\alpha} \lambda_{\alpha}/2} B_\alpha^{\vphantom{\dagger}}(w+ \lambda_\alpha/2)e^{+i\mu_{\alpha} \lambda_{\alpha}/2} B_\alpha^\dagger(z - \lambda_\alpha/2)\,\omega_E \bigr]   \\
        & \overset{\rm (ii)}{=} \zeta e^{+i\mu_{\alpha} \lambda_{\alpha}}\tr\bigl[B_\alpha^{\vphantom{\dagger}}(w+ \lambda_\alpha/2) B_\alpha^\dagger(z- \lambda_\alpha/2)\,\omega_E \bigr]   \\[3pt]
        & = \zeta e^{+i\mu_{\alpha} \lambda_{\alpha}}\tr\bigl[B_\alpha^{\vphantom{\dagger}}(w-z+ \lambda_\alpha) B_\alpha^\dagger(0)\,\omega_E \bigr]   \\[3pt]
        & = \zeta e^{+i\mu_{\alpha} \lambda_{\alpha}} \chi_\alpha^+(w-z+\lambda_\alpha)~, 
    \end{split}
\end{equation}
using (i) ciclicity and (ii) linearity of the trace operation. For completeness, we obtain the remaining matrix elements by effective conjugation [Def.~\ref{defn:effective-conjugate}]:
\begin{align}
    \tilde \C_{\alpha, +}^{\,>}(z,w; \bm \lambda) & = \zeta \bigl[\tilde \C_{\alpha, -}^{\,<}(z,w; -\bm \lambda)\bigr]^* = \zeta^2 e^{+i\mu_{\alpha} \lambda_{\alpha}} \bigl[\chi_\alpha^+(w-z-\lambda_\alpha)\bigr]^* = e^{+i\mu_{\alpha} \lambda_{\alpha}} \chi_\alpha^+(z-w+\lambda_\alpha)~,\\
    \tilde \C_{\alpha, +}^{\,<}(z,w; \bm \lambda) & = \zeta \bigl[\tilde \C_{\alpha, -}^{\,>}(z,w; -\bm \lambda)\bigr]^* = \zeta e^{-i\mu_{\alpha} \lambda_{\alpha}} \bigl[\chi_\alpha^-(z-w-\lambda_\alpha)\bigr]^* = \zeta e^{-i\mu_{\alpha} \lambda_{\alpha}} \chi_\alpha^-(w-z+\lambda_\alpha)~.
\end{align}
thus concluding the proof.
\end{proof}
\section{Derivation of the tilted Gaussian master equation}
\label{app:B}
In this appendix, we derive the tGME [Eq.~\eqref{eq:tGME}] and the Dyson form [Eq.~\eqref{eq:tilted-dyson}]. To be pragmatic, we will assume the initial state to be separable, i.e. $\rho_{SE}(0) = \rho_S(0) \otimes \omega_E$, although the results we present only rely on the Gaussian hypothesis. A derivation for correlated initial state is a straightforward generalization of the construction employed in Ref.~\cite{d2025exact} for the standard GME. 

By expanding the exponential in Eq.~\eqref{eq:CF-exp} and performing the partial trace on the environment, we obtain
\begin{equation}
    \label{sm-eq:rho-tilde}
    \tilde \rho_S(t;\bm \lambda) = \sum_{n=0}^\infty \frac{(-i)^n}{n!} {\sumint}_{\!\!\!\!\!\gamma(t)} \d^n\bm z \,\tr\bigl[\tord_\zeta \bigl\{\tilde B_1\dots \tilde B_n\,\omega_E\bigr\}\bigr]\,\tord_\zeta \bigl\{A_1\dots A_n\,\rho_S(0)\bigr\}~,
\end{equation}
where we introduced the shorthand notation $A_i := A_{\mu_i}(z_i)$ and $\tilde B_i := B_{\mu_i}(z_i; \bm \lambda)$, and the sum symbol superimposed on the integral is a reminder that we should also sum over $\mu$ indices. Moreover, we made the substitution $\tord \to \tord_\zeta$, which does not change the result as the sum number of transposition is involved in ordering the strings $\tilde B_1\dots \tilde B_n\,\omega_E$ and $A_1\dots A_n\,\rho_S(0)$. Now we make the assumption of a quadratic environment linearly coupled to the system, so that we can invoke the following form of Wick's theorem:
\begin{theorem}[Wick]
Assuming $\tr[\tilde B_i \omega_E] = 0$ $\forall i$, the following holds:
\begin{align}
\label{sm-eq:wick1}
    \tr\bigl[\tord_\zeta \bigl\{\tilde B_1\dots \tilde B_{2m+1}\,\omega_E\bigr\}\bigr] &= 0~,\\
\label{sm-eq:wick2}
    \tr\bigl[\tord_\zeta \bigl\{\tilde B_1\dots \tilde B_{2m}\,\omega_E\bigr\}\bigr] &= \frac{1}{m! 2^m} \sum_{\sigma \in {\mathfrak S}_{2m} } \zeta^{N(\sigma)} \tilde \C_{\sigma(1), \sigma(2)} \cdots \tilde \C_{\sigma(2m-1), \sigma(2m)}~,
\end{align}
where ${\mathfrak S}_{2\nu}$ is the set of permutations of $\{1, \dots, 2\nu\}$, $N(\sigma)$ the parity of the permutation $\sigma$, and $\tilde \C$ is the bare GF defined in Eq.~\eqref{eq:bare-GF} and given in the current notation by
\begin{equation}
    \tilde\C_{ij} = \tr\bigl[\tord_\zeta \{ \tilde B_i \tilde B_j \,\omega_E \}\bigr] = \zeta \tilde \C_{ji}~.
\end{equation} 
\end{theorem}
Using Eqs.~\eqref{sm-eq:wick1}-\eqref{sm-eq:wick2} in Eq.~\eqref{sm-eq:rho-tilde}, we get
\begin{equation}
    \label{sm-eq:rho-tilde2}
    \begin{split}
        \tilde \rho_S(t;\bm \lambda) &= \sum_{m=0}^\infty \frac{(-i)^{2m}}{(2m)!m!2^m} \sum_{\sigma \in {\mathfrak S}_{2m} } \zeta^{N(\sigma)} {\sumint}_{\!\!\!\!\!\gamma(t)} \d^n \bm z \,\tilde \C_{\sigma(1), \sigma(2)} \cdots \tilde \C_{\sigma(2m-1), \sigma(2m)}\,\tord_\zeta \bigl\{A_1\dots A_{2m}\,\rho_S(0)\bigr\} \\
        & \overset{\rm(i)}{=}  \sum_{m=0}^\infty \frac{(-i)^{2m}}{(2m)!m!2^m} \sum_{\sigma \in {\mathfrak S}_{2m} } \zeta^{N(\sigma)} {\sumint}_{\!\!\!\!\!\gamma(t)} \d^n \bm w \,\tilde \C_{1,2} \cdots \tilde \C_{2m-1, 2m}\,\tord_\zeta \bigl\{A_{\sigma^{-1}(1)}\dots A_{\sigma^{-1}(2m)}\,\rho_S(0)\bigr\} \\
        & \overset{\rm(ii)}{=}  \sum_{m=0}^\infty \frac{(-i)^{2m}}{(2m)!m!2^m} \sum_{\sigma \in {\mathfrak S}_{2m} } \zeta^{N(\sigma)} \zeta^{N(\sigma^{-1})} {\sumint}_{\!\!\!\!\!\gamma(t)} \d^n \bm w \,\tilde \C_{1,2} \cdots \tilde \C_{2m-1, 2m}\,\tord_\zeta \bigl\{A_{1}\dots A_{2m}\,\rho_S(0)\bigr\}\\
        & \overset{\rm(iii)}{=}  \sum_{m=0}^\infty \frac{(-i)^{2m}}{m!2^m} {\sumint}_{\!\!\!\!\!\gamma(t)} \d^n \bm w \,\tilde \C_{1,2} \cdots \tilde \C_{2m-1, 2m}\,\tord_\zeta \bigl\{A_{1}\dots A_{2m}\,\rho_S(0)\bigr\}~.
    \end{split}
\end{equation}
Here, in (i) we changed the integration variables to $w_i = z_{\sigma(i)}$ and $\nu_i = \mu_{\sigma(i)}$, in (ii) we used the simple relation
\begin{equation}
    \tord_\zeta \bigl\{A_{\sigma^{-1}(1)}\dots A_{\sigma^{-1}(2m)}\,\rho_S(0)\bigr\}  = \zeta^{N(\sigma^{-1})} \tord_\zeta \bigl\{A_1\dots A_{2m}\,\rho_S(0)\bigr\} ~,
\end{equation}
and in (iii), noting that nothing depends on $\sigma$ anymore, we performed the sum on $\sigma \in {\mathfrak S}_{2m}$, yielding the cardinality of the set ${\mathfrak S}_{2m}$, i.e. $(2m)!$. This result is conveniently presented as
\begin{equation}
\label{sm-eq:tilde-rho-M}
    \tilde\rho_S(t; \bm \lambda) = \sum_{m=0}^{\infty} \frac{(-i)^{2m}}{m!2^m} \tilde M_m(t; \bm \lambda)~, \qquad \tilde M_m(t; \bm \lambda) := {\sumint}_{\!\!\!\!\!\gamma(t)} \d^n \bm w \,\tilde \C_{1,2} \cdots \tilde \C_{2m-1, 2m}\,\tord_\zeta \bigl\{A_{1}\dots A_{2m}\,\rho_S(0)\bigr\}~.
\end{equation}
In order to find an exact effective master equation for $\tilde\rho_S(t; \bm \lambda)$, we need to take the time derivative of Eq.~\eqref{sm-eq:tilde-rho-M}. For such purpose, we state the following:
\begin{theorem}[Fundamental theorem of calculus on the Keldysh contour]
Let $f(\bm z)$ be an operator-valued function on $[\gamma(t)]^n$, with $n \in \mathbb N^+$. Then,
\begin{equation}
    \label{sm-eq:fundamental-theorem}
    \frac{\d}{\d t} \int_{\gamma(t)} \d^n \bm z\,f(z) = \sum_{k=1}^n \int_{\gamma(t)} \d ^{n-1} \bm w\,\bigl[ f(\bm w_1^{k-1}, t_-, \bm w_k^{n-1}) - f(\bm w_1^{k-1}, t_+, \bm w_k^{n-1})\bigr]~,
\end{equation}
where $\bm w_i^j := (w_i, \dots, w_j)$ with $i,j \in (1, \dots, n)$ and $i \leq j$.
\end{theorem}
A proof of this statement proceeds by induction and is carried out in the Supplemental Material of Ref.~\cite{d2025exact}. 
\begin{corollary}
Let $f(\bm z)$ be an operator-valued function on $[\gamma(t)]^n$, with $n \in \mathbb N^+$. If
\begin{equation}
\label{sm-eq:condition}
    \int_{\gamma(t)} \d^{n-1} \bm w\,f(\bm w_{1}^{k-1}, t_{\pm}, \bm w_k^{n-1}) =  \int_{\gamma(t)} \d^{n-1} \bm w\,f(t_{\pm}, \bm w) 
\end{equation}
holds for every $k \in \{1, \dots, n\}$, then
\begin{equation}
    \label{sm-eq:fundamental-theorem-2}
    \frac{\d}{\d t} \int_{\gamma(t)} \d^n \bm z\,f(z) = n \int_{\gamma(t)} \d ^{n-1} \bm w\,\bigl[ f(t_-, \bm w) - f(t_+, \bm w)\bigr]~.
\end{equation}
\end{corollary}
Looking at Eq.~\eqref{sm-eq:tilde-rho-M}, we see that the operator-valued function to consider is
\begin{equation}
    f(\bm z) = \sum_{\{\mu\}} \tilde \C_{1,2} \cdots \tilde \C_{2m-1, 2m}\,\tord_\zeta \bigl\{A_{1}\dots A_{2m}\,\rho_S(0)\bigr\}~,
\end{equation}
where, recalling the shorthand notation introduced in Eq. \eqref{sm-eq:rho-tilde}, we explicitly wrote the sum over the indices $\mu$. We need to prove that this $f(\bm z)$ satisfies condition~\eqref{sm-eq:condition}. To achieve such a task, it is useful to introduce the shorthand notation
\begin{equation}
    A_{\underline{i}}^{\pm} := A_{\mu_i}(t_\pm)~,
\end{equation}
so that an underline index indicates evaluation in $t_\pm$. Similarly, we define \footnote{The label $\pm$ introduced below is not to be confused with the notation established in Eq.~\eqref{sm-eq:C-alpha-block}. While $ \tilde{\C}_{i,\underline{j}}^\pm$, $ \tilde{\C}_{\underline{i},j}^\pm$ denote generic correlators with at least one argument lying in $t_{\pm}$, the correlators $\tilde\C_{\alpha, -}$ and  $\tilde\C_{\alpha, +}$ are specifically referring to absorption or emission processes and their arguments on the contour are generic.}
\begin{equation}
    \tilde \C_{\underline{i},j}^\pm := \tilde \C_{\mu_i, \mu_j}(t_\pm, z_j)~, \qquad \quad
    \tilde \C_{i, \underline{j}}^\pm := \tilde \C_{\mu_i, \mu_j}(z_i, t_\pm)~.
\end{equation}
Now, we fix $k \in \{1, \dots, 2m\}$ and let $\bm w = (w_1, \dots, w_{k-1}, w_{k+1}, \dots w_{2m})$. For $k$ odd, the LHS of Eq.~\eqref{sm-eq:condition} reads
\begin{equation}
    {\sumint}_{\!\!\!\!\!\gamma(t)} \d^{2m-1} \bm w\, 
    \tilde \C_{1,2} \cdots \tilde \C_{k-2,k-1} \tilde \C_{\underline{k},k+1}^{\pm}  \tilde \C_{k+2,k+3} \dots \tilde \C_{2m-1, 2m}
    \tord_\zeta \bigl\{A_{1}\dots A_{k-1}A_{\underline{k}}^\pm A_{k+1} \dots A_{2m}\,\rho_S(0)\bigr\}~.
\end{equation}
With $k-1$ transpositions we bring $A_{\underline{k}}^\pm$ at the beginning of the string, and with another $k-1$ transpositions we move $A_{k+1}$ to the right of $A_{\underline{k}}^\pm$. Having performed $2(k-1)$ transpositions in total, no additional sign appears:
\begin{equation}
    {\sumint}_{\!\!\!\!\!\gamma(t)} \d^{2m-1} \bm w\, 
    \tilde \C_{\underline{k},k+1}^{\pm}\tilde \C_{1,2} \cdots \tilde \C_{k-2,k-1}   \tilde \C_{k+2,k+3} \dots \tilde \C_{2m-1, 2m}
    \tord_\zeta \bigl\{A_{\underline{k}}^\pm A_{k+1}A_{1}\dots A_{k-1}A_{k+2} \dots A_{2m}\,\rho_S(0)\bigr\}~.
\end{equation}
The latter expression is equivalent, after an appropriate change of variables, to the RHS of Eq.~\eqref{sm-eq:condition}. A similar argument can be carried out for even $k$. This time, the LHS of Eq.~\eqref{sm-eq:condition} reads
\begin{equation}
    {\sumint}_{\!\!\!\!\!\gamma(t)} \d^{2m-1} \bm w\, 
    \tilde \C_{1,2} \cdots \tilde \C_{k-3,k-2} \tilde \C_{k-1,\underline{k}}^{\pm}  \tilde \C_{k+1,k+1} \dots \tilde \C_{2m-1, 2m}
    \tord_\zeta \bigl\{A_{1}\dots A_{k-1}A_{\underline{k}}^\pm A_{k+1} \dots A_{2m}\,\rho_S(0)\bigr\}~.
\end{equation}
We need to perform $k-1$ transpositions to bring $A_{\underline{k}}^\pm$ to the beginning of the string, and another $k-2$ to move $A_{k-1}$ to its right. A total of $2k-3$ transpositions, which is odd, causes a $\zeta$ factor to appear. However, we can compensate this by writing $\tilde \C_{k-1,\underline{k}}^{\pm} = \zeta\tilde \C_{\underline{k}, k-1}^{\pm}$, resulting in
\begin{equation}
    {\sumint}_{\!\!\!\!\!\gamma(t)} \d^{2m-1} \bm w\, 
    \tilde \C_{\underline{k},k-1}^{\pm}\tilde \C_{1,2} \cdots \tilde \C_{k-3,k-2}   \tilde \C_{k+1,k+2} \dots \tilde \C_{2m-1, 2m}
    \tord_\zeta \bigl\{A_{\underline{k}}^\pm A_{k-1}A_{1}\dots A_{k-2}  A_{k+1} \dots A_{2m}\,\rho_S(0)\bigr\}~.
\end{equation}
We conclude that Eq.~\eqref{sm-eq:fundamental-theorem-2} can be used to evaluate the time derivative of $\tilde M_m(t; \bm \lambda)$ in Eq.~\eqref{sm-eq:tilde-rho-M}. Note that
\begin{align}
    \tord_\zeta\{A^-_{\underline{1}} A_2 \cdots A_{2m} \rho_S(0)\} & = A_{\underline{1}} \tord_\zeta \{A_2 \cdots A_{2m} \rho_S(0)\}~,\\
    \tord_\zeta\{A^+_{\underline{1}} A_2 \cdots A_{2m} \rho_S(0)\} & = \tord_\zeta \{A_2 \cdots A_{2m} \rho_S(0)\}A_{\underline{1}} ~.
\end{align}
Therefore,
\begin{equation}
\label{sm-eq:dMdt}
    \begin{split}
        \frac{\d \tilde M_m(t; \bm \lambda)}{\d t} = 2m  \,{\sumint}_{\!\!\!\!\!\gamma(t)} \d ^{2m-1}\bm z\,&\Bigl[\tilde \C_{\underline{1},2}^- \cdots \tilde \C_{2m-1, 2m}\,A_{\underline{1}} \tord_\zeta \{A_2 \cdots A_{2m} \rho_S(0)\}\\
        & \!\!\!\! - \tilde \C_{\underline{1},2}^+ \cdots \tilde \C_{2m-1, 2m}\,\tord_\zeta \{A_2 \cdots A_{2m} \rho_S(0)\} A_{\underline{1}}\Bigr]~.
    \end{split}
\end{equation}
At this point, we introduce the quadratic assumption, that will allow us to expand the RHS of Eq.~\eqref{sm-eq:dMdt} and make $M_k(t;\bm \lambda)$ with $k\leq m$ appear, yielding a recursive form. As a preliminary step, we state the following lemma, proven in the Supplemental Material of Ref.~\cite{d2025exact}:
\begin{lemma}
\label{lemma:3}
Suppose $[A_i, A_j]_\zeta$ is a c-number $\forall \,i, j$. Then
\begin{equation}
    \label{sm-eq:lemma3}
     \tord_\zeta\{A_0A_1  \cdots A_{2m} \rho_S(0)\} = \overline{A_0\tord_\zeta\{A_1  \cdots A_{2m} \rho_S(0)\}} - \sum_{j=1}^{2m} \zeta^{j-1} \Sigma_{0,j} \tord_\zeta\{A_1  \cdots A_{j-1} A_{j+1}\cdots A_{2m} \rho_S(0)\}~,
\end{equation}
where
\begin{equation}
    \overline{ A(z)\,\bullet} = 
    \begin{cases}
        A(z) \,\bullet & {\rm if} \quad z \in \gamma_-(t)~,\\
        \zeta \bullet A(z)  & {\rm if} \quad z \in \gamma_+(t)~,
    \end{cases}
\end{equation}
and we introduced the system “self-energy"
\begin{equation}
\label{sm-eq:self-energy}
    \Sigma_{\mu \nu}(z, w) : =
    \begin{cases}
        +\theta_{z \prec w} \bigl[A_\mu(z), A_{\nu}(w)\bigr]_\zeta & {\rm if} \quad z \in \gamma_-(t)~,\\
        -\theta_{z \succ w} \bigl[A_\mu(z), A_{\nu}(w)\bigr]_\zeta & {\rm if} \quad z \in \gamma_+(t)~.
    \end{cases}
\end{equation}
\end{lemma}
Substituting Eq.~\eqref{sm-eq:lemma3} into Eq.~\eqref{sm-eq:dMdt} yields
\begin{equation}
\label{sm-eq:dMdt-lunga}
    \begin{split}
        \frac{\d \tilde M_m(t; \bm \lambda)}{\d t} &= 2m  \,{\sumint}_{\!\!\!\!\!\gamma(t)} \d ^{2m-1}\bm z\,\tilde \C_{\underline{1},2}^- \cdots \tilde \C_{2m-1, 2m} \,A_{\underline{1}} \overline{A_2 \tord_\zeta \{\cdots A_{2m} \rho_S(0)\}}\\
         &- 2m  \,{\sumint}_{\!\!\!\!\!\gamma(t)} \d ^{2m-1}\bm z\,\tilde \C_{\underline{1},2}^+ \cdots \tilde \C_{2m-1, 2m} \,\overline{A_2 \tord_\zeta \{\cdots A_{2m} \rho_S(0)\}}A_{\underline{1}} \\
        & - 2m \sum_{j=3}^{2m} \zeta^{j-1} \,{\sumint}_{\!\!\!\!\!\gamma(t)} \d ^{2m-1}\bm z\,\tilde \C_{\underline{1},2}^- \cdots \tilde \C_{2m-1, 2m} \Sigma_{2,j}\,A_{\underline{1}} \tord_\zeta\{A_3  \cdots A_{j-1} A_{j+1}\cdots A_{2m} \rho_S(0)\}\\
        & + 2m \sum_{j=3}^{2m} \zeta^{j-1} \,{\sumint}_{\!\!\!\!\!\gamma(t)} \d ^{2m-1}\bm z\,\tilde \C_{\underline{1},2}^+ \cdots \tilde \C_{2m-1, 2m} \Sigma_{2,j}\tord_\zeta\{A_3  \cdots A_{j-1} A_{j+1}\cdots A_{2m} \rho_S(0)\} A_{\underline{1}} ~.
    \end{split}
\end{equation}
In the first and second terms of the RHS, we see $\tilde M_{m-1}(t;\bm \lambda)$ appear if we isolate the variables $z_3, \dots, z_{2m}$. For what concerns the third and fourth terms, we will perform the sum on $j$ by showing that all terms in the sum are equal to each other. Let us consider the third term in the case of even $j$. We have
\begin{equation}
    \begin{split}
        \zeta^{j-1} &\,{\sumint}_{\!\!\!\!\!\gamma(t)} \d ^{2m-1}\bm z\,\tilde \C_{\underline{1},2}^-\cdots  \tilde \C_{j-1, j} \cdots \tilde \C_{2m-1, 2m} \Sigma_{2,j}\,A_{\underline{1}} \tord_\zeta\{A_3 A_4 \cdots A_{j-1} A_{j+1}\cdots A_{2m} \rho_S(0)\} \\
        & =\zeta \,{\sumint}_{\!\!\!\!\!\gamma(t)} \d ^{2m-1}\bm z\,\tilde \C_{\underline{1},2}^-  \tilde\C_{j, j-1} \cdots \tilde\C_{4, 3} \cdots  \tilde \C_{2m-1, 2m} \Sigma_{2,3}\,A_{\underline{1}} \tord_\zeta\{A_j A_{j-1} A_5 \cdots A_{j-2} A_4 A_{j+1}\cdots A_{2m} \rho_S(0)\} ~,
    \end{split}
\end{equation}
where the variable changes $j \leftrightarrow3$ and $j-1 \leftrightarrow 4$ were performed. We now move $A_4$ to the leftmost position in the string using $j-4$ transpositions. Then, we move the product with $A_jA_{j-1}$ to the former position of $A_4$---without accounting for the permutation parity since $A_jA_{j-1}$ is the product of an even number of operators---and we finally switch $A_j$ with $A_{j-1}$, resulting in $j-3$ total transpositions, which is odd. Therefore, using the fact that $\tilde \C_{j,j-1} \tilde \C_{4,3} = \tilde \C_{j-1, j} \tilde \C_{3,4}$, we obtain with
\begin{equation}
     {\sumint}_{\!\!\!\!\!\gamma(t)} \d ^{2m-1}\bm z\,\tilde \C_{\underline{1},2}^-\tilde \C_{3, 4} \cdots\tilde \C_{j-1, j} \cdots   \tilde \C_{2m-1, 2m} \Sigma_{2,3}\,A_{\underline{1}} \tord_\zeta\{A_4 A_5 \cdots A_{j-1} A_j A_{j+1}\cdots A_{2m} \rho_S(0)\} ~,
\end{equation}
which we can express without specifying $j$. A similar procedure can be followed for the case of odd $j$ and for the fourth term on the RHS of Eq.~\eqref{sm-eq:dMdt-lunga}. As a consequence, we get
\begin{equation}
\label{sm-eq:dMdt-lunga-2}
    \begin{split}
        \frac{\d \tilde M_m(t; \bm \lambda)}{\d t} &= 2m  \,{\sumint}_{\!\!\!\!\!\gamma(t)} \d z_2\,\Bigl[\tilde \C_{\underline{1},2}^- A_{\underline{1}} \overline{A_2 \tilde M_{m-1}(t; \bm \lambda)}-\tilde \C_{\underline{1},2}^+ \overline{A_2  \tilde M_{m-1}(t; \bm \lambda)}A_{\underline{1}} \Bigr] \\
        & - 2m (2m-2)\,{\sumint}_{\!\!\!\!\!\gamma(t)} \d ^{2m-1}\bm z\,\tilde \C_{\underline{1},2}^- \Sigma_{2,3}\tilde \C_{3,4} \cdots \tilde \C_{2m-1, 2m} \,A_{\underline{1}} \tord_\zeta\{A_4 \cdots A_{2m} \rho_S(0)\}\\
        & + 2m (2m-2)\,{\sumint}_{\!\!\!\!\!\gamma(t)} \d ^{2m-1}\bm z\,\tilde \C_{\underline{1},2}^+ \Sigma_{2,3}\tilde \C_{3,4} \cdots \tilde \C_{2m-1, 2m} \,\tord_\zeta\{A_4 \cdots A_{2m} \rho_S(0)\}A_{\underline{1}} ~.
    \end{split}
\end{equation}
The above reasoning can be applied again to reduce the string $\tord_\zeta\{A_4 \cdots A_{2m} \rho_S(0)\}A_{\underline{1}}$, generating a term containing $\tilde M_{m-2}(t; \bm \lambda)$ and a term containing $\tord_\zeta\{A_6 \cdots A_{2m} \rho_S(0)\}A_{\underline{1}}$. We iterate this pattern until we exhaust all operators inside $\tord_\zeta$. The result is 
\begin{equation}
     \frac{\d \tilde M_m(t; \bm \lambda)}{\d t} = \sum_{k=1}^m \frac{(-1)^{k+1} (2m)!!}{(2m-2k)!!} \, \tilde{\cal S}_k(t; \bm \lambda) \tilde M_{m-k}(t; \bm \lambda)~,
\end{equation}
where we defined the superoperator
\begin{equation}
\label{sm-eq:kernel}
    \tilde{\cal S}_k(t; \bm \lambda) \,\bullet := {\sumint}_{\!\!\!\!\!\gamma(t)} \d ^{2k-1}\bm z\,\bigl[\tilde \C_{\underline{1},2}^- A_{\underline{1}} \overline{A_{2k} \,\bullet } \, - \tilde \C_{\underline{1},2}^+ \overline{A_{2k} \,\bullet }\,A_{\underline{1}}  \bigr]\,\Sigma_{2,3}\tilde \C_{3,4} \Sigma_{4,5} \cdots \Sigma_{2k-2, 2k-1}\tilde \C_{2k-1, 2k} ~.
\end{equation}
If we recall that $(2x)!! = x! 2^x$ and Eq.~\eqref{sm-eq:tilde-rho-M}, we get
\begin{equation}
    \frac{\d \tilde \rho_S(t; \bm \lambda)}{\d t} = \sum_{m=1}^\infty \frac{(-1)^m}{m! 2^m}  \frac{\d \tilde M_m(t; \bm \lambda)}{\d t} = \sum_{m=1}^\infty \sum_{k=1}^m\frac{(-1)^{m+k+1}}{(m-k)! 2^{(m-k)}}  \tilde{\cal S}_k(t; \bm \lambda) \tilde M_{m-k}(t; \bm \lambda)~,
\end{equation}
and recognizing the Cauchy product of two series
\begin{equation}
    \sum_{m=1}^\infty \sum_{k=1}^m F(m,k) = \sum_{k=1}^\infty \sum_{m=k}^\infty F(m,k) = \sum_{k=1}^\infty \sum_{m=0}^\infty F(m+k,k)~,
\end{equation}
we finally obtain
\begin{equation}
\label{sm-eq:tGME-implicit}
    \frac{\d \tilde \rho_S(t; \bm \lambda)}{\d t} = \sum_{k=1}^\infty \sum_{m=0}^\infty \frac{(-1)^{m+1}}{m! 2^m} \tilde{\cal S}_k(t; \bm \lambda) \tilde M_{m}(t; \bm \lambda) = - \tilde{\cal S}(t; \bm \lambda) \tilde \rho_S(t; \bm \lambda)~,
\end{equation}
where $\tilde {\cal S} := \sum_{k=1}^{\infty}\tilde {\cal S}_k $. We'll now show that the latter is exactly the 
effective master equation reported in Eq.~\eqref{eq:tGME} of the main text. First, we define the following recursion:
\begin{equation}
    \begin{cases}
        \tilde\G^{(1)}_{\mu\nu}(z,w; \bm \lambda) = \tilde \C_{\mu\nu}(z,w; \bm \lambda)~, \\[3pt]
        \tilde\G^{(k)}_{\mu\nu}(z,w; \bm \lambda) = \sum_{\mu',\nu'}\int_{\gamma(t)} \d^2 \bm y\, \tilde\G^{(k-1)}_{\mu\mu'}(z,y_1; \bm \lambda) \Sigma_{\mu'\nu'}(y_1, y_2) \tilde \C_{\nu'\nu}(y_2, w; \bm \lambda)~,
    \end{cases}
\end{equation}
so that we can rewrite the kernel of Eq~\eqref{sm-eq:tGME-implicit} as
\begin{equation}
\label{sm-eq:kernel2}
    \tilde{\cal S}(t; \bm \lambda) \,\bullet = \sum_{\mu,\nu}\int_{\gamma(t)}\d z\, \bigl[\tilde \G_{\mu\nu}(t_-, z; \bm \lambda) A_\mu(t) \overline{A_{\nu}(z) \,\bullet } \, - \tilde \G (t_+, z; \bm \lambda) \overline{A_\nu(z) \,\bullet }\,A_\mu(t)  \bigr] ~,
\end{equation}
where performing the sum on $k$ yields Eq.~\eqref{eq:tilted-dyson} for the tilted GF:
\begin{equation}
\label{sm-eq:tilted-GF-proof}
    \begin{split}
        \tilde \G (12)
        & := \sum_{k=1}^{\infty} \tilde \G^{(k)}(12)\\
        & = \tilde \C(12) + \int_{\gamma(t)}  \d(34) \,\tilde \C(13) \Sigma(34) \tilde \C(42) + \int_{\gamma(t)} \d(34) \int_{\gamma(t)}  \d(56) \,\tilde \C(15) \Sigma(56) \tilde \C(63)  \Sigma(34) \tilde \C(42) + \cdots \\
        & = \tilde \C(12) + \int_{\gamma(t)}  \d(34) \biggl[\tilde \C(13) + \int_{\gamma(t)}  \d(56) \,\tilde \C(15) \Sigma(56) \tilde \C(63)  + \cdots \biggr]\Sigma(34) \tilde \C(42) \\
        & = \tilde \C(12) + \int_{\gamma(t)}  \d(34)\, \tilde \G(13)\Sigma(34) \tilde \C(42) ~.
    \end{split}
\end{equation}
Inserting Eq.~\eqref{sm-eq:kernel2} into Eq.~\eqref{sm-eq:tGME-implicit} we recover the explicit form of the tGME [Eq.~\eqref{eq:tGME}]:
\begin{equation}
\label{sm-eq:tGME}
    \begin{split}
        \frac{\d \tilde\rho_S(t; \bm \lambda)}{\d t} 
         = \sum_{\mu,\nu}\int_{\gamma(t)} \d z\, \Bigl\{
            \tilde\G_{\mu \nu}(t_+, z; \bm \lambda) \overline{ A_\nu(z) \tilde\rho_S(t; \bm \lambda) } A_\mu(t)  
            - \tilde\G_{\mu \nu}(t_-, z; \bm \lambda) A_\mu(t) \overline{A_\nu(z) \tilde\rho_S(t; \bm \lambda)}\Bigr\} ~.
    \end{split}
\end{equation}
We observe that, setting $\bm \lambda = \bm 0$ from the very beginning, we can prove the standard GME given in Eq.~\eqref{eq:gme} of the main text.
\section{Derivation of the exact heat transport equation}
\label{app:C}
We are now prepared to derive the transport equation~\eqref{eq:transport-equation}. We start by writing the tGME~\eqref{eq:tGME} in physical time, introducing the Keldysh components of $\tilde \G$:
\begin{equation}
\label{sm-eq:tGME-physical}
        \begin{split}
        \frac{\d \tilde\rho_S(t; \bm \lambda)}{\d t} 
        & = \sum_{\mu, \nu} \int_{\gamma(t)} \d z\, \Bigl\{
            \tilde\G_{\mu \nu}(t_+, z; \bm \lambda) \overline{ A_\nu(z) \tilde\rho_S(t; \bm \lambda) } A_\mu(t) 
            - \tilde\G_{\mu \nu}(t_-, z; \bm \lambda) A_\mu(t) \overline{A_\nu(z) \tilde\rho_S(t; \bm \lambda)}\Bigr\} \\[3pt]
        & = \sum_{\mu, \nu} \Bigg[\int_{\gamma_-(t)} \d z\, \Bigl\{
            \tilde\G_{\mu \nu}(t_+, z_-; \bm \lambda) A_\nu(z) \tilde\rho_S(t; \bm \lambda) A_\mu(t) 
            - \tilde\G_{\mu \nu}(t_-, z_-; \bm \lambda) A_\mu(t) A_\nu(z) \tilde\rho_S(t; \bm \lambda) \Bigr\} \\
            & \phantom{\sum_{\mu, \nu} \Bigg[} +\zeta\int_{\gamma_+(t)} \d z\, \Bigl\{
                \tilde\G_{\mu \nu}(t_+, z_+; \bm \lambda) \tilde\rho_S(t; \bm \lambda) A_\nu(z) A_\mu(t) 
                - \tilde\G_{\mu \nu}(t_-, z_+; \bm \lambda) A_\mu(t)  \tilde\rho_S(t; \bm \lambda)A_\nu(z) \Bigr\} \Bigg]\\[3pt]
        & \overset{\rm (i)}{=} \sum_{\mu, \nu} \int_{0}^{t} \d \tau\, \Bigl\{
            \tilde\G^{\,>}_{\mu \nu}(t, \tau; \bm \lambda) A_\nu(\tau) \tilde\rho_S(t; \bm \lambda) A_\mu(t)  -
                 \tilde\G^{\,\tord}_{\mu \nu}(t, \tau; \bm \lambda) A_\mu(t) A_\nu(\tau) \tilde\rho_S(t; \bm \lambda)\\
            & \phantom{\sum_{\mu, \nu} \int_{0}^{t} \d \tau\, \Bigl\{}  +\zeta\tilde\G^{\,<}_{\mu \nu}(t, \tau; \bm \lambda)A_\mu(t)  \tilde\rho_S(t; \bm \lambda)A_\nu(\tau) - \zeta\tilde\G^{\,\drot}_{\mu \nu}(t, \tau; \bm \lambda) \tilde\rho_S(t; \bm \lambda) A_\nu(\tau) A_\mu(t)  
                 \Bigr\} \\[3pt]
    \end{split}
\end{equation}
In (i), we wrote the explicit form of the contour integral on the two branches, using the correct integration measure $\pm \d \tau$ on $\gamma_{\mp}(t)$ respectively.
Now, recalling definition~\eqref{eq:goth-g} of the dressed heat kernel
\begin{equation}
\label{sm-eq:goth-g}
    \g_\alpha (z,w) : = i\frac{\partial\tilde \G(z, w; \bm \lambda )}{\partial \lambda_\alpha} \bigg|_{\bm \lambda = \bm 0}~,
\end{equation}
we take the derivative of both sides of Eq.~\eqref{sm-eq:tGME-physical} with respect to $\lambda_\alpha$ in $\bm \lambda = \bm 0$, defining $\theta_\alpha(t) := i\partial \tilde\rho_S(t; \bm \lambda)/\partial \lambda_\alpha \big|_{\bm \lambda = \bm 0} $ such that $I_{Q}^{(\alpha)}(t) = \tr_S \theta_\alpha(t)$:
\begin{equation}
\label{sm-eq:ode-for-theta}
    \begin{split}
        \frac{\d \theta_\alpha(t)}{\d t}  
        & = \sum_{\mu, \nu} \int_{0}^{t} \d \tau\, \Bigl\{
            \g^>_{\alpha;\mu\nu}(t, \tau) A_\nu(\tau) \rho_S(t) A_\mu(t) + i\G^{\,>}_{\mu \nu}(t, \tau) A_\nu(\tau)  \theta_\alpha(t) A_\mu(t)  \\
            &\phantom{=\sum_{\mu, \nu} \int_{0}^{t} \d \tau\, }- \g^\tord_{\alpha;\mu\nu}(t, \tau) A_\mu(t) A_\nu(\tau) \rho_S(t) -
                 i\G^{\,\tord}_{\mu \nu}(t, \tau) A_\mu(t) A_\nu(\tau)\theta_\alpha(t) \\
            & \phantom{=\sum_{\mu, \nu} \int_{0}^{t} \d \tau\, }  +\zeta\g^<_{\alpha;\mu\nu}(t, \tau)A_\mu(t) \rho_S(t)A_\nu(\tau)+i\zeta \G_{\mu\nu}^{<}(t, \tau) A_\mu(t) \theta_\alpha(t) A_\nu(\tau)\\
            &\phantom{=\sum_{\mu, \nu} \int_{0}^{t} \d \tau\, } -\zeta\g^\drot_{\alpha;\mu\nu}(t, \tau)\rho_S(t) A_\nu(\tau) A_\mu(t) -i\zeta\G_{\mu\nu}^{\drot}(t, \tau)   \theta_\alpha(t) A_\nu(\tau) A_\mu(t)\Bigr\}~. 
    \end{split}
\end{equation}
We notice that all the terms containing $\G$ exactly sum up to the standard GME kernel, i.e.
\begin{equation}
    \frac{\d \theta_\alpha(t) }{\d t}  = i\sum_{\mu, \nu} \int_{\gamma(t)} \d z \,\G_{\mu\nu}(t_+, z) \left[\,\overline{A_\nu(\tau) \theta_\alpha(t)}, A_\mu(t)\right] + \Bigl[ \text{ terms with $\g$ }\Bigr]~.
\end{equation}
Since the trace of a commutator is zero, we need only to retain the terms containing the dressed heat kernel when evaluating the heat current [Eq.~\eqref{eq:heat-exchange}]. Therefore, we have
\begin{equation}
\label{sm-eq:transport-equation-lunga}
    \begin{split}
        I_{Q}^{(\alpha)}(t)
        & = \tr_S\sum_{\mu, \nu} \int_{0}^{t} \d \tau\, \Bigl\{
            \g^>_{\alpha;\mu\nu}(t, \tau) A_\nu(\tau) \rho_S(t) A_\mu(t) - \g^\tord_{\alpha;\mu\nu}(t, \tau) A_\mu(t) A_\nu(\tau) \rho_S(t) \\
            & \phantom{i\tr\sum_{\mu, \nu} \int_{0}^{t} \d \tau\, \Bigl\{}  +\zeta\g^<_{\alpha;\mu\nu}(t, \tau)A_\mu(t) \rho_S(t)A_\nu(\tau) -\zeta\g^\drot_{\alpha;\mu\nu}(t, \tau)\rho_S(t) A_\nu(\tau) A_\mu(t)\Bigr\}\\
            & \overset{\rm (i)}{=} \sum_{\mu,\nu}\int_{0}^{t} \d \tau\, \Bigl\{
            \bigl[\g^>_{\alpha;\mu\nu}(t, \tau) - \g^\tord_{\alpha;\mu\nu}(t, \tau) \bigr] \ev{A_\mu(t)A_\nu(\tau)}_t +\zeta\bigl[\g^<_{\alpha;\mu\nu}(t, \tau) - \g^\drot_{\alpha;\mu\nu}(t, \tau) \bigr] \ev{A_\nu(\tau)A_\mu(t)}_t \Bigr\} ~,
    \end{split}
\end{equation}
where (i) follows from linearity and ciclicity of the trace and we employed the shorthand notation $\ev{\bullet}_t := \tr_S[\,\bullet \rho_S(t)]$.

Eq.~\eqref{sm-eq:transport-equation-lunga} contains all four Keldysh components of $\g_\alpha$. To get to the more compact form given in Eq.~\eqref{eq:transport-equation} in the main text, we need to show that the second term on the RHS of~\eqref{sm-eq:transport-equation-lunga} is just the complex conjugate of the first. To achieve this, we need to state the generalization of Lemma~\ref{lemma:1} to the dressed GF:
\begin{lemma}
\label{lemma:4}
For every pair of indices $\mu, \nu$ and for every $z, w \in \gamma(t)$,
\begin{equation}
\label{sm-eq:lemma4}
    \bigl[\tilde\G_{\mu \nu}(z, w; \bm \lambda)\bigr]^* = \zeta \tilde \G_{\bar \mu \bar \nu} (z^*, w^*;-\bm \lambda)~.
\end{equation}
\end{lemma}
\begin{proof}[Proof of Lemma~\ref{lemma:4}] It suffices to show that property~\eqref{sm-eq:lemma4} is satisfied by all terms $\tilde \G^{(k)}$ in the series~\eqref{sm-eq:tilted-GF-proof}. We do this by induction on $k$. The $k=1$ case is just Lemma~\ref{lemma:1}. Similarly, we can prove that the same property holds for the self-energy. For $z \in \gamma_-(t)$, we have
\begin{equation}
    \begin{split}
        \Sigma_{\mu\nu}^* (z,w) 
        & := \theta_{z \prec w} [A_\mu(z), A_\nu(w)]_\zeta^*  \\
        & \overset{\rm (i)}{=} \theta_{z \prec w} [A_\nu^\dagger(w), A^\dagger_\mu(z)]_\zeta 
        = \theta_{z \prec w} [A_{\bar \nu}(w^*), A_{\bar \mu}(z^*)]_\zeta  \\
        & = -\zeta\theta_{z \prec w} [A_{\bar \mu}(z^*), A_{\bar \nu}(w^*)]_\zeta 
        = -\zeta\theta_{z^* \succ w^*} [A_{\bar \mu}(z^*), A_{\bar \nu}(w^*)]_\zeta \\
        & := \zeta\Sigma_{\bar \mu \bar \nu}(z^*, w^*)~,
    \end{split}
\end{equation}
and the same holds for $z \in \gamma_+(t)$ (in (i), we recall that $\Sigma_{\mu\nu}$ is a $c$-number $\forall \mu, \nu$). Thus, the inductive step of the proof is
\begin{equation}
    \begin{split}
        \bigl[\tilde\G_{\mu \nu}^{(k)}(z, w; \bm \lambda)\bigr]^* 
        & = \sum_{\mu',\nu'}\iint_{\gamma(t)} \d^2 \bm y\, \left[\tilde\G^{(k-1)}_{\mu\mu'}(z,y_1; \bm \lambda) \Sigma_{\mu'\nu'}(y_1, y_2) \tilde \C_{\nu'\nu}(y_2, w; \bm \lambda)\right]^*\\
        & = \zeta \sum_{\mu',\nu'}\iint_{\gamma(t)} \d^2 \bm y\, \tilde\G^{(k-1)}_{\bar \mu \bar\mu'}(z^*,y_1^*; -\bm \lambda) \Sigma_{\bar\mu'\bar\nu'}(y_1^*, y_2^*) \tilde \C_{\bar \nu' \bar\nu}(y_2^*, w^*;- \bm \lambda)\\
        & \overset{\rm (i)}{=} \zeta \sum_{\mu',\nu'}\iint_{\gamma(t)} \d^2 \bm y\, \tilde\G^{(k-1)}_{\bar \mu \mu'}(z^*,y_1; -\bm \lambda) \Sigma_{\mu'\nu'}(y_1, y_2) \tilde \C_{ \nu' \bar\nu}(y, w^*; -\bm \lambda)\\
        & := \zeta \tilde\G_{\bar \mu \bar \nu}^{(k)}(z^*, w^*; -\bm \lambda)~,
    \end{split}
\end{equation}
where in (i) we performed the harmless variable change $(y_1^*, y_2^*) \to(y_1, y_2)$, under which the contour is invariant, and renamed the dummy indices $(\bar \mu', \bar \nu') \to (\mu', \nu')$. 
\end{proof}
\begin{corollary}
\label{cor:2}
Given Definition~\ref{defn:effective-conjugate}, the following identities for the dressed tilted GF hold:
\begin{equation} 
\label{sm-eq:cor2}
    \tilde \G^{<} = \effconj\tilde \G^> ~, \qquad  \tilde \G^{\,\drot} = \effconj\tilde \G^{\,\tord} ~.
\end{equation}
\end{corollary}
\begin{proof}[Proof of Corollary~\ref{cor:2}]
Given Lemma~\ref{lemma:4}, the proof is formally identical to that of Corollary~\ref{cor:1} with $\tilde \C \to \tilde \G$.
\end{proof}
\begin{corollary}
\label{cor:3}
The following identities for the dressed heat kernel hold:
\begin{equation} 
\label{sm-eq:cor3}
    \g^{<}_\alpha = \effconj\g_{\alpha}^> ~, \qquad  \g_{\alpha}^{\,\drot} = \effconj\g_{\alpha}^{\,\tord} ~,
\end{equation}
where $\effconj$ [Def.~\eqref{defn:effective-conjugate}] has been straightforwardly extended to $\bm\lambda$-independent quantities.
\end{corollary}
\begin{proof}[Proof of Corollary~\ref{cor:3}]
\begin{align}
    \g_{\alpha,\mu\nu}^<(z, w) & := i\frac{\partial}{\partial \lambda_\alpha} \tilde \G^<_{\mu\nu}(z, w; \bm \lambda )\bigg|_{\bm \lambda = \bm 0} =i\frac{\partial}{\partial \lambda_\alpha} \left[\zeta  \tilde \G^>_{\bar \mu\bar \nu}(z, w; -\bm \lambda ) \right]^*\bigg|_{\bm \lambda = \bm 0} \\[5pt]
    & = \zeta  \left[-i \frac{\partial}{\partial \lambda_\alpha} \tilde \G^>_{\bar \mu\bar \nu}(z, w; -\bm \lambda )\bigg|_{\bm \lambda = \bm 0}\right]^*  = \zeta \left[-i \frac{\partial}{\partial (-\lambda_\alpha)} \tilde \G^>_{\bar \mu\bar \nu}(z, w; \bm \lambda )\bigg|_{\bm \lambda = \bm 0}\right]^*  \notag \\[5pt]
    & = \zeta \left[i \frac{\partial}{\partial \lambda_\alpha} \tilde \G^>_{\bar \mu\bar \nu}(z, w; \bm \lambda )\bigg|_{\bm \lambda = \bm 0}\right]^*\notag \\[5pt]
    & := \effconj\g_{\alpha,\mu\nu}^>\notag ~.
\end{align}
The calculation is the same with $\g^{<}_\alpha  \to \g^{\drot}_\alpha $ and $\g^{>}_\alpha  \to \g^{\tord}_\alpha $.
\end{proof}
Given these facts, we can show the desired relation:
\begin{equation}
    \begin{split}
        \zeta\sum_{\mu,\nu}\bigl[\g^<_{\alpha;\mu\nu}(t, \tau) - \g^\drot_{\alpha;\mu\nu}(t, \tau) \bigr] \ev{A_\nu(\tau)A_\mu(t)}_t & \overset{\rm (i)}{=} \zeta^2 \sum_{\mu,\nu}\bigl[\g^>_{\alpha;\bar\mu\bar\nu}(t, \tau) - \g^\tord_{\alpha;\bar\mu\bar\nu}(t, \tau) \bigr]^* \ev{A_\nu(\tau)A_\mu(t)}_t \\
        & \overset{}{=} \sum_{\mu,\nu} \bigl[\g^>_{\alpha;\bar\mu\bar\nu}(t, \tau) - \g^\tord_{\alpha;\bar\mu\bar\nu}(t, \tau) \bigr]^* \ev{A^\dagger_\mu(t)A^\dagger_\nu(\tau)}_t^*\\
        & = \sum_{\mu,\nu} \bigl[\g^>_{\alpha;\bar\mu\bar\nu}(t, \tau) - \g^\tord_{\alpha;\bar\mu\bar\nu}(t, \tau) \bigr]^* \ev{A_{\bar \mu}(t)A_{\bar \nu}(\tau)}_t^* \\
        & \overset{\rm (ii)}{=} \sum_{\mu,\nu} \bigl[\g^>_{\alpha;\mu\nu}(t, \tau) - \g^\tord_{\alpha;\mu\nu}(t, \tau) \bigr]^* \ev{A_{ \mu}(t)A_{ \nu}(\tau)}_t^*~,
    \end{split}
\end{equation}
where (i) follows from Corollary~\ref{cor:3} and in (ii) we relabeled dummy indices. Plugging this result into Eq.~\eqref{sm-eq:transport-equation-lunga}, we finally recover Eq.~\eqref{eq:transport-equation}:
\begin{equation}
    \begin{split}
        I_Q^{(\alpha)}(t)=  \int_{0}^{t} \d \tau\, 
            \bigl[\g^>_{\alpha;\mu\nu}(t, \tau) - \g^\tord_{\alpha;\mu\nu}(t, \tau) \bigr] \ev{A_\mu(t)A_\nu(\tau)}_t +{\rm c.c.} ~.
    \end{split}
\end{equation}
\section{Recursive equation for the dressed transport kernel}
\label{app:D}
In this appendix, we derive Eqs.~\eqref{eq:goth-g-dyson-contour} and~\eqref{eq:physical-time-goth-g-dyson} for the practical calculation of the dressed heat kernel $\g_\alpha$. 

We start by expanding Eq.~\eqref{eq:tilted-dyson} up to first order in $\bm \lambda$:
\begin{equation}
\label{eq:tilted-dyson-taylor}
    \begin{split}
         \tilde \G(12) & = \C(12)  +\sum_\alpha (-i\lambda_\alpha)\cc_\alpha(12)  \\
         &+ \iint_{\gamma(t)} \d(34)\,\Bigl[\G(13) + \sum_\alpha (-i\lambda_\alpha) \g_\alpha(13) \Bigr]\Sigma(34)  \Bigl[\C(42)  +\sum_\alpha (-i\lambda_\alpha)\cc_\alpha(42)\Bigr] + \O(\lambda^2)~,\\
         & =  \C(12) + \iint_{\gamma(t)} \d(34)\,\G(13) \,\Sigma(34)\, \C(42)  \\
         &+ \sum_\alpha (-i\lambda_\alpha)\bigg[\cc_\alpha(12) +\iint_{\gamma(t)} \d(34)\,\tilde \G(13) \,\Sigma(34)  \, \cc_\alpha(42) + \iint_{\gamma(t)} \d(34)\,\g_\alpha(13) \,\Sigma(34)  \, \tilde\C(42) \bigg]+ \O(\lambda^2)~.
    \end{split}
\end{equation}
Since this must hold for every $\bm \lambda \in \mathbb R^R$, equating coefficients of equivalent order on the LHS and RHS we obtain Eq.~\eqref{eq:goth-g-dyson-contour}:
\begin{equation}
\label{sm-eq:goth-g-dyson-contour}
    \begin{split}
        {\g}_\alpha(12)  = {\cc}_\alpha(12) +\iint_{\gamma(t)} \d^2\bm w\,\Bigl\{\G(13) \,\Sigma(34)  \, {\cc}_{\alpha}(42) + {\g}_{\alpha}(13) \,\Sigma(34)  \, \C(42)\Bigr\}~.
    \end{split}
\end{equation}
We will now unwrap this equation on the Keldysh contour. To achieve this, we take advantage of the fact that, by definition~\eqref{sm-eq:self-energy}, we have $\Sigma^> = \Sigma^< = 0$. Therefore, given the contour function
\begin{equation}
    f(1234) = \G(13) \,\Sigma(34)  \, {\cc}_{\alpha}(42) + {\g}_{\alpha}(13) \,\Sigma(34)  \, \C(42)~,
\end{equation}
the contour integral becomes
\begin{equation}
    \begin{split}
        \iint_{\gamma(t)} \d^2 \bm w\,f(1234) 
        & = \iint_{\gamma_+(t)} \d(34)\,f(123_+4_+) +  \iint_{\gamma_-(t)} \d(34)\,f(123_-4_-)\\
        & = \iint_{[0,t]} \d(34)\,\Bigl\{f(123_+4_+) +f(123_-4_-)\Bigr\}~,
    \end{split}
\end{equation}
noting that the sign of the double integration measure $\d (34)$ is the same on the two branches when both the variables $3$ and $4$ lie on the same branch. For reasons that will be clear later, we only need to write physical-time equations for the pair $\vec \g_{\alpha} = (\g_\alpha^>, \g_\alpha^\drot)$, which are
\begin{align}
\label{sm-eq:dyson-for-g>}
        {\g}^>_\alpha(12)  
            & = {\cc}^>_\alpha(12)+\iint_{[0, t]} \d (34)\,\Bigl\{\G(1_+3_+) \,\Sigma(3_+4_+)  \, {\cc}_{\alpha}(4_+2_-) + \G(1_+3_-) \,\Sigma(3_-4_-)  \, {\cc}_{\alpha}(4_-2_-)\Bigr\}\\
                & \phantom{={\cc}_\alpha(t, \tau)} + \iint_{[0, t]}\d (34))\,\Bigl\{{\g}_{\alpha}(1_+3_+) \,\Sigma(3_+4_+)  \, \C(4_+2_-)+ {\g}_{\alpha}(1_+3_-) \,\Sigma(3_-4_-)  \, \C(4_-2_-)\Bigr\}\notag  \\[3pt]
            & = {\cc}^>_\alpha(12)+\iint_{[0, t]} \d (34)\,\Bigl\{\G^\drot(13) \,\Sigma^\drot(34)  \, {\cc}^>_{\alpha}(42) + \G^>(13) \,\Sigma^\tord(34)  \, {\cc}^\tord_{\alpha}(42)\Bigr\}\notag \\
                & \phantom{={\cc}_\alpha(t, \tau)} + \iint_{[0, t]} \d (34)\,\Bigl\{{\g}^\drot_{\alpha}(13) \,\Sigma^\drot(34)  \, \C^>(42)+ {\g}^>_{\alpha}(13) \,\Sigma^\tord(34)  \, \C^\tord(42)\Bigr\} \notag \\[3pt]
            & \overset{\rm (i)}{=} {\cc}^>_\alpha(12)+\iint_{[0, t]} \d (34)\,\Bigl\{-\G^\drot(13) \,\Sigma^\tord(34)  \, {\cc}^>_{\alpha}(42) - {\g}^\drot_{\alpha}(13) \,\Sigma^\tord(34)  \, \C^>(42)\notag + {\g}^>_{\alpha}(13) \,\Sigma^\tord(34)  \, \C^\tord(42)\Bigr\}~,\notag 
\end{align}
\begin{align}
\label{sm-eq:dyson-for-gdrot}
        {\g}^\drot_\alpha(12)  
        & = \iint_{[0, t]} \d (34)\,\Bigl\{\G(1_+3_-) \,\Sigma(3_-4_-)  \, {\cc}_{\alpha}(4_-2_+) + \G(1_+3_+) \,\Sigma(3_+4_+)  \, {\cc}_{\alpha}(4_+2_+) \Bigr\} \\
        & + \iint_{[0, t]} \d (34)\,\Bigl\{{\g}_{\alpha}(1_+3_+) \,\Sigma(3_+ 4_+)  \, \C(4_+2_+)+ {\g}_{\alpha}(1_+3_-) \,\Sigma(3_-4_-)  \, \C(4_-2_+)\Bigr\} \notag \\
        & = \iint_{[0, t]} \d (34)\,\Bigl\{\G^>(13) \,\Sigma^\tord(34)  \, {\cc}^<_{\alpha}(42) + \G^\drot(13) \,\Sigma^\drot(34)  \, {\cc}^\drot_{\alpha}(42)  \Bigr\} \notag  \\
        & + \iint_{[0, t]} \d (34)\,\Bigl\{{\g}_{\alpha}^\drot(13) \,\Sigma^\drot(34)  \, \C^\drot(42)+ {\g}_{\alpha}^>(13) \,\Sigma^\tord(34)  \, \C^<(42)\Bigr\} \notag  \\
        & \overset{\rm (i)}{=} \iint_{[0, t]} \d (34)\,\Bigl\{\G^>(13) \,\Sigma^\tord(34)  \, {\cc}^<_{\alpha}(42)- {\g}_{\alpha}^\drot(13) \,\Sigma^\tord(34)  \, \C^\drot(42)+ {\g}_{\alpha}^>(13) \,\Sigma^\tord(34)  \, \C^<(42) \Bigr\} ~.\notag 
\end{align}
Here, we have made heavy use of the fact that ${\cc}_\alpha^\tord = {\cc}_\alpha^\drot \equiv 0$, and in (i) we have exploited the identity $\Sigma^\drot = - \Sigma^\tord$. {\it A posteriori}, we observe that neither of the other two Keldysh components (i.e. $\g^<$ and $\g^\tord$) appear in Eqs.~\eqref{sm-eq:dyson-for-g>}-\eqref{sm-eq:dyson-for-gdrot}. Therefore, the latter define a self-consistent system of coupled equations for $(\g_\alpha^>, \g_\alpha^\drot)$. The conjugate pair $(\g_\alpha^<, \g_\alpha^\tord)$ is immediately obtained by effective conjugation [Def.~\ref{defn:effective-conjugate}] thanks to Corollary~\ref{cor:3} to Lemma~\ref{lemma:4}. To achieve the neater form~\eqref{eq:physical-time-goth-g-dyson}, we define the following functionals:
\begin{align}
\label{sm-eq:K-functional}
    \bigl[\K^X f\bigr](12) & : = \iint_{[0, t]} \d^2 \bm w\,\G^X(13) \,\Sigma^\tord(34)  \, f(42)\qquad {\rm for}\; X \in [>,<]~,\\
\label{sm-eq:I-functional}
    \bigl[\I^Yf\bigr](12) & : = \iint_{[0, t]} \d^2 \bm w\,f(13)\,\Sigma^\tord(34)\, \C^X(42) \qquad {\rm for}\; Y \in [>, <, \tord, \drot]~,
\end{align}
therefore getting
\begin{align}
    {\g}^>_\alpha &= \cc^>_\alpha - \K^<\cc^>_\alpha +\I^\tord\g^>_\alpha- \I^>\g^\drot_\alpha ~,\\
    {\g}^\drot_\alpha &= \K^>\cc^<_\alpha - \I^\drot\g^\drot_\alpha+ \I^<\g^>_\alpha ~,
\end{align}
where we employed Lemma F.2 in the Supplemental Material of Ref.~\cite{d2025exact}, which states that for all $z \in \gamma(t)$ we can write
\begin{equation}
    \G^>(t, z) \equiv \G^\tord(t, z)~.
\end{equation}
Turning to matrix form, we immediately recover Eq.~\eqref{eq:physical-time-goth-g-dyson} of the main text. 

In order to solve such an equation via numerical methods, we need to write out explicitly the integral functionals~\eqref{sm-eq:K-functional}-\eqref{sm-eq:I-functional}. First, we observe that $\Sigma^\tord$ can be written as follows:
\begin{equation}
    \Sigma^\tord_{ij}(u,v) = \theta(v-u)\,\sigma_{ij}(u-v) := \theta(v-u)\,\bigl[A_i(u-v), A_j(0)\bigr]_\zeta
\end{equation}
Then, we recall the definition of the Keldysh components of $\C$ in terms of the physical-time GF $\chi$, which are just Eqs.~\eqref{eq:C-tord}-\eqref{eq:C-drot} together Eqs~\eqref{eq:C->}-\eqref{eq:C-<} for $\bm \lambda = \bm 0$:
\begin{align}
    \C_{\mu\nu}^>(v, \tau) & = \chi_{\mu\nu}(v-\tau)~,\\
    \C_{\mu\nu}^<(v, \tau) & = \zeta\chi_{\nu\mu}(\tau-v)~,\\
    \C_{\mu\nu}^\tord(v, \tau) & = \theta(v-\tau)\chi_{\mu\nu}(v-\tau) +\zeta \theta(\tau-v)\chi_{\nu\mu}(\tau-v)~,\\
    \C_{\mu\nu}^\drot(v, \tau) & = \theta(\tau-v)\chi_{\mu\nu}(v-\tau) +\zeta \theta(v-\tau)\chi_{\nu\mu}(\tau-v)~.
\end{align}
The $\K$ functionals~\eqref{sm-eq:K-functional} are therefore
\begin{align}
        \bigl[\K^< f\bigr](t, \tau) 
        & = \int_0^t \d v\int_0^v\d u\,\G^<(t, u) \,\sigma(u-v)  \, f(v, \tau)~,\\
        \bigl[\K^> f\bigr](t, \tau) 
        & = \int_0^t \d v\int_0^v\d u\,\G^>(t, u) \,\sigma(u-v)  \, f(v, \tau)~.
\end{align}
and for the $\I$ functionals~\eqref{sm-eq:I-functional} we have
\begin{align}
    \bigl[\I^>f\bigr](t, \tau) 
        & = \int_0^t \d v\int_0^v\d u\,f(t, u)\,\,\sigma(u-v)  \, \chi(v-\tau)  ~, \\[5pt]
    \bigl[\I^<f\bigr](t, \tau) 
        & = \zeta\int_0^t \d v\int_0^v\d u\,f(t, u)\,\,\sigma(u-v)  \, \chi^T(\tau-v) ~, \\[5pt]
    \bigl[\I^\tord f\bigr](t, \tau) 
        & = \int_\tau^t \d v \int_0^v\d u\,f(t, u)\,\sigma(u-v)  \, \chi(v-\tau) +\zeta\int_0^\tau \d v \int_0^v\d u\,f(t, u)\,\sigma(u-v)  \,\chi^T(\tau-v)~, \\[5pt]
    \bigl[\I^\drot f\bigr](t, \tau) 
        & = \int_0^\tau \d v \int_0^v\d u\,f(t, u)\,\sigma(u-v)  \, \chi(v-\tau)   +\zeta\int_\tau^t \d v \int_0^v\d u\,f(t, u)\,\sigma(u-v)  \,\chi^T(\tau-v) ~.
\end{align}
\section{Explicit form of the bare heat kernel}
\label{app:E}
Here, we derive Eqs.~\eqref{eq:bare-heat-kernel},~\eqref{eq:goth-c-} and~\eqref{eq:goth-c+} of Section~\ref{sec:IV} of the main text. We start by direct inspection of Eq.~\eqref{eq:chi} for the physical-time GF: in absence of anomalous correlations, the $\alpha$-th block of $\chi$ takes the form given in~\eqref{sm-eq:C-alpha-block}. By appropriate definition of the density of states $J_\alpha(\omega)$, the two relevant matrix elements can be presented as
\begin{align}
\label{sm-eq:chi-minus-explicit}
    \chi_{\alpha}^-(s) & : =\tr \bigl[B^\dagger_\alpha(s) B_\alpha^{\vphantom{\dagger}}(0) \,\omega_E\bigr] = \int_0^{\infty} \d \omega\,J_\alpha(\omega) f_\alpha(\omega)\,e^{i\omega s}~,\\
\label{sm-eq:chi-plus-explicit}
    \chi_{\alpha}^+(s)&: =\tr \bigl[B_\alpha^{\vphantom{\dagger}}(s) B_\alpha^\dagger(0)\, \omega_E\bigr] = \int_0^{\infty} \d \omega\,J_\alpha(\omega) \bigl[1+\zeta f_\alpha(\omega)\bigr]\,e^{-i\omega s}~,
\end{align}
where $f_\alpha(\omega) = \bigl[e^{\beta_\alpha(\omega-\mu_\alpha) }-\zeta\bigr]^{-1}$ characterizes the reservoir statistics. Per definition~\eqref{eq:goth-c} and Eq.~\eqref{eq:C->}, we have
\begin{equation}
    \begin{split}
        \cc_{\alpha,\pm}^{>}(z,w) &:= i\frac{\partial\tilde \C^{\, >}_{\alpha, \pm}(z, w; \bm \lambda)}{\partial \lambda_\alpha} \Bigg|_{\bm \lambda = \bm 0}   = i\frac{\partial }{\partial \lambda_\alpha} \bigl[e^{\pm i\mu_\alpha \lambda_\alpha}\chi_\alpha^{\pm}(z-w+\lambda_\alpha)\bigr]\Bigg|_{\bm\lambda = \bm 0}\\[3pt]
        & = i \frac{\partial \chi^{\pm}_\alpha(s)}{\partial s}\Bigg|_{s=z-w} \mp \mu_\alpha  \chi_{\alpha}^{\pm}(z-w)  ~,
    \end{split}
\end{equation}
which in matrix form is exactly Eq.~\eqref{eq:bare-heat-kernel}. By inserting Eqs.~\eqref{sm-eq:chi-minus-explicit}-\eqref{sm-eq:chi-plus-explicit} in the latter, we recover Eqs.~\eqref{eq:goth-c-} and~\eqref{eq:goth-c+}.
\section{Differential equation for the second moments of the heat exchange PDF}
\label{app:F}
In this appendix we show how, from the definition~\eqref{eq:CF} of the generating function, we can also derive the second moments of the heat distribution as
\begin{equation}
    \ev{Q_\alpha Q_\beta}(t) := -\frac{\partial^2 \M(t; \bm \lambda)}{\partial \lambda_\alpha \partial \lambda_\beta} \Bigg|_{\bm \lambda = \bm 0} = - \tr \partial^2_{\alpha\beta} \tilde \rho_S(t;\bm \lambda) \big|_{\bm \lambda = \bm 0}~.
\end{equation}
where $\partial_\alpha := \partial/\partial \lambda_\alpha$. We can obtain an ODE for these correlators by differentiating Eq.~\eqref{eq:tGME} twice with respect to $\bm \lambda$ and setting $\bm \lambda = \bm 0$:
\begin{equation}
\label{sm-eq:tgme-second-derivative}
    \begin{split}
        \frac{\d}{\d t} \partial^2_{\alpha\beta} \tilde \rho_S(t;\bm \lambda) \big|_0
        & = \sum_{\mu, \nu} \int_{0}^{t} \d \tau\, \Bigl\{+
            \partial^2_{\alpha\beta}\tilde \G^{>}_{\mu \nu}(t, \tau; \bm \lambda)  A_\nu(\tau) \tilde \rho_S(t; \bm \lambda) A_\mu(t) + \partial_\beta \tilde \G^{>}_{\mu \nu}(t, \tau; \bm \lambda) A_\nu(\tau) \partial_\alpha \tilde \rho_S(t;\bm \lambda)  A_\mu(t)  \\
            &   \phantom{=\sum_{\mu, \nu} \int_{0}^{t} \d \tau\, \Bigl\{}+
            \partial_\alpha \tilde \G^{>}_{\mu \nu}(t, \tau; \bm \lambda) A_\nu(\tau)\partial_\beta \tilde \rho_S(t; \bm \lambda)  A_\mu(t) + \tilde \G^{>}_{\mu \nu}(t, \tau; \bm \lambda) A_\nu(\tau) \partial^2_{\alpha\beta} \tilde \rho_S(t;\bm \lambda)  A_\mu(t)  \\
        &\phantom{=\sum_{\mu, \nu} \int_{0}^{t} \d \tau\, \Bigl\{}-
                  \partial^2_{\alpha\beta}\tilde \G^{\tord}_{\mu \nu}(t, \tau; \bm \lambda)A_\mu(t) A_\nu(\tau) \tilde \rho_S(t; \bm \lambda) -
                 \partial_\beta \tilde \G^{\tord}_{\mu \nu}(t, \tau; \bm \lambda)  A_\mu(t) A_\nu(\tau) \partial_\alpha \tilde \rho_S(t;\bm \lambda) \\
            &\phantom{=\sum_{\mu, \nu} \int_{0}^{t} \d \tau\, \Bigl\{}-
                 \partial_\alpha \tilde \G^{\tord}_{\mu \nu}(t, \tau; \bm \lambda)A_\mu(t) A_\nu(\tau) \partial_\beta\tilde \rho_S(t; \bm \lambda)  -
                 \tilde \G^{\tord}_{\mu \nu}(t, \tau; \bm \lambda) A_\mu(t) A_\nu(\tau) \partial^2_{\alpha \beta}\tilde \rho_S(t;\bm \lambda) \\
        & \phantom{=\sum_{\mu, \nu} \int_{0}^{t} \d \tau\, \Bigl\{}  +\zeta \partial^2_{\alpha\beta}\tilde \G_{\mu\nu}^{<}(t, \tau; \bm \lambda) A_\mu(t) \tilde \rho_S(t; \bm \lambda)A_\nu(\tau)+\zeta\partial_\beta \tilde \G_{\mu\nu}^{<}(t, \tau; \bm \lambda)  A_\mu(t) \partial_\alpha \tilde \rho_S(t;\bm \lambda) A_\nu(\tau)\\
            & \phantom{=\sum_{\mu, \nu} \int_{0}^{t} \d \tau\, \Bigl\{}  +\zeta \partial_\alpha\tilde \G_{\mu\nu}^{<}(t, \tau; \bm \lambda) A_\mu(t) \partial_\beta \tilde \rho_S(t; \bm \lambda)A_\nu(\tau) +\zeta\tilde \G_{\mu\nu}^{<}(t, \tau; \bm \lambda) A_\mu(t) \partial_{\alpha\beta}^2 \tilde \rho_S(t;\bm \lambda) A_\nu(\tau)\\
        &\phantom{=\sum_{\mu, \nu} \int_{0}^{t} \d \tau\, \Bigl\{} -\zeta \partial^2_{\alpha\beta} \tilde \G_{\mu\nu}^{\drot}(t, \tau; \bm \lambda)\tilde \rho_S(t; \bm \lambda) A_\nu(\tau) A_\mu(t) -\zeta\partial_\beta\tilde \G_{\mu\nu}^{\drot}(t, \tau; \bm \lambda)  \partial_\alpha \tilde \rho_S(t;\bm \lambda)A_\nu(\tau) A_\mu(t)\Bigr\}\\
            &\phantom{=\sum_{\mu, \nu} \int_{0}^{t} \d \tau\, \Bigl\{} -\zeta\partial_\alpha \tilde \G_{\mu\nu}^{\drot}(t, \tau; \bm \lambda) \partial_\beta\tilde \rho_S(t; \bm \lambda) A_\nu(\tau) A_\mu(t) -\zeta\tilde \G_{\mu\nu}^{\drot}(t, \tau; \bm \lambda)  \partial^2_{\alpha \beta}\tilde \rho_S(t;\bm \lambda)A_\nu(\tau) A_\mu(t)\Bigr\} \Bigg|_0~. \\[3pt]
    \end{split}
\end{equation}
Out of the latter 16 terms on the RHS, the four containing $\partial_{\alpha\beta}^2\tilde \rho_S$ disappear when we take the trace. We may now define
\begin{equation}
    {\mathfrak C}_{\alpha\beta}(z,w) : = -\frac{\partial^2\tilde \C(z, w; \bm \lambda )}{\partial \lambda_\alpha\partial \lambda_\beta} \bigg|_{\bm \lambda = \bm 0}~, \qquad {\mathfrak G}_{\alpha\beta} (z,w) : = -\frac{\partial\tilde \G(z, w; \bm \lambda )}{\partial \lambda_\alpha\partial \lambda_\beta} \bigg|_{\bm \lambda = \bm 0}~.
\end{equation}
As a consequence of Lemmas~\ref{lemma:1} and~\ref{lemma:4}, we see that $\effconj{\mathfrak C}^> =  \mathfrak{C}^<$ and $\effconj{\mathfrak C}^\tord =  \mathfrak{C}^\drot$, and the same holds for $\mathfrak G$. This can be proven by following the same procedure as in the proof of Corollary~\ref{cor:3}. Taking the trace of Eq.~\eqref{sm-eq:tgme-second-derivative} and recalling the definition of $\theta_\alpha(t):= i \partial_\alpha \tilde \rho_S(t; \bm \lambda) \big|_{\bm \lambda = \bm 0}$, we get
\begin{equation}
    \begin{split}
        \frac{\d}{\d t} \ev{Q_\alpha Q_\beta}_t
        & = \sum_{\mu, \nu} \int_{0}^{t} \d \tau\, \Bigl\{
            \bigl[{\mathfrak G}_{\alpha\beta}^>  - {\mathfrak G}_{\alpha\beta}^\tord\bigr]_{\mu\nu}(t, \tau)  \tr\bigl[A_{\mu}(t) A_\nu(\tau) \rho_S(t)\bigr]  \\
            &   \phantom{=\sum_{\mu, \nu} \int_{0}^{t} \d \tau\, }+ \bigl[\g^>_{\alpha} - \g^\tord_\alpha\bigr]_{\mu\nu}(t, \tau) \tr \bigl[A_{\mu}(t) A_\nu(\tau) \theta_\beta(t)\bigr]\\
            &   \phantom{=\sum_{\mu, \nu} \int_{0}^{t} \d \tau\, }+ \bigl[\g^>_{\beta} - \g^\tord_\beta\bigr]_{\mu\nu}(t, \tau) \tr \bigl[A_{\mu}(t) A_\nu(\tau) \theta_\alpha(t)\bigr]\Bigr\} + {\rm c.c.}~. \\[3pt]
    \end{split}
\end{equation}
In order to evaluate such correlators, we would need to solve the inhomogeneous ODE~\eqref{sm-eq:ode-for-theta} for $\theta_\alpha(t)$ and $\theta_\beta(t)$, and to obtain a self-consistent equation for the second-order dressed kernel ${\mathfrak G}_{\alpha\beta}$. The latter can be found in a straightforward fashion by following the procedure outlined in App.~\ref{app:D} for the dressed heat kernel, rewriting Eq.~\eqref{eq:tilted-dyson-taylor} up to second order in $\bm \lambda$.
\section{Notes on the case-study implementation}
\label{app:G}
In this final section, we give some additional information on the numerical analysis which results are shown in Sec.~\ref{sec:VI} of the main text.

The system-environment coupling Hamiltonian can be put in the form~\eqref{eq:interaction-hamiltonian} as follows:
\begin{equation}
    \V = A_1 \otimes B_L^\dagger +  A_1^\dagger \otimes B_L + A_2 \otimes B_R^\dagger +  A_2^\dagger \otimes B_R~,
\end{equation}
where $A_1 = a_1$, $A_2 = a_2$ and
\begin{equation}
    B_\alpha^{\vphantom{\dagger}} = \sum_{l} g_{\alpha, l} \c{\alpha, l} {\cal P}_\alpha~,
\end{equation}
with ${\cal P}_\alpha := (-1)^{\sum_l \cd{\alpha, l}  \c{\alpha, l}}$ being the parity operator associated with reservoir $\alpha$. The presence of this operator, due to fermionic anti-commutation rules, is inconsequential and can be forgotten hereafter. The physical GF [Eq.~\eqref{eq:chi}] for each reservoir is then
\begin{equation}
    \begin{split}
        \chi_\alpha(s) & = 
        \begin{pmatrix}
            0 & \vev{B_\alpha^\dagger(s) B_\alpha^{\vphantom{\dagger}}(0)}\\[3pt]
            \vev{B_\alpha^{\vphantom{\dagger}}(s) B_\alpha^\dagger(0) } & 0
        \end{pmatrix}
        \\[3pt]
        & = 
        \sum_{l, m} g_{\alpha, l}\, g^*_{\alpha, m}
        \begin{pmatrix}
            0 & \vev{\cd{\alpha, m}(s) \c{\alpha, l}(0)}\\[3pt]
            \vev{\c{\alpha, l}(s) \cd{\alpha, m}(0) } & 0
        \end{pmatrix}\\[3pt]
        & = 
        \sum_{l} |g_{\alpha, l}|^2
        \begin{pmatrix}
            0 & f_{\alpha, l}e^{+i\omega_{\alpha, l} s}\\[3pt]
            \bigl[1- f_{\alpha, l}\bigr]e^{-i\omega_{\alpha, l} s} & 0
        \end{pmatrix}\\[3pt]
        & = 
        \int_0^\infty \d \omega \,J_\alpha(\omega)
        \begin{pmatrix}
            0 & f_{\alpha}(\omega) e^{+i\omega s}\\[3pt]
            \bigl[1- f_{\alpha}(\omega)\bigr] e^{-i\omega s} & 0
        \end{pmatrix}~.
    \end{split}
\end{equation}
Here, we employed the shorthand notation $\vev{\bullet} := \tr [\,\bullet\,\rho_{SE}(0)]$, we defined the density of states (DOS)
\begin{equation}
    J_\alpha(\omega) := \sum_l|g_{\alpha, l}|^2 \delta(\omega - \omega_{\alpha, l})~,
\end{equation}
and used the following results:
\begin{align}
    \vev{\cd{\alpha, m}(s) \c{\alpha, l}(0)} & =\vev{\cd{\alpha, m} \c{\alpha, l}}  e^{+i\omega_{\alpha, m} s}:= \delta_{l, m} f_{\alpha, l}e^{+i\omega_{\alpha, m} s}~,\\
    \vev{\c{\alpha, l}(s)\cd{\alpha, m}(0) } & =\vev{\c{\alpha, l}\cd{\alpha, m} }  e^{-i\omega_{\alpha, l} s}:= \delta_{l, m}\bigl[1- f_{\alpha, l}\bigr]e^{-i\omega_{\alpha, l} s} ~.
\end{align}
For numerical implementation, we choose for both the reservoirs the same Lorentzian spectral line shape
\begin{equation}
    J_L(\omega) = J_R(\omega) = \frac{\eta}{\pi} \frac{\Gamma}{\omega^2 + \Gamma^2} \,e^{-|\omega|/\omega_c}~,
\end{equation}
where $\eta$ quantifies the coupling strength, $\Gamma$ is the excitation linewidth and we employed a long-wavelength exponential cutoff with $\omega_c \gg \Gamma$ to facilitate numerical implementation. To produce the numerical results in Fig.~\ref{fig:5}a of the main text, we fixed $\Gamma = 1$, $\omega_c = 10$. The bare heat kernels are calculated from Eqs.~\eqref{eq:goth-c-}-\eqref{eq:goth-c+}. The recursive equations~\eqref{eq:bare-dyson} and~\eqref{eq:physical-time-goth-g-dyson} are formulated as fixed-point problems and solved via Newton--Krylov algorithm after discretization on a nonlinear time grid.

\twocolumngrid

\bibliography{bibliography.bib}

\end{document}

%% file: figs/fig1.pdf_tex
\begingroup%
  \makeatletter%
  \providecommand\color[2][]{%
    \errmessage{(Inkscape) Color is used for the text in Inkscape, but the package 'color.sty' is not loaded}%
    \renewcommand\color[2][]{}%
  }%
  \providecommand\transparent[1]{%
    \errmessage{(Inkscape) Transparency is used (non-zero) for the text in Inkscape, but the package 'transparent.sty' is not loaded}%
    \renewcommand\transparent[1]{}%
  }%
  \providecommand\rotatebox[2]{#2}%
  \newcommand*\fsize{\dimexpr\f@size pt\relax}%
  \newcommand*\lineheight[1]{\fontsize{\fsize}{#1\fsize}\selectfont}%
  \ifx\svgwidth\undefined%
    \setlength{\unitlength}{167.33342328bp}%
    \ifx\svgscale\undefined%
      \relax%
    \else%
      \setlength{\unitlength}{\unitlength * \real{\svgscale}}%
    \fi%
  \else%
    \setlength{\unitlength}{\svgwidth}%
  \fi%
  \global\let\svgwidth\undefined%
  \global\let\svgscale\undefined%
  \makeatother%
  \begin{picture}(1,0.55117609)%
    \lineheight{1}%
    \setlength\tabcolsep{0pt}%
    \put(0,0){\includegraphics[width=\unitlength,page=1]{fig1.pdf}}%
    \put(0.66877612,0.07951111){\makebox(0,0)[lt]{\lineheight{1.10000002}\smash{\begin{tabular}[t]{l}{\Huge $\H{}$}\end{tabular}}}}%
    \put(0.50434239,0.20560963){\makebox(0,0)[lt]{\lineheight{1.10000002}\smash{\begin{tabular}[t]{l}{\large $\H{S}$}\end{tabular}}}}%
    \put(0.70655442,0.41218373){\makebox(0,0)[lt]{\lineheight{1.10000002}\smash{\begin{tabular}[t]{l}{\large $\H{1}$}\end{tabular}}}}%
    \put(0.20760648,0.3887528){\makebox(0,0)[lt]{\lineheight{1.10000002}\smash{\begin{tabular}[t]{l}{\large $\H{3}$}\end{tabular}}}}%
    \put(0.22509606,0.12454626){\makebox(0,0)[lt]{\lineheight{1.10000002}\smash{\begin{tabular}[t]{l}{\large $\H{2}$}\end{tabular}}}}%
    \put(0,0){\includegraphics[width=\unitlength,page=2]{fig1.pdf}}%
    \put(0.30251389,0.20994909){\makebox(0,0)[lt]{\lineheight{1.10000002}\smash{\begin{tabular}[t]{l}$\V_2$\end{tabular}}}}%
    \put(0.29820526,0.31748396){\makebox(0,0)[lt]{\lineheight{1.10000002}\smash{\begin{tabular}[t]{l}$\V_3$\end{tabular}}}}%
    \put(0.87302417,0.4119953){\makebox(0,0)[lt]{\lineheight{1.10000002}\smash{\begin{tabular}[t]{l}$\beta_1, \,\mu_1$\end{tabular}}}}%
    \put(-0.00209691,0.38958125){\makebox(0,0)[lt]{\lineheight{1.10000002}\smash{\begin{tabular}[t]{l}$\beta_3, \,\mu_3$\end{tabular}}}}%
    \put(-0.00209691,0.12544631){\makebox(0,0)[lt]{\lineheight{1.10000002}\smash{\begin{tabular}[t]{l}$\beta_2, \,\mu_2$\end{tabular}}}}%
    \put(0.6204012,0.3294141){\makebox(0,0)[lt]{\lineheight{1.10000002}\smash{\begin{tabular}[t]{l}$\V_1$\end{tabular}}}}%
  \end{picture}%
\endgroup%

%% file: figs/fig2.pdf_tex
\begingroup%
  \makeatletter%
  \providecommand\color[2][]{%
    \errmessage{(Inkscape) Color is used for the text in Inkscape, but the package 'color.sty' is not loaded}%
    \renewcommand\color[2][]{}%
  }%
  \providecommand\transparent[1]{%
    \errmessage{(Inkscape) Transparency is used (non-zero) for the text in Inkscape, but the package 'transparent.sty' is not loaded}%
    \renewcommand\transparent[1]{}%
  }%
  \providecommand\rotatebox[2]{#2}%
  \newcommand*\fsize{\dimexpr\f@size pt\relax}%
  \newcommand*\lineheight[1]{\fontsize{\fsize}{#1\fsize}\selectfont}%
  \ifx\svgwidth\undefined%
    \setlength{\unitlength}{198.70299091bp}%
    \ifx\svgscale\undefined%
      \relax%
    \else%
      \setlength{\unitlength}{\unitlength * \real{\svgscale}}%
    \fi%
  \else%
    \setlength{\unitlength}{\svgwidth}%
  \fi%
  \global\let\svgwidth\undefined%
  \global\let\svgscale\undefined%
  \makeatother%
  \begin{picture}(1,0.3908932)%
    \lineheight{1}%
    \setlength\tabcolsep{0pt}%
    \put(0,0){\includegraphics[width=\unitlength,page=1]{fig2.pdf}}%
    \put(0.36637574,0.28587962){\makebox(0,0)[lt]{\lineheight{1.10000002}\smash{\begin{tabular}[t]{l}$\gamma_+(t)$\end{tabular}}}}%
    \put(0.36637574,0.0697289){\makebox(0,0)[lt]{\lineheight{1.10000002}\smash{\begin{tabular}[t]{l}$\gamma_-(t)$\end{tabular}}}}%
    \put(0.82710196,0.20689213){\makebox(0,0)[lt]{\lineheight{1.10000002}\smash{\begin{tabular}[t]{l}$t$\end{tabular}}}}%
    \put(0.79374804,0.07204543){\makebox(0,0)[lt]{\lineheight{1.10000002}\smash{\begin{tabular}[t]{l}$t_-$\end{tabular}}}}%
    \put(0.79374804,0.28308638){\makebox(0,0)[lt]{\lineheight{1.10000002}\smash{\begin{tabular}[t]{l}$t_+$\end{tabular}}}}%
    \put(0.11068442,0.20092514){\makebox(0,0)[lt]{\lineheight{1.10000002}\smash{\begin{tabular}[t]{l}$0$\end{tabular}}}}%
    \put(0.11551738,0.366858){\makebox(0,0)[lt]{\lineheight{1.10000002}\smash{\begin{tabular}[t]{l}$\Im z$\end{tabular}}}}%
    \put(0.898463,0.21542423){\makebox(0,0)[lt]{\lineheight{1.10000002}\smash{\begin{tabular}[t]{l}$\Re z$\end{tabular}}}}%
  \end{picture}%
\endgroup%

%% file: figs/fig4.pdf_tex
\begingroup%
  \makeatletter%
  \providecommand\color[2][]{%
    \errmessage{(Inkscape) Color is used for the text in Inkscape, but the package 'color.sty' is not loaded}%
    \renewcommand\color[2][]{}%
  }%
  \providecommand\transparent[1]{%
    \errmessage{(Inkscape) Transparency is used (non-zero) for the text in Inkscape, but the package 'transparent.sty' is not loaded}%
    \renewcommand\transparent[1]{}%
  }%
  \providecommand\rotatebox[2]{#2}%
  \newcommand*\fsize{\dimexpr\f@size pt\relax}%
  \newcommand*\lineheight[1]{\fontsize{\fsize}{#1\fsize}\selectfont}%
  \ifx\svgwidth\undefined%
    \setlength{\unitlength}{133.89995064bp}%
    \ifx\svgscale\undefined%
      \relax%
    \else%
      \setlength{\unitlength}{\unitlength * \real{\svgscale}}%
    \fi%
  \else%
    \setlength{\unitlength}{\svgwidth}%
  \fi%
  \global\let\svgwidth\undefined%
  \global\let\svgscale\undefined%
  \makeatother%
  \begin{picture}(1,0.40856293)%
    \lineheight{1}%
    \setlength\tabcolsep{0pt}%
    \put(0,0){\includegraphics[width=\unitlength,page=1]{fig4.pdf}}%
    \put(0.37205604,0.19922424){\makebox(0,0)[lt]{\lineheight{1.10000002}\smash{\begin{tabular}[t]{l}\Large $a_1$\end{tabular}}}}%
    \put(0.56395837,0.19922424){\makebox(0,0)[lt]{\lineheight{1.10000002}\smash{\begin{tabular}[t]{l}\Large $a_2$\end{tabular}}}}%
    \put(0.06932845,0.19176894){\makebox(0,0)[lt]{\lineheight{1.10000002}\smash{\begin{tabular}[t]{l}\Large $\beta_L$\end{tabular}}}}%
    \put(0.8330095,0.19176894){\makebox(0,0)[lt]{\lineheight{1.10000002}\smash{\begin{tabular}[t]{l}\Large $\beta_R$\end{tabular}}}}%
    \put(0.24086221,0.23320471){\makebox(0,0)[lt]{\lineheight{1.10000002}\smash{\begin{tabular}[t]{l}\large $\eta$\end{tabular}}}}%
    \put(0.7107757,0.23320471){\makebox(0,0)[lt]{\lineheight{1.10000002}\smash{\begin{tabular}[t]{l}\large $\eta$\end{tabular}}}}%
    \put(0.17949909,-0.00663632){\makebox(0,0)[lt]{\lineheight{1.10000002}\smash{\begin{tabular}[t]{l}\Large $-I_Q^{(L)}(t)$\end{tabular}}}}%
    \put(0.65432065,-0.00663632){\makebox(0,0)[lt]{\lineheight{1.10000002}\smash{\begin{tabular}[t]{l}\Large $I_Q^{(R)}(t)$\end{tabular}}}}%
    \put(0.57006495,0.27946133){\makebox(0,0)[lt]{\lineheight{1.10000002}\smash{\begin{tabular}[t]{l}\large $\varepsilon_2$\end{tabular}}}}%
    \put(0.37792971,0.27946133){\makebox(0,0)[lt]{\lineheight{1.10000002}\smash{\begin{tabular}[t]{l}\large $\varepsilon_1$\end{tabular}}}}%
    \put(0.47433962,0.22783827){\makebox(0,0)[lt]{\lineheight{1.10000002}\smash{\begin{tabular}[t]{l}\large $\Delta$\end{tabular}}}}%
    \put(0.59815745,0.34782831){\makebox(0,0)[lt]{\lineheight{1.10000002}\smash{\begin{tabular}[t]{l}\Large $S$\end{tabular}}}}%
  \end{picture}%
\endgroup%

%% file: bibliography.bib
@preamble{ " \newcommand{\noop}[1]{} " }

@book{breuer2007book,
  author    = {Breuer, Heinz-Peter and Petruccione, Francesco},
  publisher = {Oxford University Press},
  title     = {The Theory of Open Quantum Systems},
  year      = {2007},
  doi       = {10.1093/acprof:oso/9780199213900.001.0001},
  url       = {https://doi.org/10.1093/acprof:oso/9780199213900.001.0001}
}

@article{lindblad1976classic,
  author  = {Lindblad, G\"oran},
  journal = {Commun. Math. Phys.},
  title   = {On the generators of quantum dynamical semigroups},
  volume  = {48},
  pages   = {119},
  year    = {1976},
  doi     = {10.1007/BF01608499}
}

@article{gorini1976classic,
  author  = {Gorini, Vittorio and Kossakowski, Andrzej and Sudarshan, Ennackal Chandy George},
  journal = {J. Math. Phys.},
  title   = {Completely positive dynamical semigroup of {$N$}-level systems},
  volume  = {17},
  pages   = {821},
  year    = {1976},
  doi     = {10.1063/1.522979}
}

@article{deVega2017review,
  author  = {de Vega, Inés and Alonso, Daniel},
  title   = {Dynamics of {non-Markovian} open quantum systems},
  journal = {Rev. Mod. Phys.},
  volume  = {89},
  pages   = {015001},
  year    = {2017},
  doi     = {10.1103/RevModPhys.89.015001}
}

@book{FetterWalecka,
  title     = {Quantum Theory of Many-Particle Systems},
  author    = {Fetter, Alexander L. and Walecka, John Dirk},
  publisher = {McGraw--Hill},
  address   = {New York},
  year      = {1971}
}

@book{GiulianiVignale,
  title     = {Quantum Theory of the Electron Liquid},
  author    = {Giuliani, Gabriele F. and Vignale, Giovanni},
  publisher = {Cambridge University Press},
  address   = {Cambridge, England},
  year      = {2005},
  doi       = {10.1017/CBO9780511619915}
}

@book{Datta1995,
  title     = {Electronic Transport in Mesoscopic Systems},
  author    = {Datta, Supriyo},
  year      = {1995},
  publisher = {Cambridge University Press},
  address   = {Cambridge},
  series    = {Cambridge Studies in Semiconductor Physics and Microelectronic Engineering},
  doi       = {10.1017/CBO9780511805776}
}

@book{Imry2002,
  title     = {Introduction to Mesoscopic Physics},
  author    = {Imry, Yoseph},
  year      = {2002},
  edition   = {2nd},
  publisher = {Oxford University Press},
  address   = {Oxford},
  doi       = {10.1093/oso/9780198507383.002.0002}
}

@book{NazarovBlanter2009,
  title     = {Quantum Transport: Introduction to Nanoscience},
  author    = {Nazarov, Yuli V. and Blanter, Yaroslav M.},
  year      = {2009},
  publisher = {Cambridge University Press},
  address   = {Cambridge},
  doi       = {10.1017/CBO9780511626906}
}

@article{Magazzu2024unified,
  title     = {Unified diagrammatic approach to quantum transport in few-level junctions for bosonic and fermionic reservoirs: Application to the quantum Rabi model},
  author    = {Magazz\`u, L. and Paladino, E. and Grifoni, M.},
  journal   = {Phys. Rev. B},
  volume    = {110},
  issue     = {8},
  pages     = {085419},
  year      = {2024},
  publisher = {American Physical Society},
  doi       = {10.1103/PhysRevB.110.085419},
  url       = {https://link.aps.org/doi/10.1103/PhysRevB.110.085419}
}

@book{Donarini2024transport,
  title     = {Quantum Transport in Interacting Nanojunctions: A Density Matrix Approach},
  author    = {Donarini, Andrea and Grifoni, Milena},
  year      = {2024},
  publisher = {Springer},
  doi       = {10.1007/978-3-031-55619-7},
  isbn      = {978-3-031-55619-7},
  url       = {https://link.springer.com/book/10.1007/978-3-031-55619-7}
}

@article{Magazzu2022influence,
  author  = {Magazzù, Luca and Grifoni, Milena},
  title   = {{Feynman-Vernon} influence functional approach to quantum transport in interacting nanojunctions: An analytical hierarchical study},
  journal = {Phys. Rev. B},
  year    = {2022},
  volume  = {105},
  pages = {125417},
  doi     = {10.1103/physrevb.105.125417}
}

@book{kamenev2011book,
  author    = {Kamenev, Alex},
  publisher = {Cambridge University Press},
  title     = {Field Theory of Non-Equilibrium Systems},
  year      = {2011},
  doi       = {10.1017/CBO9781139003667}
}

@book{stefanucci2013book,
  author    = {Stefanucci, Gianluca and van Leeuwen, Robert},
  publisher = {Cambridge University Press},
  title     = {Nonequilibrium Many-Body Theory of Quantum Systems},
  year      = {2013},
  doi       = {10.1017/CBO9781139023979}
}

@article{cavina2023convenient,
  title   = {A convenient {Keldysh} contour for thermodynamically consistent perturbative and semiclassical expansions},
  author  = {Cavina, Vasco and Kadijani, Sadeq S. and Esposito, Massimiliano and Schmidt, Thomas L.},
  journal = {SciPost Phys.},
  volume  = {15},
  number  = {5},
  pages   = {209},
  year    = {2023},
  doi     = {10.21468/SciPostPhys.15.5.209}
}

@article{funo2018pathint,
  title   = {Path Integral Approach to Quantum Thermodynamics},
  author  = {Funo, Ken and Quan, Haitao T.},
  journal = {Phys. Rev. Lett.},
  number  = {4},
  volume  = {121},
  pages   = {040602},
  year    = {2018},
  doi     = {10.1103/PhysRevLett.121.040602}
}

@article{diosi1998diffusion,
  title   = {Non-{Markovian} quantum state diffusion},
  author  = {Di\'osi, L. and Gisin, N. and Strunz, W. T.},
  journal = {Phys. Rev. A},
  volume  = {58},
  issue   = {3},
  pages   = {1699--1712},
  year    = {1998},
  doi     = {10.1103/PhysRevA.58.1699}
}

@article{stockburger2001diffusion,
  author  = {Stockburger, J\"urgen T. and Grabert, Hermann},
  title   = {Non-{Markovian} quantum state diffusion},
  journal = {Chem. Phys.},
  volume  = {268},
  number  = {1},
  pages   = {249-256},
  year    = {2001},
  doi     = {10.1016/S0301-0104(01)00307-X}
}

@article{ankerhold2026colloquium,
  title   = {Colloquium: Simulating {non-Markovian} dynamics in open quantum systems},
  author  = {Xu, Meng and Vadimov, Vasilii and Stockburger, J. T. and Ankerhold, J.},
  journal = {Rev. Mod. Phys.},
  volume = {98},
  pages   = {021002},
  year    = {2026},
  month   = {Jan},
  publisher = {American Physical Society},
  doi     = {10.1103/w3nw-hbjc}
}

@article{cavina2025unifying,
  title     = {Unifying quantum stochastic methods using {Wick}'s theorem on the {Keldysh} contour},
  author    = {Cavina, Vasco and D’Abbruzzo, Antonio and Giovannetti, Vittorio},
  journal   = {Phys. Rev. Res.},
  volume    = {7},
  number    = {4},
  pages     = {043262},
  year      = {2025},
  publisher = {APS},
  doi       = {10.1103/cnvm-w8cy}
}

@article{stockburger2002exact,
  title   = {Exact $\mathit{c}$-Number Representation of Non-{Markovian} Quantum Dissipation},
  author  = {Stockburger, J\"urgen T. and Grabert, Hermann},
  journal = {Phys. Rev. Lett.},
  volume  = {88},
  issue   = {17},
  pages   = {170407},
  year    = {2002},
  doi     = {10.1103/PhysRevLett.88.170407}
}

@article{tilloy2017unraveling,
  title   = {Time-local unraveling of non-{Markovian} stochastic {Schr\"odinger} equations},
  author  = {Tilloy, Antoine},
  journal = {{Quantum}},
  volume  = {1},
  pages   = {29},
  year    = {2017},
  doi     = {10.22331/q-2017-09-19-29}
}

@article{talkner2016aspects,
  title   = {Aspects of quantum work},
  author  = {Talkner, Peter and H{\"a}nggi, Peter},
  journal = {Phys. Rev. E},
  volume  = {93},
  number  = {2},
  pages   = {022131},
  year    = {2016},
  doi     = {10.1103/PhysRevE.93.022131}
}

@article{Espositoreview,
  title   = {Nonequilibrium fluctuations, fluctuation theorems, and counting statistics in quantum systems},
  author  = {Esposito, Massimiliano and Harbola, Upendra and Mukamel, Shaul},
  journal = {Rev. Mod. Phys.},
  volume  = {81},
  issue   = {4},
  pages   = {1665--1702},
  year    = {2009},
  doi     = {10.1103/RevModPhys.81.1665}
}

@article{d2025exact,
  title   = {Exact Non-{Markovian} Master Equations: A Generalized Derivation for {Gaussian} Systems},
  author  = {D'Abbruzzo, Antonio and Giovannetti, Vittorio and Cavina, Vasco},
  journal = {Phys. Rev. Lett.},
  volume  = {135},
  issue   = {24},
  pages   = {240401},
  year    = {2025},
  doi     = {10.1103/cb7c-5f66}
}

@article{Talarico2019scalable,
  author  = {Talarico, N. W. and Maniscalco, S. and Lo Gullo, N.},
  title   = {A scalable numerical approach to the solution of the {Dyson} equation for the non-equilibrium single-particle {Green’s} function},
  journal = {Phys. Status Solidi B},
  volume  = {256},
  pages   = {1800501},
  year    = {2019},
  doi     = {10.1002/pssb.201800501}
}

@article{lamic2025solving,
  title   = {Solving the transient {Dyson} equation with quasilinear complexity via matrix compression},
  author  = {Lamic, Baptiste},
  journal = {Phys. Rev. B},
  volume  = {112},
  issue   = {23},
  pages   = {235142},
  year    = {2025},
  doi     = {10.1103/q222-w14m}
}

@article{sierra2014strongly,
  title   = {Strongly nonlinear thermovoltage and heat dissipation in interacting quantum dots},
  author  = {Sierra, Miguel A. and S{\'a}nchez, David},
  journal = {Phys. Rev. B},
  volume  = {90},
  pages   = {115313},
  year    = {2014},
  doi     = {10.1103/PhysRevB.90.115313}
}

@article{ingi2017reversal,
  title   = {Reversal of Thermoelectric Current in Tubular Nanowires},
  author  = {Erlingsson, Sigurdur I. and Manolescu, Andrei and Nemnes, George Alexandru and Bardarson, Jens H. and S\'anchez, David},
  journal = {Phys. Rev. Lett.},
  volume  = {119},
  pages   = {036804},
  year    = {2017},
  doi     = {10.1103/PhysRevLett.119.036804}
}

@article{tesser2022heat,
  title   = {Heat rectification through single and coupled quantum dots},
  author  = {Tesser, Ludovico and Bhandari, Bibek and Erdman, Paolo Andrea and Paladino, Elisabetta and Fazio, Rosario and Taddei, Fabio},
  journal = {New J. Phys.},
  volume  = {24},
  number  = {3},
  pages   = {035001},
  year    = {2022},
  doi     = {10.1088/1367-2630/ac53b8}
}

@article{antola2025quantum,
  title        = {{Quantum Bipolar Thermoelectricity}},
  author       = {Antola, Filippo and De Simoni, Giorgio and Giazotto, Francesco and Braggio, Alessandro},
  journal = {npj Quantum Inf.},
  year         = {2026},
  doi = {0.1038/s41534-026-01237-8}
}

@article{arrachea2007heat,
  title   = {Heat production and energy balance in nanoscale engines driven by time-dependent fields},
  author  = {Arrachea, Liliana and Moskalets, Michael and Martin-Moreno, Luis},
  journal = {Phys. Rev. B},
  volume  = {75},
  number  = {24},
  pages   = {245420},
  year    = {2007},
  doi     = {10.1103/PhysRevB.75.245420}
}

@article{kosloff2014quantum,
  title   = {Quantum heat engines and refrigerators: continuous devices},
  author  = {Kosloff, Ronnie and Levy, Amikam},
  journal = {Annu. Rev. Phys. Chem.},
  volume  = {65},
  number  = {1},
  pages   = {365--393},
  year    = {2014},
  doi     = {10.1146/annurev-physchem-040513-103724}
}

@article{erdman2017thermoelectric,
  title   = {Thermoelectric properties of an interacting quantum dot based heat engine},
  author  = {Erdman, Paolo Andrea and Mazza, Francesco and Bosisio, Riccardo and Benenti, Giuliano and Fazio, Rosario and Taddei, Fabio},
  journal = {Phys. Rev. B},
  volume  = {95},
  number  = {24},
  pages   = {245432},
  year    = {2017},
  doi     = {10.1103/PhysRevB.95.245432}
}

@article{benenti2017fundamental,
  title   = {Fundamental aspects of steady-state conversion of heat to work at the nanoscale},
  author  = {Benenti, Giuliano and Casati, Giulio and Saito, Keiji and Whitney, Robert S.},
  journal = {Phys. Rep.},
  volume  = {694},
  pages   = {1--124},
  year    = {2017},
  doi     = {10.1016/j.physrep.2017.05.008}
}

@article{FeiQuan2019,
  title = {Group-theoretical approach to the calculation of quantum work distribution},
  author = {Fei, Zhaoyu and Quan, H. T.},
  journal = {Phys. Rev. Res.},
  volume = {1},
  issue = {3},
  pages = {033175},
  numpages = {8},
  year = {2019},
  month = {Dec},
  publisher = {American Physical Society},
  doi = {10.1103/PhysRevResearch.1.033175},
  url = {https://link.aps.org/doi/10.1103/PhysRevResearch.1.033175}
}

@article{esposito2015efficiency,
  title   = {Efficiency fluctuations in quantum thermoelectric devices},
  author  = {Esposito, Massimiliano and Ochoa, Maicol A. and Galperin, Michael},
  journal = {Phys. Rev. B},
  volume  = {91},
  number  = {11},
  pages   = {115417},
  year    = {2015},
  doi     = {10.1103/PhysRevB.91.115417}
}

@article{dubi2011colloquium,
  title   = {Colloquium: Heat flow and thermoelectricity in atomic and molecular junctions},
  author  = {Dubi, Yonatan and Di Ventra, Massimiliano},
  journal = {Rev. Mod. Phys.},
  volume  = {83},
  number  = {1},
  pages   = {131--155},
  year    = {2011},
  doi     = {10.1103/RevModPhys.83.131}
}

@article{paladino20141,
  title   = {1/f noise: Implications for solid-state quantum information},
  author  = {Paladino, E. and Galperin, Y. M. and Falci, G. and Altshuler, B. L.},
  journal = {Rev. Mod. Phys.},
  volume  = {86},
  number  = {2},
  pages   = {361--418},
  year    = {2014},
  doi     = {10.1103/RevModPhys.86.361}
}

@article{paladino2002decoherence,
  title   = {Decoherence and 1/f noise in {Josephson} qubits},
  author  = {Paladino, E. and Faoro, L. and Falci, G. and Fazio, Rosario},
  journal = {Phys. Rev. Lett.},
  volume  = {88},
  number  = {22},
  pages   = {228304},
  year    = {2002},
  doi     = {10.1103/PhysRevLett.88.228304}
}

@article{zou2024spatially,
  title   = {Spatially correlated classical and quantum noise in driven qubits},
  author  = {Zou, Ji and Bosco, Stefano and Loss, Daniel},
  journal = {npj Quantum Inf.},
  volume  = {10},
  number  = {1},
  pages   = {46},
  year    = {2024},
  doi     = {10.1038/s41534-024-00842-9}
}

@article{gulacsi2023signatures,
  title   = {Signatures of non-{Markovianity} of a superconducting qubit},
  author  = {Gul{\'a}csi, Bal{\'a}zs and Burkard, Guido},
  journal = {Phys. Rev. B},
  volume  = {107},
  number  = {17},
  pages   = {174511},
  year    = {2023},
  doi     = {10.1103/PhysRevB.107.174511}
}

@article{wang2025non,
  title   = {Non-{Markovian} noise mitigation: Practical implementation, error analysis, and the role of spectral properties of the environment},
  author  = {Wang, Ke and Li, Xiantao},
  journal = {Phys. Rev. A},
  volume  = {112},
  number  = {1},
  pages   = {012406},
  year    = {2025},
  doi     = {10.1103/db4q-8ny9}
}

@article{sannia2025non,
  title         = {Non-{Markovianity} and memory enhancement in Quantum Reservoir Computing},
  author        = {Sannia, Antonio and Rodr{\'\i}guez, Ricard Ravell and Giorgi, Gian Luca and Zambrini, Roberta},
  journal = {npj Quantum Inf.},
  volume = {12},
  pages = {111},
  year = {2026},
  doi = {10.1038/s41534-026-01257-4}
}

@article{soret2022thermodynamic,
  title   = {Thermodynamic consistency of quantum master equations},
  author  = {Soret, Ariane and Cavina, Vasco and Esposito, Massimiliano},
  journal = {Phys. Rev. A},
  volume  = {106},
  number  = {6},
  pages   = {062209},
  year    = {2022},
  doi     = {10.1103/PhysRevA.106.062209}
}

@article{landauer1957spatial,
  title   = {Spatial variation of currents and fields due to localized scatterers in metallic conduction},
  author  = {Landauer, Rolf},
  journal = {IBM J. Res. Dev.},
  volume  = {1},
  number  = {3},
  pages   = {223--231},
  year    = {1957},
  doi     = {10.1147/rd.13.0223}
}

@article{buttiker1986four,
  title   = {Four-terminal phase-coherent conductance},
  author  = {B{\"u}ttiker, M.},
  journal = {Phys. Rev. Lett.},
  volume  = {57},
  number  = {14},
  pages   = {1761},
  year    = {1986},
  doi     = {10.1103/PhysRevLett.57.1761}
}

@article{FisherLee1981,
  title   = {Relation between conductivity and transmission matrix},
  author  = {Fisher, Daniel S. and Lee, Patrick A.},
  journal = {Phys. Rev. B},
  volume  = {23},
  pages   = {6851--6854},
  year    = {1981},
  doi     = {10.1103/PhysRevB.23.6851}
}

@article{vacchini2011markovianity,
  title   = {Markovianity and non-{Markovianity} in quantum and classical systems},
  author  = {Vacchini, Bassano and Smirne, Andrea and Laine, Elsi-Mari and Piilo, Jyrki and Breuer, Heinz-Peter},
  journal = {New J. Phys.},
  volume  = {13},
  number  = {9},
  pages   = {093004},
  year    = {2011},
  doi     = {10.1088/1367-2630/13/9/093004},
  url     = {https://doi.org/10.1088/1367-2630/13/9/093004}
}

@article{tamascelli2018nonperturbative,
  title     = {Nonperturbative Treatment of Non-{Markovian} Dynamics of Open Quantum Systems},
  author    = {Tamascelli, Dario and Smirne, Andrea and Huelga, Susana F. and Plenio, Martin B.},
  journal   = {Phys. Rev. Lett.},
  volume    = {120},
  issue     = {3},
  pages     = {030402},
  year      = {2018},
  doi       = {10.1103/PhysRevLett.120.030402},
  url       = {https://journals.aps.org/prl/abstract/10.1103/PhysRevLett.120.030402}
}

@article{Breuer2016Review,
  title   = {Colloquium: Non-{Markovian} dynamics in open quantum systems},
  author  = {Breuer, Heinz-Peter and Laine, Elsi-Mari and Piilo, Jyrki and Vacchini, Bassano},
  journal = {Rev. Mod. Phys.},
  volume  = {88},
  issue   = {2},
  pages   = {021002},
  year    = {2016},
  doi     = {10.1103/RevModPhys.88.021002}
}

@article{Li2020Perspective,
  title   = {Non-{Markovian} quantum dynamics: What is it good for?},
  author  = {Li, C.-F. and Guo, G.-C. and Piilo, J.},
  journal = {EPL},
  volume  = {128},
  number  = {3},
  pages   = {30001},
  year    = {2020},
  doi     = {10.1209/0295-5075/128/30001}
}

@article{luczka1990spin,
  title={Spin in contact with thermostat: Exact reduced dynamics},
  author={{\L}uczka, Jerzy},
  journal={Physica A},
  volume={167},
  number={3},
  pages={919--934},
  year={1990},
  publisher={Elsevier},
  doi={10.1016/0378-4371(90)90299-8}
}

@article{Hu1992Classic,
  title   = {Quantum {Brownian} motion in a general environment: Exact master equation with nonlocal dissipation and colored noise},
  author  = {Hu, B. L. and Paz, Juan Pablo and Zhang, Yuhong},
  journal = {Phys. Rev. D},
  volume  = {45},
  issue   = {8},
  pages   = {2843--2861},
  year    = {1992},
  doi     = {10.1103/PhysRevD.45.2843}
}

@article{Zhang2012ME,
  title   = {General Non-{Markovian} Dynamics of Open Quantum Systems},
  author  = {Zhang, Wei-Min and Lo, Ping-Yuan and Xiong, Heng-Na and Tu, Matisse Wei-Yuan and Nori, Franco},
  journal = {Phys. Rev. Lett.},
  volume  = {109},
  issue   = {17},
  pages   = {170402},
  year    = {2012},
  doi     = {10.1103/PhysRevLett.109.170402}
}

@article{colla2025unveiling,
      title={Unveiling coherent dynamics in {non-Markovian} open quantum systems: exact expression and recursive perturbation expansion}, 
      author={Alessandra Colla and Heinz-Peter Breuer and Giulio Gasbarri},
      journal = {Phys. Rev. A},
      colume =  {112},
      pages = {L050203},
      year={2025},
      doi = {10.1103/n5nl-gn1y}
}

@article{abiuso2019non,
  title   = {Non-{Markov} enhancement of maximum power for quantum thermal machines},
  author  = {Abiuso, Paolo and Giovannetti, Vittorio},
  journal = {Phys. Rev. A},
  volume  = {99},
  number  = {5},
  pages   = {052106},
  year    = {2019},
  doi     = {10.1103/PhysRevA.99.052106}
}

@article{strasberg2016nonequilibrium,
  title   = {Nonequilibrium thermodynamics in the strong coupling and non-{Markovian} regime based on a reaction coordinate mapping},
  author  = {Strasberg, Philipp and Schaller, Gernot and Lambert, Neill and Brandes, Tobias},
  journal = {New J. Phys.},
  volume  = {18},
  number  = {7},
  pages   = {073007},
  year    = {2016},
  doi     = {10.1088/1367-2630/18/7/073007}
}

@article{perarnau2015extractable,
  title   = {Extractable work from correlations},
  author  = {Perarnau-Llobet, Mart{\'\i} and Hovhannisyan, Karen V. and Huber, Marcus and Skrzypczyk, Paul and Brunner, Nicolas and Ac{\'\i}n, Antonio},
  journal = {Phys. Rev. X},
  volume  = {5},
  number  = {4},
  pages   = {041011},
  year    = {2015},
  doi     = {10.1103/PhysRevX.5.041011}
}

@article{talkner2020colloquium,
  title   = {Colloquium: Statistical mechanics and thermodynamics at strong coupling: Quantum and classical},
  author  = {Talkner, Peter and H{\"a}nggi, Peter},
  journal = {Rev. Mod. Phys.},
  volume  = {92},
  number  = {4},
  pages   = {041002},
  year    = {2020},
  doi     = {10.1103/RevModPhys.92.041002}
}

@article{thomas2018thermodynamics,
  title   = {Thermodynamics of non-{Markovian} reservoirs and heat engines},
  author  = {Thomas, George and Siddharth, Nana and Banerjee, Subhashish and Ghosh, Sibasish},
  journal = {Phys. Rev. E},
  volume  = {97},
  number  = {6},
  pages   = {062108},
  year    = {2018},
  doi     = {10.1103/PhysRevE.97.062108}
}

@article{Cao2015thermoelectric,
  author  = {Cao, Zhan and Fang, Tie-Feng and Li, Lin and Luo, Hong-Gang},
  title   = {Thermoelectric-induced unitary {Cooper} pair splitting efficiency},
  journal = {Appl. Phys. Lett.},
  year    = {2015},
  volume  = {107},
  number  = {21},
  pages   = {212601},
  doi     = {10.1063/1.4936380}
}

@article{mayo2025tur,
    title={Thermodynamic Uncertainty Relation in Hybrid Normal-Superconducting Systems: The Role of superconducting coherence}, 
    author={Franco Mayo and Nahual Sobrino and Rosario Fazio and Fabio Taddei and Michele Governale},
    journal ={Phys. Rev. Research},
    volume = {8},
    pages = {013242},
    year={2026},
    doi = {10.1103/tfsv-fjzd}
}

@article{blasi2024exact,

  title = {Exact finite-time correlation functions for multiterminal setups: Connecting theoretical frameworks for quantum transport and thermodynamics},

  author = {Blasi, Gianmichele and Khandelwal, Shishir and Haack, G\'eraldine},

  journal = {Phys. Rev. Res.},

  volume = {6},

  issue = {4},

  pages = {043091},

  numpages = {30},

  year = {2024},

  month = {Nov},

  publisher = {American Physical Society},

  doi = {10.1103/PhysRevResearch.6.043091},

  url = {https://link.aps.org/doi/10.1103/PhysRevResearch.6.043091} 
  
}

@article{bourgeois2026transport,
  title = {Transport Approach to Quantum State Tomography},
  author = {Bourgeois, Jeanne and Blasi, Gianmichele and Haack, G\'eraldine},
  journal = {Phys. Rev. Lett.},
  volume = {136},
  issue = {1},
  pages = {010802},
  numpages = {9},
  year = {2026},
  month = {Jan},
  publisher = {American Physical Society},
  doi = {10.1103/zk56-jn7t},
  url = {https://link.aps.org/doi/10.1103/zk56-jn7t}
}

@article{Aurell_2020,
doi = {10.1088/1751-8121/ab9274},
url = {https://doi.org/10.1088/1751-8121/ab9274},
year = {2020},
month = {jun},
publisher = {IOP Publishing},
volume = {53},
number = {27},
pages = {275303},
author = {Aurell, Erik and Kawai, Ryochi and Goyal, Ketan},
title = {An operator derivation of the {Feynman–Vernon} theory, with applications to the generating function of bath energy changes and to an-harmonic baths},
journal = {J. Phys. A: Math. Theor.},
}
